\documentclass[letterpaper, DIV=13]{scrartcl}

\usepackage{amssymb, amsthm, mathtools,bbm}
\usepackage[square,sort,comma,numbers]{natbib}
\usepackage{enumitem}
\usepackage{verbatim}

\usepackage[none]{hyphenat}
\newtheorem{theorem}{Theorem}
\newtheorem{corollary}[theorem]{Corollary}
\newtheorem{lemma}[theorem]{Lemma}
\newtheorem{definition}[theorem]{Definition}
\newtheorem{proposition}[theorem]{Proposition}

\numberwithin{equation}{section}
\numberwithin{theorem}{section}

\newcommand{\rr}{{\mathbb{R}}}
\newcommand{\cc}{{\mathbb{C}}}
\newcommand{\zz}{{\mathbb{Z}}}
\newcommand{\nn}{{\mathbb{N}}}

\newcommand{\Ee}{{\mathbbm{E}}}

\newcommand{\pp}{{\mathbb{P}}}

\newcommand{\im}{{\operatorname{Im}\,}}

\newcommand{\supp}{{\operatorname{supp}\,}}
\newcommand{\tr}{{\operatorname{Tr}\,}}

\newcommand{\beq}[1]{\begin{equation} \label{#1}}
	\newcommand{\eeq}{\end{equation}}

\newcommand{\arcosh}{\operatorname{arcosh}}
\renewcommand{\epsilon}{\varepsilon}
\newcommand{\vv}{\vert}

\let\Re\undefined

\DeclareMathOperator{\Re}{Re}
\DeclareMathOperator{\Prob}{Prob}
\DeclareMathOperator{\PPP}{PPP}
\def\be{\begin{equation}}
	\def\ee{\end{equation}} 
\DeclareMathOperator{\spec}{spec}
\DeclareMathOperator{\dist}{dist}

\begin{document}
	\addtokomafont{author}{\raggedright}
	
	\title{ \raggedright  Spectral Analysis of the\\ Quantum Random Energy Model}
	
	\author{\hspace{-.075in} Chokri Manai and Simone Warzel}
	\date{\vspace{-.3in}}
	
	\maketitle

	\minisec{Abstract}
	The Quantum Random Energy Model (QREM) is a random matrix of Anderson-type which describes effects of a transversal magnetic field on Derrida's spin glass. The model exhibits a glass phase as well as a classical and a quantum paramagnetic phase. 
    We analyze in detail the low-energy spectrum and establish a localization-delocalization transition for the corresponding eigenvectors of the QREM. Based on a combination of random matrix and operator techniques as well as insights in the random geometry, we derive next-to-leading order asymptotics for the ground-state energy and eigenvectors in all regimes of the parameter space. Based on this, we also deduce the next-to-leading order of the free energy, which turns out to be deterministic and on order one in the system size in all phases of the QREM. As a result, we determine the nature of the fluctuations of the free energy in the spin glass regime. 
	
	\bigskip
	
	\tableofcontents

	\section{Introduction and main results}\label{sec:intro}

\subsection{Quantum random energy model}

In the theory of disordered systems the random energy model (REM) is a simple, yet ubiquitous toy model.  
It assigns to every $ N $-bit or Ising string  $ \pmb{\sigma} = (\sigma_1 , \dots , \sigma_N) \in   \{ -1, 1\}^N  \eqqcolon \mathcal{Q}_N $  a  
rescaled Gaussian random variable 
$$
U(\pmb{\sigma}) := \sqrt{N} \, \omega(\pmb{\sigma}) 
$$
with $ \left( \omega(\pmb{\sigma}) \right) $ forming $ 2^N $ canonically realized independent and identically distributed (i.i.d.) random variables with standard normal law $ \mathbb{P} $.  
The Hamming cube $  \mathcal{Q}_N  $ is rendered a graph by declaring two bit strings connected by an edge if they differ by a single bit flip:  introducing the flip operators 
$ F_j\pmb{\sigma} := (\sigma_1, \dots , - \sigma_j  , \dots , \sigma_N ) $ on components $ j \in \{ 1, \dots , N \} $, the edges of the Hamming cube are formed by all pairs  of the form $ ( \pmb{\sigma} , F_j\pmb{\sigma} )$.  
The graph's negative adjacency matrix 
$$ \left( T\psi\right)(\pmb{\sigma}) := -  \sum_{j=1}^N \psi( F_j \pmb{\sigma} ) 
$$
is defined on $ \psi \in  \ell^2( \mathcal{Q}_N)  $, the $ 2^N $-dimensional Hilbert space of complex-valued functions on $ N $-bit strings. Since every vertex  in $  \mathcal{Q}_N  $ has a constant degree $ N $, the negative graph Laplacian, $ T + N \mathbbm{1} $, just differs by $ N $ times the identity matrix. We study the quantum random energy model (QREM) which is the random matrix 
\begin{equation}\label{eq:Ham}
	H := \Gamma \ T + U ,
\end{equation}
where $ \Gamma \geq 0 $ is a parameter, and $ U $ is diagonal in the canonical configuration basis $ (\delta_{\pmb{\sigma}} ) $ of $ \ell^2(\mathcal{Q}_N) $, i.e.,  $U \delta_{\pmb{\sigma}}  = U(\pmb{\sigma}) \delta_{\pmb{\sigma}}  $ and $ \psi(\pmb{\sigma}) = \langle \delta_{\pmb{\sigma}}  \vert\psi \rangle $. As usual in mathematical physics, we choose the scalar product  $ \langle \cdot \vert \cdot \rangle $ on $  \ell^2( \mathcal{Q}_N) $ to be linear in its second component. 

The QREM is a random matrix of Anderson type -- albeit on a quite unconventional graph whose connectivity grows to infinity with the system size $ N $, and with a scaling of the random potential $ U $ which enforces the operator norm of both, $ T $ and $ U $, to be of the same order $ N $ (cf.~\eqref{eq:specT} and \eqref{eq:2ndspin}). It is thus natural to investigate the localization properties of its eigenfunctions. 
The interest in the QREM is however many-faceted. In mathematical biology,  the model has received attention under the name REM House-of-Cards model \cite{SS+13} as an element of a simplistic probabilistic model of population genetics, in which $ \mathcal{Q}_N $ is the space of gene types and $ U $ encodes their fitness~\cite{Ki78,BBW97,BW01,HRWB02}. In this interpretation, the operator $ T $ implements mutations of the gene type, and one is interested  in the long-time limit of the semigroup generated by $ H $ (cf.~\cite{AGH20}, in which the parameter regime $  \Gamma =  \kappa/N$ with fixed $ \kappa > 0 $ corresponding to the normalized Laplacian is considered). 

The Anderson-perspective has also attracted attention in discussions of many-body or Fock-space localization, where the QREM occasionally serves as an analytically more approach\-able toy to test ideas about more realistic disordered spin systems~\cite{LPS14,BLPS16,BFSTV21,TKPW21}. We will comment on some of the conjectures in the physics literature concerning the localization properties of the eigenfunctions  after presenting our results on this topic.

In statistical mechanics, the QREM was introduced~\cite{Gold90} as a simplified model to investigate the quantum effects caused by a transversal magnetic field on classical mean-field spin-glass models~\cite{BM80,FS86,YI87,RCC89,SIC12}. In this context, 
the Hilbert space $\ell^2( \mathcal{Q}_N )$ is unitarily identified with the tensor-product Hilbert space  $ \otimes_{j=1}^N \mathbbm{C}^2$  of $ N $ spin-$\tfrac{1}{2}$ quantum objects. A corresponding unitary 
maps the canonical basis $ (\delta_{\pmb{\sigma}} ) $ to the tensor-product basis in which the Pauli-$ z $-matrix is diagonal on each tensor component. 
For completeness, we recall the form of the Pauli matrices in this basis:
$$
\sigma^{x} = \left( \begin{matrix} 0 & 1 \\  1 & 0 \end{matrix}  \right) , \quad  \sigma^{y} = \left( \begin{matrix} 0 & - i \\  i & 0 \end{matrix}  \right)  , \quad  \sigma^{z} = \left( \begin{matrix} 1 & 0 \\  0 & -1 \end{matrix}  \right)  . 
$$
The Pauli matrices are naturally lifted to  $ \otimes_{j=1}^N \mathbbm{C}^2$ by their action on the $ j$th tensor component, $ \sigma^{x}_j := \mathbbm{1} \otimes \dots \otimes \ \sigma^{x}  \ \otimes \dots \otimes   \mathbbm{1}  $. 
Upon the above unitary equivalence,  $ T $ corresponds to $- \sum_{j=1}^N  \sigma^{x}_j  $, i.e., a constant field in the negative $ x $-direction exerted on  all $ N $ spin-$\tfrac{1}{2}$ (cf.~\cite{MW20c}). 
In this interpretation, the random potential $ U $ is the energy operator of the  spin-$\tfrac{1}{2}$-objects, which interact disorderly only through their $ z $-components. 
Derrida~\cite{Der80,Der81} originally invented the classical REM $ U $ as a simplification to other mean-field spin glasses such as the Sherrington-Kirkpatrick model. 

The phenomenon common to such classical spin glass models is a glass freezing transition  into a low temperature phase which, due to lack of translation invariance, is described by an order parameter (due to Parisi) more complicated than a global magnetization \cite{Par80,MPV86,Tal11,Pan13}. In the absence of external fields the latter typically vanishes.   These thermodynamic properties are encoded in 
the (normalized) partition function 
$$ Z(\beta, \Gamma) := 2^{-N}\,  \tr e^{-\beta H} $$
at inverse temperature $ \beta \in [0, \infty) $, or, equivalently,  its pressure 
\begin{equation}
	\Phi_N(\beta, \Gamma) :=  \ln Z(\beta, \Gamma) .
\end{equation}
Up to a factor of $ - \beta^{-1} $, the latter coincides with the  free energy.
In the thermodynamic limit $ N \to \infty $, the specific pressure of the REM  converges almost surely~\cite{Der80,Der81,OP84, Bov06},
\begin{equation}\label{eq:REMc}
	N^{-1} \Phi_N(\beta, 0) \to  p^{\mathrm{REM}}(\beta) = \left\{ \begin{array}{l@{\quad}r} \tfrac{1}{2} \beta^2 & \mbox{if } \; \beta \leq \beta_c , \\[1ex]   \tfrac{1}{2} \beta_c^2 + (\beta - \beta_c) \beta_c & \mbox{if } \;   \beta > \beta_c .\end{array} \right.
\end{equation}
It exhibits a freezing transition into a low-temperature phase characterized by the vanishing of the specific entropy above 
$$ \beta_c := \sqrt{2 \ln 2 } . $$
\begin{figure}[ht]
	\vspace*{-.5cm}
	\begin{center}
		\includegraphics[width=.75\textwidth] {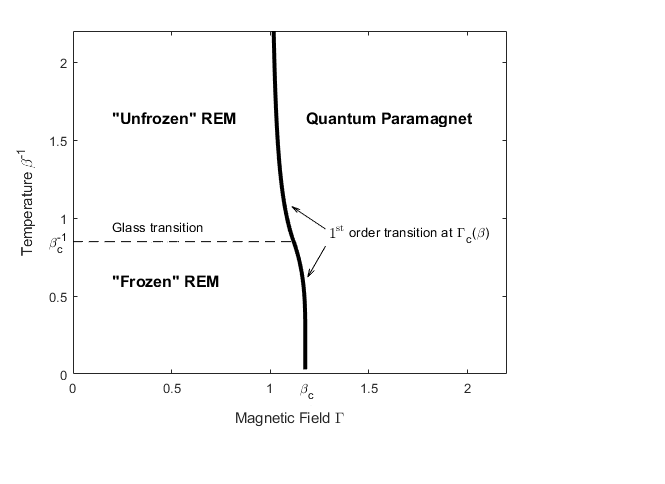}
	\end{center}
	\vspace*{-1cm}
	\caption{Phase diagram of the QREM as a function of the transversal magnetic field $ \Gamma  $ and the temperature $ \beta^{-1}$ \cite{Gold90, MW20}. The first-order transition occurs at fixed $ \beta $ and $ \Gamma_c(\beta) $. The freezing transition is found at temperature $ \beta_c^{-1} $, which is unchanged in the presence of a small magnetic field. } \label{fig:phase}
\end{figure}

In the presence of the transversal field,  the spin-glass phase of the REM disappears for large $ \Gamma > 0 $ and a first-order phase transition into a quantum paramagnetic phase described by
$$ p^{\mathrm{PAR}}(\beta \Gamma)  := \ln \cosh\left(\beta \Gamma\right) $$
occurs at the critical magnetic field strength
$$ \Gamma_c(\beta) := \beta^{-1} \arcosh\left( \exp\left( p^{\mathrm{REM}}(\beta)\right)  \right) . $$ 
In particular, $ \Gamma_c(0) = 1  $ and $ \Gamma_c(\beta_c) = \beta_c^{-1} \arcosh(2) $.
The precise location of this first-order transition and the shape of the phase diagram of the QREM, which we sketch in Figure~\ref{fig:phase}, had been predicted by Goldschmidt~\cite{Gold90} in the 1990s and was rigorously established  in \cite{MW20}. 

\begin{proposition}[\cite{MW20}]\label{thm:old}
	For any $ \Gamma , \beta \geq 0 $ we have the almost sure convergence as $ N \to \infty $:
	$$
	N^{-1} \Phi_N(\beta, \Gamma) \to  \max\{ p^{\mathrm{REM}}(\beta)  , p^{\mathrm{PAR}}(\beta \Gamma) \} . 
	$$
\end{proposition}

\subsection{Low energy states}
Through the low-temperature limit 
$ \beta \to \infty$, Proposition~\ref{thm:old} contains also information on the ground state energy  of the QREM,
\[
N^{-1} \inf\spec H  \to  \begin{cases} - \beta_c & \text{ if } \Gamma < \beta_c, \\
	- \Gamma & \text{ if } \Gamma > \beta_c .
\end{cases}
\]
The critical coupling for this quantum phase transition is at the endpoint  $ \lim_{\beta \to \infty} \Gamma_c(\beta) = \beta_c $ of the first order phase transition. 
As will be demonstrated below, this ground-state phase transition at $ \Gamma = \beta_c $ is manifested by a change of the nature of the corresponding eigenvector from sharply localized to (almost) uniformly delocalized. 
To provide some heuristics, it is useful to compare the ground-state energy and eigenvectors of the two operators entering $ H = \Gamma T + U  $:
\begin{enumerate}
	\item The spectrum of $T$ consists of $ N + 1$ eigenvalues, 
	\be\label{eq:specT}
	\spec T =  \{ 2n - N \, \vert \, n \in \mathbb{N}_0 , \; n \leq N \} ,
	\ee
	with degeneracy given by the binomials $  \binom{N}{n} $. The corresponding $ \ell^2 $-normalized eigenvectors are the natural orthonormal basis for the Hadamard transformation, which diagonalizes $ T $. They are indexed by  subsets $ A \subset \{ 1, \dots , N \} $:
	\be\label{eq:efs}
	\Phi_A(\pmb{\sigma}) := \frac{1}{\sqrt{2^N}} \prod_{j \in A} \sigma_j .
	\ee
	The eigenvalue to $ \Phi_A $ is  $ 2\vert A \vert -N $ with $ \vert A \vert $ the cardinality of the set. 
	In particular, 
	the unique ground-state of $ \Gamma T $ is $ \Phi_\emptyset $  with energy $ - N \Gamma $. All eigenvectors $ \Phi_A $  are maximally uniformly delocalized over the Hamming cube. 	
	\item
	In contrast, all eigenvectors $ \delta_{\pmb{\sigma} }$ of $ U $ are maximally localized.  
	The REM's minimum energy, $ \min U $, is roughly at $ - N \beta_c $.
	For $\eta > 0$ the event that $ \| U \|_{\infty} \coloneqq \max_{\pmb{\sigma}\in \mathcal{Q}_N} \vert U(\pmb{\sigma}) \vert > (\beta_c + \eta) N$ has exponentially small probability, i.e,
	\begin{align}\label{eq:REMevent}
		\Omega_{N,\eta}^\textup{REM}  \coloneqq & \ \{ \| U \|_{\infty} \leq (\beta_c + \eta) N  \}, \\
		& \pp(\Omega_{N,\eta}^\textup{REM}) \geq 1 - 2^{N+1} e^{-\frac12(\eta + \beta_c)^2 N} = 1 - 2 \ e^{-N( \eta \beta_c + \frac{\eta^2}{2}) } , \notag 
	\end{align}
	where the inequality follows from the union bound and a Markov-Chernoff estimate.
	A more precise description of the extremal value statistics of $\min U$ is~\cite{LLR83,Bov06}
	\be\label{eq:Uinfty} 
	\pp\left(  \min  U  \geq   s_N(x) \right) = \left( 1- 2^{-N} e^{-x + o(1)} \right)^{2^N} 
	\ee 
	for any $ x $ in terms of the function $ s_N $ given by
	\be\label{eq:rescaling}
	s_N(x) \coloneqq -\beta_c N+\frac{\ln(N \ln2) + \ln(4\pi)}{2 \beta_c} - \frac{x}{\beta_c} .   \ee
	By symmetry of the distribution, a  similar expression applies to the maximum.
\end{enumerate}

These limiting cases suggest the following heuristic, perturbative description of the ground-state  of $ H = \Gamma T + U $ in the regimes of small and large $ \Gamma $. To our knowledge, it goes back to~\cite{JKrKuMa08}:
\begin{enumerate}
	\item
	For small $ \Gamma $, second-order perturbation theory starting from the vector $  \delta_{{\pmb{\sigma}}_{\min} } $, which is localized at  ${\pmb{\sigma}}_{\min} \coloneqq \arg\min U $, reads:
	\begin{align}\label{eq:2ndspin}
		\inf\spec H  & \approx \min U + \Gamma \  \langle \delta_{{\pmb{\sigma}}_{\min} } \vert T  \delta_{{\pmb{\sigma}}_{\min} } \rangle + \Gamma^2  \sum_{ \pmb{\sigma} \neq {\pmb{\sigma}}_{\min} } \frac{\left\vert  \langle \delta_{{\pmb{\sigma}}} \vert \ T  \delta_{{\pmb{\sigma}}_{\min} } \rangle \right\vert^2  }{U(\pmb{\sigma}_{\min}  )- U(\pmb{\sigma})  } \approx - N \beta_c  - \frac{\Gamma^2}{\beta_c}  . 
	\end{align}
	The first-order term vanishes. The sum in the second-order term is restricted to the neighbors of the minimum, whose potential term typically is only of the order $ \sqrt{N} $. 
	\item
	For large $ \Gamma $, second-order perturbation theory starting from the ground state $  \Phi_\emptyset  $ of $ T $ reads:
	\begin{align}\label{eq:2ndpara}
		\inf\spec H  & \approx - N \Gamma  +  \langle \Phi_\emptyset   \vert  U  \Phi_\emptyset  \rangle - \sum_{ A \neq \emptyset } \frac{\left\vert  \langle \Phi_\emptyset  \vert \ U  \Phi_A\rangle\right\vert^2  }{2 \vert A \vert \ \Gamma } \notag \\ &  \approx - N \Gamma  -  \frac{ \langle \Phi_\emptyset  \vert \ U^2  \Phi_\emptyset \rangle }{N \ \Gamma }  \approx - N \Gamma  -  \frac{1 }{ \Gamma } .
	\end{align}
	The next-to-leading term, $ \langle \Phi_\emptyset   \vert  U  \Phi_\emptyset  \rangle = 2^{-N} \sum_{ \pmb{\sigma} \in \mathcal{Q}_N } U(\pmb{\sigma}) $, vanishes by the law of large numbers. In the order $ \Gamma^{-1} $-term, one uses the approximation that most of the states of $ T $  are found near $ \vert A \vert \approx N/2 $. As will be explained in more detail in Section~\ref{sec:deloc}, one crucial point is that $ U $ is exponentially small when restricted to the eigenspace of $ T $ outside an interval around $ |A| \approx N/2 $. By a decomposition of unity one is therefore left with $  \langle \Phi_\emptyset  \vert \ U^2  \Phi_\emptyset \rangle \approx N $, again by the law of large numbers.
\end{enumerate}
Unlike in a finite-dimensional situation, higher orders in this naive perturbation theory turn out to be of lower order with $ N^{-1} $ the relevant parameter. 
One result of this paper is that these predictions can be confirmed: 
for $ \Gamma < \beta_c $ the ground state is sharply localized near a lowest-energy configuration of the REM. In contrast, 
for $ \Gamma > \beta_c $ the ground state 
resembles the maximally delocalized state given by the ground state of $ T $. In both cases, the ground state is energetically separated and the ground-state gap only closes exponentially near $ \Gamma = \beta_c $, see also~\cite{AW16}. 	
In fact, we do not only restrict attention to the ground state but characterize a macroscopic window of the entire low-energy spectrum in the different parameter regimes. 

Before delving into the details, let us emphasize that the localization-delocalization transition at extreme energies presented here relies on the delocalization properties of $ T $ on $\mathcal{Q}_N $, which fundamentally differ from the finite-dimensional situation. The eigenfunctions of $ T $ can only form localized states from linear combinations in the center of its spectrum. This is given a precise mathematical formulation in the form of novel estimates on the spectral shift and Green function of  Dirichlet  restrictions of $ T $ to Hamming balls in Section~\ref{sec:Adjmatrix}, and random matrix estimates on projections of multiplication operators in Section~\ref{sec:deloc}.  A separation of the localized versus delocalized parts of the spectrum beyond the extremal energies, on which the subsequent results concerning the finite-size corrections of the free energy rest, is facilitated by a novel detailed description of the geometry of extremal fluctuations the REM in Section~\ref{sec:free}.\\

Aside from Theorem~\ref{thm:free}, our results pertain to fixed, but arbitrarily large $ N $ on the product probability space $ \Omega_N $ corresponding to $ 2^N $ i.i.d.\ standard normal random variables whose product measure we denote by $ \mathbb{P} $. We suppress its dependence on $ N $.
In this setting, the results apply to all realizations $ \omega $, aside from exceptional events whose probability will be estimated and goes (exponentially) to zero as $ N \to \infty $. This strong concentration enables the use of Borel-Cantelli arguments in Theorem~\ref{thm:free}, which then apply to the product space $ \prod_{N=1}^\infty \Omega_N $ (which is also the set-up in Proposition~\ref{thm:old}).
To present our results and estimates in a precise, yet reader-friendly, manner, we will make use  of an "indexed" version of Landau's O-notation. 

\begin{definition}
	Let $\Theta = (\theta_1,\ldots, \theta_k)$ be a tuple of parameters, $(a_N)_{N \in \nn}$ a real and $(b_N)_{N \in \nn}$
	a positive sequence. We then write 
	\begin{equation} 
		a_N = \mathcal{O}_{\Theta}(b_N) \quad \text{if} \quad  \limsup_{N \to \infty} \frac{\vert a_N \vert}{b_N} \leq C(\Theta), \end{equation}
	for some positive constant $C(\Theta)$, which may depend on $\Theta$.
	Analogously, we write 
	\begin{equation} 	a_N = o_{\Theta}(b_N) \quad \text{if} \quad    \vert a_N \vert \leq c_N(\Theta) \vert b_N \vert, \end{equation}
	where $c_N(\Theta)$ denotes a null sequence.
\end{definition}
In particular, the appearing constant $C(\Theta)$ or, respectively sequence $c_N(\Theta)$, does not depend on any other parameters in question not included in $\Theta$. That is, if $a_N$ is a random sequence and the realization $ \omega $ of the randomness is not included in the list $
\Theta$ of parameters, the estimates are understood to hold uniformly on the event of interest.

\subsubsection{Paramagnetic regime $\Gamma > \beta_c$ }

Our first main result shows that  in the paramagnetic regime the addition of the REM shifts the eigenvalues~\eqref{eq:specT} of $ T $ at energies below the minimum of $ U $ 
deterministically. 
\begin{theorem}\label{thm:qpgs}
	For $ \Gamma > \beta_c $ and any $ \tau \in( 0,1) $ there are events $ \Omega_{N,\tau}^\textup{par} $ with probability  
	\[ \pp(\Omega_{N,\tau}^\textup{par}) \geq 1 - e^{-N/C} \] and $ C \in (0,\infty) $ a universal constant such that  for all sufficiently large $ N $  and any $\eta > 0$ on $ \Omega_{N,\tau}^\textup{par} \cap \Omega_{N,\eta}^\textup{REM}$ (cf.\ \eqref{eq:REMevent}) all eigenvalues 
	of $ H = \Gamma T  +  U $ below $ -(\beta_c +2\eta) N $  are found in the union of intervals of radius $ \mathcal{O}_{\Gamma,\eta}(N^{\frac{\tau-1}{2}} )$  centered at 
	\begin{equation}\label{eq:center}
		(2n-N)\Gamma +  \frac{N}{(2n-N)\Gamma}  
	\end{equation}
	with $ n \in \{ m \in \mathbb{N}_0  \, \vert (2m-N) \Gamma <  -(\beta_c +2\eta) N \} $.
	Moreover, the ball centered at~\eqref{eq:center} contains exactly $ \binom{N}{n} $ eigenvalues of $ H$ .
\end{theorem}

For the ground-state in the regime $\Gamma > \beta_c$, Theorem~\ref{thm:qpgs} implies that with overwhelming probability
\begin{equation} \label{eq:GSpara}
	\inf\spec H = - \Gamma N - \frac{1}{\Gamma} + o_\Gamma(1). 
\end{equation} 
The energy shift with respect to the ground state of $\Gamma T$ is as predicted by naive second-order perturbation theory~\eqref{eq:2ndpara}. Second-order perturbation theory for the eigenvalues corresponds to first-order perturbation theory for the eigenvectors: the eigenvectors are well approximated  by  their first order corrections. In particular, the ground state in the paramagnetic phase is close to the fully paramagnetic state~$ \Phi_\emptyset $. This is made more precise in our next main result, whose proof alongside that of Theorem~\ref{thm:qpgs} can be found in Section~\ref{sec:deloc}.
\begin{theorem}\label{thm:qpstate}
	In the situation of Theorem~\ref{thm:qpgs} on the event $ \Omega_{N,\tau}^\textup{par} \cap \Omega_{N,\eta}^\textup{REM}   $ with $0 < \eta < (\Gamma - \beta_c)/4$ the $ \ell^2 $-normalized ground state  $\psi$ of $ H = \Gamma T + U $ satisfies:
	\begin{itemize}
		\item[1.] The $\ell^2$-distance of $\psi$ and $  \Phi_\emptyset  $ is
		$ \displaystyle	\| \psi - \Phi_\emptyset\| =  \mathcal{O}_{\Gamma}(N^{\frac{\tau-1}{2}} ) $. 
		\item[2.] The ground state  $\psi$ is exponentially delocalized in the maximum norm, i.e.\ 
		\begin{equation}\label{eq:qpstate}
			\| \psi \|_{\infty}^2 \leq 2^{-N} e^{N\gamma((\beta_c + \eta)/(2\Gamma) )+ o_\Gamma(N) },
		\end{equation}
		where $\gamma \colon [0,1] \to \rr $ denotes the binary entropy 
		\begin{equation}\label{eq:binaryent}
			\gamma(x) \coloneqq - x\ln x - (1-x) \ln (1-x) . 
		\end{equation}
	\end{itemize}
\end{theorem}
The true $ \ell^2 $-distance of the ground-state function to the fully delocalized state $ \Phi_\emptyset  $ is presumably of order $ N^{-\frac{1}{2}}  $ up to a logarithmic correction in $ N $. 
The participation ratio~\eqref{eq:qpstate} is not expected to be  sharp: we conjecture a delocalization bound of the form $ \| \psi \|_{\infty}^2 \leq 2^{-N+o(N)}$.
Section~\ref{sec:deloc}, in which the proofs of Theorems~\ref{thm:qpgs} and~\ref{thm:qpstate} can be found, also contains (non-optimal) $ \ell^\infty $-delocalization estimates for all eigenvalues strictly below the threshold $ - \beta_c N $ in the paramagnetic regime. The optimal decay rates for excited states are not known. 
 In Section~\ref{sec:deloc2} we record a method, which improves the estimate \eqref{eq:qpstate} if $\Gamma - \beta_c$ is small.  A similar, but more elaborate, argument might result in better estimates for all field strengths $\Gamma > \beta_c$.

\subsubsection{Spin glass regime $\Gamma < \beta_c$ }
In the spin glass phase the low-energy configurations of the REM, which occur on the extremal sites
\begin{equation}\label{eq:large}
	\mathcal{L}_{\epsilon} \coloneqq \{\pmb{\sigma} \vert \, U(\pmb{\sigma}) \leq - \epsilon N \} \quad \text{with} \; \epsilon \in (0,\beta_c) ,
\end{equation}
are also shifted by a deterministic, order-one correction by the transverse field as predicted by second-order perturbation theory. 
To characterize localization properties of the corresponding eigenvectors in the canonical  $z$-basis, i.e., the configuration basis $ (\delta_{\pmb{\sigma}} ) $ of $ \ell^2(\mathcal{Q}_N) $, we
let 
\[B_R(\pmb{\sigma}) \coloneqq \{\pmb{\sigma}^\prime  \vert \, d(\pmb{\sigma}, \pmb{\sigma}^\prime ) \leq R \} , \qquad S_R(\pmb{\sigma}) \coloneqq \{\pmb{\sigma}^\prime  \vert \, d(\pmb{\sigma}, \pmb{\sigma}^\prime ) = R \} \]
stand for the Hamming ball and sphere of radius $R$, which are defined in terms of the Hamming distance 
\[ d(\pmb{\sigma}, \pmb{\sigma}^\prime ) \coloneqq \frac{1}{2} \sum_{i=1}^{N} \vert \sigma_i - \sigma^\prime_i \vert \] 
of two configurations $ \pmb{\sigma}, \pmb{\sigma}^\prime \in \mathcal{Q}_N $.

\begin{theorem}\label{thm:sggs}
	For $\Gamma < \beta_c$ and $ \delta > 0 $ small enough there are events $ \Omega_{N,\Gamma,\delta}^\textup{loc} $
	with probability 
	\[ \pp(\Omega_{N,\Gamma,\delta}^\textup{loc}) \geq 1 - e^{-cN}  \] 
	for some $ c = c(\Gamma,\delta) $ such that 
	the following applies for sufficiently large $N$ on $ \Omega_{N,\Gamma,\delta}^\textup{loc} $:
	\begin{enumerate}
		\item
		The eigenvalues $ E $ of $ H = \Gamma T + U $ below $-(\beta_c - \delta) N$ 
		and low-energy configurations $U(\pmb{\sigma})$ are in a one-to-one correspondence such that 
		\begin{equation}\label{eq:sgenergy}
			E = U(\pmb{\sigma})	+ \frac{\Gamma^2 N}{	U(\pmb{\sigma})} + \mathcal{O}_{\Gamma,\delta}(N^{-1/4}). 
		\end{equation}
		In particular,  the estimate $\mathcal{O}_{\Gamma,\delta}(N^{-1/4})$ is independent of $ \pmb{\sigma} \in \mathcal{L}_{\beta_c - \delta}$. 
		\item The $ \ell^2 $-normalized eigenvector $ \psi $ corresponding to $ E $ and $  \pmb{\sigma} $ concentrates near this configuration in the sense that:
		\begin{enumerate}
			\item \emph{Close to extremum:} For any $K \in \nn$ and
			for all  $\pmb{\sigma}^\prime \in S_{K}(\pmb{\sigma})$: 
			\[
			\vert \psi(\pmb{\sigma}^\prime) \vert = \mathcal{O}_{\Gamma,\delta,K}(N^{-K}), \quad \text{and} \quad \sum_{\pmb{\sigma}^\prime \notin B_K(\pmb{\sigma})} \vert \psi(\pmb{\sigma}) \vert^2  = \mathcal{O}_{\Gamma,\delta,K}(N^{-(K+1)}) .
			\]
			\item \emph{Far from extremum:} For any $0 < \alpha <1$, there is $c_\alpha  \in (0,\infty) $ such that 
			\begin{equation}\label{eq:decfar}
				\sum_{\pmb{\sigma}^\prime \notin B_{\alpha N }(\pmb{\sigma})} \vert \psi(\pmb{\sigma}^\prime) \vert^2 \leq  e^{-c_\alpha N} .
			\end{equation}
		\end{enumerate}
	\end{enumerate}
\end{theorem}
This theorem covers states in the extreme localization regime in which the eigenvectors are sharply localized -- each in its own extremal site of the potential.  In this regime, the estimates on the decay rate of the eigenvectors close to the extremum are  optimal and far from the extremum they are optimal up to determining the decay rate $ c_\alpha $.  
Concrete, non-optimized values of the energy threshold $ - N (\beta_c -\delta) $ as well as more precise values of the error terms can be found in the proof of Theorem~\ref{thm:sggs} in Section~\ref{sec:loc}.  In essence, the localization analysis in Section~\ref{sec:loc} proves that resonances and tunneling among different large deviation sites does not play a role in this energy regime. An upper bound for our technique to fail is at $ \delta = \beta_c /2$. The energy threshold at which eigenvectors are believed \cite{BLPS16,BFSTV21} to occupy a positive fraction of $ \mathcal{Q}_N $  is strictly larger than $ - N \beta_c/2  $ and for small fields yet smaller than $- N \Gamma  $. \\

The precise low energy statistics of the REM $ U $  beyond the location of the minimum~\eqref{eq:Uinfty}  is well known. Utilizing the rescaling~\eqref{eq:rescaling} around its minimal value,  the point process 
\be\label{eq:REMconv} 
\sum_{\pmb{\sigma} \in \mathcal{Q}_N} \delta_{s_N^{-1}(U(\pmb{\sigma}))} \to \PPP(e^{-x} \, dx)  \ee
converges weakly to the Poisson point process with intensity measure $e^{-x} \, dx$ on $\rr$ (i.e, when integrating the random measure against a continuous compactly supported function, the resulting random variables converge weakly, see e.g.~\cite[Thm~9.2.2]{Bov06} or~\cite{LLR83}). 
Theorem~\ref{thm:sggs} implies a similar result for the low energy statistics in the QREM.

\begin{corollary}\label{cor:sgstat}
	Let $\Gamma < \beta_c$ and let
	\begin{equation}\label{eq:scal}
		s_N(x;\Gamma) \coloneqq -\beta_c N+\frac{\ln(N \ln2) + \ln(4\pi)}{2 \beta_c} - \frac{\Gamma^2}{\beta_c} - \frac{x}{\beta_c} . 
	\end{equation}
	Then the rescaled eigenvalue process $ \spec H $ of the QREM $ H = \Gamma T + U $ converges  weakly 
	\begin{equation}
		\sum_{E\in   \spec H } \delta_{s_N^{-1}(E;\Gamma)} \to \PPP(e^{-x} \, dx). 
	\end{equation}
	In particular, the ground state energy converges weakly
	\begin{equation}
		\inf\spec H - \left( -\beta_c N + \frac{\ln(N \ln2) + \ln(4\pi)}{2 \beta_c} - \frac{\Gamma^2}{\beta_c} \right) \to -\frac{X}{\beta_c} ,
	\end{equation}
	where $X$ is a random variable distributed according to the law of the maximum of a Poisson point process $\PPP(e^{-x}\, dx )$ with intensity $e^{-x} \, dx$ on the real line. 
\end{corollary}
\begin{proof} Corollary~\ref{cor:sgstat} is a straightforward consequence of Theorem~\ref{thm:sggs} combined with~\eqref{eq:REMconv}. \end{proof} 

Theorem~\ref{thm:sggs} in particular covers the ground-state of the QREM and thus extends the result~\cite[Lemma~2.1]{AGH20} on the leading asymptotics of the ground-state energy in the parameter regime $ \Gamma = \kappa/N $ with $ \kappa > 0 $.  
The proof already contains more information on the $ \ell^2 $-properties of the ground-state eigenvector, which we record next. 
More can be said in on its $ \ell^1 $-localization properties. The latter is of interest in the context of the  interpretation of the QREM in population genetics~\cite{BBW97,BW01,HRWB02}.

\begin{theorem}\label{thm:sgstate}
	For $\Gamma < \beta_c$ there are events $ \hat{\Omega}_{N,\Gamma}^\textup{loc} $
	with probability 
	\[ \pp(\hat{\Omega}_{N,\Gamma}^\textup{loc}) \geq 1 - e^{-cN}  \] 
	for some $ c = c(\Gamma) $ such that on  $ \hat{\Omega}_{N,\Gamma}^\textup{loc} $ for all $N$ large enough there is $ \delta > 0 $ and $\pmb{\sigma}_0 \in \mathcal{L}_{\beta_c - \delta}$ such that the positive $ \ell^2 $-normalized ground state $ \psi $ of the QREM Hamiltonian is concentrated near $ \pmb{\sigma}_0 $ in the sense that: 
	\begin{enumerate}
		\item the  $\ell^2$-distance of $\psi$ and $\delta_{\pmb{\sigma}_0 }$  is $ 
		\| \psi -\delta_{\pmb{\sigma}_0 } \|^2 = \mathcal{O}_{\Gamma}\left(\frac{1}{N}\right) $, 
		and its first order correction
		\begin{equation}
			\xi \coloneqq \sqrt{1-\frac{\Gamma^2}{\beta_c^2 N}} \delta_{\pmb{\sigma}_0} + \frac{\Gamma}{\beta_c N} \sum_{\pmb{\sigma} \in S_1} \delta_{\pmb{\sigma}} 
		\end{equation}
		has the same energy as $\psi$ up to order one, and 
		$	\| \psi - \xi \|^2 = \mathcal{O}_{\Gamma}\left(\frac{1}{N^2}\right) $. 
		\item the $\ell^1$-norm of $\psi$ converges to a bounded constant:
		\begin{equation}\label{eq:ell1}
			\| \psi \|_{1} =	\sum_{\pmb{\sigma}} \psi(\pmb{\sigma}) = \frac{\beta_c}{\beta_c-\Gamma} + o_{\Gamma}(1),
		\end{equation}
		and, for any $1 < p < \infty$: \quad 
		$ \displaystyle	\| \psi \|_{p}^p =	\sum_{\pmb{\sigma}} \vert \psi(\pmb{\sigma}) \vert^p = 1 + o_{\Gamma,p}(1) $.
	\end{enumerate}
	
\end{theorem}

It is natural to assume that the configuration $\pmb{\sigma}_0$ on which the ground-state is asymptotically localized and the classical minimal configuration ${\pmb{\sigma}}_{\min} \coloneqq \arg\min U $ agree. While this is true with high probability, it does not hold almost surely. In the situation of Theorem~\ref{thm:sgstate} one may show that there are two constants $C \geq c > 0$ such that for $N$ large enough:
\begin{equation}\label{eq:hat=0}
	\frac{c}{N} \leq	\pp(\pmb{\sigma}_0 \neq {\pmb{\sigma}}_{\min} ) \leq 	\frac{C}{N}.
\end{equation}
The reason for this is found in the following description of low-energy eigenvalues, 
\[ E_{\pmb{\sigma}} = U(\pmb{\sigma}) - \frac{\Gamma^2}{\beta_c} + \frac{\Gamma^2}{\beta_c^2 N } Z_{\pmb{\sigma}} + \mathcal{O}_\Gamma(N^{-5/4}) , \quad  Z_{\pmb{\sigma}} \coloneqq \frac{1}{N} \sum_{\sigma^\prime\in S_1(\pmb{\sigma})} U(\pmb{\sigma}^\prime) ,
\]
which is proved in Lemma~\ref{lem:rank1} below and which takes into account the next leading term in comparison to~\eqref{eq:sgenergy}. 
The random variables $ Z_{\pmb{\sigma}} $ are standard normal distributed and independent of the large deviations $  U(\pmb{\sigma})  $ with $ \pmb{\sigma} \in \mathcal{L}_{\beta_c-\delta} $ and $ \delta > 0 $ small enough. 
Since the extremal energies form a Poisson process with mean density of order one, the normal fluctuations in the energy-correction of order $ \mathcal{O}(1/N) $ are able to cause the event~\eqref{eq:hat=0}. 
More generally, the method presented in this paper allows for a systematic control of subleading corrections in an expansion of the energy eigenvalues. As we will see, they are determined by potential fluctuations on increasing spheres around the extremal sites.

\subsubsection{Critical case $\Gamma = \beta_c$}
We complete the picture on the ground state by describing the situation in the critical case $\Gamma = \beta_c$, where the quantum phase transition occurs. Adapting techniques, one may also prove that typically one observes a paramagnetic behavior at criticality. 

\begin{proposition}\label{prop:trans}
	Let $\Gamma = \beta_c$. On an event of  probability $1-\mathcal{O}(N^{-1/2})$ the ground state is at $ \inf\spec H = -\Gamma N - \Gamma^{-1} + \mathcal{O}(N^{-1/4}) $ and the eigenvector $ \psi $ is paramagnetic in the sense that $ \|  \psi - \Phi_\emptyset \| =  \mathcal{O}(N^{-1/4}) $. 
	On an event of probability $\mathcal{O}(N^{-1/2})$ the ground state is at  
	$ \inf\spec H = \min U  - \Gamma + \mathcal{O}(N^{-1/4}) $,  and the eigenvector $ \psi $ is  localized in the sense that $ \| \psi - \delta_{\pmb{\sigma}_0} \| =  \mathcal{O}(N^{-1/4}) $.  	
\end{proposition}
The heuristics explanation for this is the following. 
For  $\Gamma = \beta_c$ the ground state energy of $\Gamma T$ is given by $-\beta_c N$, whereas the classical minimal energy is given by $\min U = - \beta_c N +C \ln(N) + \mathcal{O}(1)$ with $ C > 0 $. The logarithmic correction in this expression ensures that the paramagnetic behavior is dominant. 
This argument also suggests that the phase transition should be observed at the $N$-dependent field strength $\Gamma_N$, where the energy predictions of Theorem~\ref{thm:qpgs} and Theorem~\ref{thm:sggs} agree, 
\[  -\Gamma_N N - \frac{1}{\Gamma_N} = \min U	+ \frac{\Gamma_N^2 N}{\min U}  \]
which leads to 
\begin{equation}\label{eq:gammagap}
	\Gamma_N = -\frac{\min U}{N} + \frac{1}{N} \left( \frac{N}{\min U} - \frac{\min U}{N}  \right) + o(N^{-1}). 
\end{equation}
Indeed, in a $o(N^{-1})$ neighborhood of $\Gamma_N$ one can observe a sign of critical behavior, the exponential vanishing gap of the Hamiltonian.

\begin{proposition}\label{prop:gap}
	Let $\Delta_N(\Gamma) > 0$ denote the energy gap of the QREM Hamiltonian. Then, for some $ c > 0$ and $ N$ large enough
	\begin{equation}\label{eq:gap}
		\min_{\Gamma \geq 0} \Delta_N(\Gamma) \leq e^{-cN}
	\end{equation}
	except for a exponentially small event. The minimum is attained at some $\Gamma_N^\star$ satisfying~\eqref{eq:gammagap}.
\end{proposition}

The proof of both Proposition~\ref{prop:trans} and~\ref{prop:gap} are found in the appendix. It relies on a spectral analysis of $H$ and is completely different from the derivation in \cite{AW16}. 

\subsection{Free energy and partition function}\label{sec:freefluc}
The spectral techniques presented here also allow to pin down the pressure $\Phi_N$ and its fluctuations up to order one in $ N $  in all three phases of the QREM: the spin-glass phase as well as the classical ('unfrozen REM') and quantum paramagnetic phase, cf.~Figure~\ref{fig:phase}. 
\begin{theorem}\label{thm:free}
	
	\begin{enumerate}
		\item If $\Gamma > \Gamma_c(\beta)$ the pressure $\Phi_N(\beta,\Gamma)$ is up to order one deterministic and one has the almost sure convergence
		\begin{equation}
			\Phi_N(\beta,\Gamma) - (\ln \cosh(\beta \Gamma)) N \to  \frac{\beta}{\Gamma \tanh(\beta \Gamma)} .
		\end{equation}
		\item If $\Gamma < \Gamma_c(\beta)$ and $\beta\leq\beta_c$, the pressure $\Phi_N(\beta,\Gamma)$ differs from the REM's pressure $  \Phi_N(\beta,0) $ by a deterministic $\beta$-independent shift of order one, i.e., one has the almost sure convergence
		\begin{equation}
			\Phi_N(\beta,\Gamma) - \Phi_N(\beta,0) \to  \Gamma^2 .
		\end{equation}
		\item If $\Gamma < \Gamma_c(\beta)$ and $\beta > \beta_c$, the pressure $\Phi_N(\beta,\Gamma)$ differs from  the REM's pressure by a deterministic $\beta$-dependent shift of  order one, i.e., one has the almost sure convergence
		\begin{equation}\label{eq:spinglassconv}
			\Phi_N(\beta,\Gamma) - \Phi_N(\beta,0) \to  \frac{\Gamma^2 \beta}{\beta_c} .
		\end{equation}
	\end{enumerate}
\end{theorem}
The proof of the almost-sure convergence, for which the probability space is the product $ \prod_{N=1}^\infty \Omega_N$ of independently redrawn variables for every single $ N $, is based on a Borel-Cantelli argument and contained in Section~\ref{sec:free}.

At all values of $ \beta > 0 $, the fluctuations of the REM's pressure $\Phi_N(\beta,0) $ below its deterministic leading term $N  p^{\mathrm{REM}}(\beta)$ have been determined in~\cite{BKL02} (see also~\cite[Thm.~9.2.1]{Bov06}).
Their nature and scale changes from normal fluctuations on the scale $ \exp\left( - \frac{N}{2}(\ln 2 - \beta^2) \right) $ for $ \beta \leq \beta_c/2 $  into a more interesting form of exponentially small fluctuations in the regime $  \beta \in (\beta_c/2, \beta_c) $. In the spin glass phase $ \beta > \beta_c $, the fluctuations are on order one \cite{GMP89} and asymptotically described by Ruelle's partition function of the REM~\cite{Rue87}. More precisely, one has the weak convergence \cite[Thm.~1.6]{BKL02}:
\be\label{eq:Ruelle} 
e^{-N[\beta \beta_c - \ln 2] + \frac{\beta}{2\beta_c}[\ln(N \ln2) + \ln 4 \pi]  } Z_N(\beta,0) \to \int_{-\infty}^{\infty} \, e^{x \beta/\beta_c} \PPP(e^{-x} \, dx). 
\ee
As a consequence of Theorem~\ref{thm:free}, we thus obtain the analogous result for the QREM. 

\begin{corollary}\label{cor:fluctuations} 
	If $\Gamma < \Gamma_c(\beta)$ and $\beta > \beta_c$, we have  the weak convergence:
	$$
	e^{-N[\beta \beta_c - \ln 2] + \frac{\beta}{2\beta_c}[\ln(N \ln2) + \ln 4 \pi] - \frac{\beta \Gamma^2}{\beta_c} } Z_N(\beta,\Gamma) \to \int_{-\infty}^{\infty} \, e^{x \beta/\beta_c} \PPP(e^{-x} \, dx).
	$$
\end{corollary}
\begin{proof}
	By the continuity of the exponential function, this follows immediately from~\eqref{eq:spinglassconv} and~\eqref{eq:Ruelle}. 
\end{proof}

The fluctuations of the QREM's partition function outside the spin glass phase are expected to be much smaller -- for $ \Gamma < \Gamma_c(\beta) $ and $ \beta < \beta_c $ most likely on a similar scale as in the REM and for the paramagnetic regime presumably even smaller. The  methods in this paper do not allow to determine fluctuations on an exponential scale.

\subsection{Comments}
We close this introduction by putting our main results into the broader context of related questions discussed in the physics and mathematics literature.\\
%
In the past years, the QREM has attained interest in the physics community as basic testing ground for quantum annealing algorithms \cite{JKrKuMa08,J+10} and, somewhat related, physicist have started to investigate many-body localization in the QREM~\cite{LPS14,BLPS16,Bur17,FFI19,BFSTV21}. 
Based on numerical computations and non-rigorous methods such as  the forward-scattering approximation and the replica trick, they predict a dynamical phase transition between ergodic and localized behavior in the parameter region $\Gamma < \Gamma_c(\beta), \beta < \beta_c$. This transition is expected to be reflected in a change in the spread of eigenfunctions at the correspond energies, which in the ergodic regime is neither uniform nor localized. It is an interesting mathematical challenge to investigate this. As this requires a good understanding of the eigenfunctions far away from the spectral edges, the methods presented in this paper are not yet sharp enough to tackle those problems. 

In simplified models of Rosenzweig-Porter type such non-ergodic delocalization regimes have been predicted~\cite{KK+15,SK+20} and confirmed by a rigorous analysis~\cite{SW19}. 
In an even more simplified model in which one replaces $ T $ by the orthogonal projection onto its ground-state $ -  \vert \Phi_\emptyset\rangle \langle \Phi_\emptyset \vert $ a fully detailed description of  the localization-delocalization transition has been worked out in~\cite{ASW15}.

Focusing on the physics of spin glasses, the independence of the REM is an oversimpli\-fication. This was the main motivation for Derrida to introduce the Generalized Random Energy Model (GREM) \cite{Der85,DG86}, in which the basic random variables are correlated, but still with a prescribed hierarchical structure.  The free energy of the GREM has been studied extensively \cite{Rue87,CCP87,BK04a,BK04b}. On the quantum side, the specific free energy of the QGREM has been determined in \cite{MW20c}; and in \cite{MW22} the effects of an additional longitudinal field have been considered. 
We expect that our methods can be adapted to the case of a finite-level QGREM to derive analogous results as in Theorems~\ref{thm:qpgs}, \ref{thm:sggs} and \ref{thm:free}.
More precisely, we conjecture that the multiple phase transitions in the QGREM are reflected in the behavior of the ground state wavefunction, i.e., at  the critical field strengths $\Gamma_k$  the wavefunction undergoes a transition from being localized in the block $\pmb{\sigma_k}$ to  a delocalized states in the respective part of the spin components. 
The infinite-level case might require substantially new ideas, as standard interpolation techniques do not reveal
order-one corrections. Our methods, however, are strong enough to cover non-Gaussian REM type models, i.e., i.i.d. a centered square integrable random process, whose distribution satisfies a large deviation principle (see also \cite[Assumption 2.1]{MW20c}). Clearly, explicit expressions in analogous versions of Theorem~\ref{thm:qpgs}, \ref{thm:sggs} and \ref{thm:free} will depend on the distribution of the process as already the parameter $\beta_c$ is specific to  Gaussians.

Among spin glass models with a transversal field, the Quantum Sherrington-Kirkpatrick (QSK) model, in which one substitutes in~\eqref{eq:Ham} for $ U $ the classical SK potential, is of particular interest~\cite{SIC12}. In contrast to the classical SK model, which is solved by Parisi's celebrated formula, such an explicit expression for the free energy of the QSK is lacking, and its analysis remains a physical and mathematical challenge. So far, the universality of the limit of the free energy has been settled in \cite{Craw07}, and in \cite{AB19} the limit of the free energy was expressed as a limit of Parisi-type formulas for high-dimensional vector spin glass models.  Unfortunately, despite the knowledge of a Parisi-type formula, the qualitative features of the phase transition in the QSK could only be analyzed by other means, adapting the methods of \cite{BM80,ALR87}.  In terms of the glass behavior, the analysis in \cite{LRRS19} shows that the glass parameter vanishes uniformly in $\Gamma$ for all $\beta \leq 1$. This is complemented by  \cite{LMRW21}, where the existence of a glass phase has been established for $\beta > 1$ and weak magnetic fields $ \Gamma $.
%

The localization-delocalization transition for the QREM differs drastically from related results on a finite-dimensional graph such as $ \mathbb{Z}^d $  (see e.g.~\cite{AW15,Koe16} and references). Unlike on $ \mathbb{Z}^d $, all low-energy eigenvectors on $ \mathcal{Q}_N $ are delocalized in a regime of large $ \Gamma $ (a regime, which is also absent if one takes $ \Gamma = \kappa/N $ as in~\cite{AGH20}). The localized states appear only for small~$ \Gamma $. Although the norm of the adjacency matrix $ T $ is on the same scale $ N $ as the random potential $ U $, which is not the case for the any of the variety of unbounded distributions studied on subsets of $ \mathbb{Z}^d $, the localization of eigenvectors for extremal energies is even stronger on $ \mathcal{Q}_N $. For the Gaussian distribution studied here, the mass of the eigenvectors sharply concentrates not only for a finite number of eigenvalues in one of the extremal sites of $ U $, but rather for all eigenvalues below a threshold (cp.~\cite{GMS90} with Theorem~\ref{thm:sggs}).  In the finite-dimensional setting, the ground state and the first few excited states concentrate on a small, but growing subdomain of $\zz^d$ and, hence, a finite $\ell^1$-norm for the ground state is specific to the QREM. 
This seemingly contradictory strong localization property compared to  $ \mathbb{Z}^d $ can be traced to the adjacencies matrix's $ T $ bad localization properties to balls, on which we elaborate in Section~\ref{sec:Adjmatrix}: the spectral shift due to localization on a ball of radius $ K $ is order of order $ N $ and not $ K^{-2} $ as on $ \mathbb{Z}^d $. 
This together with the sparseness of the potential's extremal sites does not allow for resonances (cf.~Lemma~\ref{lem:specav}).   In this sense, our proof is in fact somewhat simpler (and hence also stronger) than existing proofs of localization in the extremal sites of a random potential on $ \mathbb{Z}^d $. E.g. most recently and notably, in~\cite{BK16} the statistics of a finite number of eigenvalues above the ground state and the localization properties of their eigenvectors were studied for single-site distributions with doubly exponential tails (see also~\cite{Koe16} for more references). While the degree of localization in the $\Gamma < \beta$ phase is significantly stronger than in the models studied in~\cite{BK16}, we observe a similar exponential decay of the localized states for larger distances and in both cases the extremal statistics is governed by a Poisson process. 
In the study of the parabolic Anderson model, an interesting question is how the shape of the localized eigenstates and the speed of convergence depend on the underlying distribution of the random potential~\cite{Koe16}. For the sake of concreteness, we only study the most prominent case of a Gaussian distribution. Although several quantities such as the constant $\beta_c$ depend crucially on the Gaussian nature, we expect the qualitative aspects of the localization-delocalization transition to be persistent even with other unbounded distributions (e.g. those which meet \cite[Ass~2.1]{MW20c}). 

The operator $T$ coincides up to a diagonal shift $N$ with the Laplacian, i.e., the generator of a simple clock process on  $\mathcal{Q}_N$. This correspondence gives rise to yet another link with the parabolic Anderson model on $\zz^d$. The dynamics of the Anderson model is a vast research topic and its study has revealed many interesting phenomena such as ageing. The spin glass nature is believed to be reflected in non-equilibrium properties and a slow relaxation to equilibrium. However, aging in spin glasses is typically not studied under an unbiased random walk, but rather under the Glauber dynamics for which the transition rates depend on the sites' energies. In the case of the REM,  
 the related Glauber dynamics has drawn considerable interest as a well treatable case for metastability and aging~\cite{ABG03a,ABG03b,CW15,GH19,Gay19}. Our spectral methods might provide some further insights
into the dynamics of REM-type clock processes. 

\section{Adjacency matrix on Hamming balls}\label{sec:Adjmatrix}

This section collects results on the spectral properties of the restriction of $T$ to Hamming balls. We focus on the analysis of the Green's function, which by rank-one perturbation theory, is closely related to the ground state for potentials corresponding to a narrow deep hole - a situation typically encountered in potentials of REM type. Most of the spectral analysis in the literature related to $ T $ is motivated by the theory of error corrections (see e.g.~\cite{CDS95,FriedTill05,BLL18} and references therein). The methods we use are rather different and neither rely on elaborate combinatorics nor a Hadamard transformation, which is applicable on a full Hamming cube only.  

\subsection{Norm estimates} 
In the following, we fix $\pmb{\sigma}_0 \in \mathcal{Q}_N$ and  $0 \leq K \leq N \in \nn$. The restriction $T_K$ of $T$ to the Hamming ball $B_K(\pmb{\sigma}_0 )$ is defined through its matrix elements in the canonical orthonormal basis on $ \ell^2(B_K(\pmb{\sigma}_0 ) )$, which is naturally embedded in $ \ell^2(\mathcal{Q}_N) $:
\begin{equation}\label{eq:tk}
	\langle \delta_{\pmb{\sigma}} \vert T_K \delta_{\pmb{\sigma}^\prime} \rangle = \begin{cases}
		\langle \delta_{\pmb{\sigma}} \vert T\delta_{\pmb{\sigma}^\prime} \rangle & \text{if } \pmb{\sigma}, \pmb{\sigma}^\prime \in B_K(\pmb{\sigma}_0 ) \\
		0 & \text{otherwise.}  
	\end{cases}
\end{equation}

We start with two known results on $T_K$. The first part of the following lemma has been already proved in~\cite{FriedTill05} in case $ K = \varrho N $. The second part is just a special case of the spectral symmetry of any bipartite graph's adjacency matrix (cf.~\cite{CDS95}).
\begin{proposition}[cf.~\cite{FriedTill05}]\label{lem:tknorm} 
	For the restriction $ T_{K} $ to balls $B_{K}(\pmb{\sigma}_0) $ of radius $K \leq N/2$:
	\begin{enumerate}
		\item 
		The operator norm is bounded according to
		\begin{equation}\label{eq:tknorm}
			\|T_K \| \leq 2 \sqrt{K(N-K+1)} ,
		\end{equation}
		and for any radius $ \varrho N $ with $ 0 < \varrho < 1/2 $:
		\begin{equation}\label{eq:tknorm2}
	E_N(\varrho) := \inf\spec T_{\varrho N} =  - \|T_{\varrho N} \| = - 2\sqrt{\varrho(1-\varrho)}N + o_{\varrho}(N) .
		\end{equation}
		\item If $\varphi$ is an eigenvector of $T_K$, then
		$\hat{\varphi}$ given by
		$
		\hat{\varphi}(\pmb{\sigma}) \coloneqq (-1)^{d(\pmb{\sigma},\pmb{\sigma}_0)} \varphi(\pmb{\sigma})
		$
		is also an eigenvector of $T_K$ with
		$	\langle \hat{\varphi} \vert T_K  \hat{\varphi} \rangle = - \langle \varphi \vert T_K  \varphi  \rangle $. 
		Consequently,   the spectrum   is symmetric, $
		\spec(T_K) = - \spec(T_K)$.
	\end{enumerate}
\end{proposition}
	\begin{proof}
		1.~The operator $T_{R} - T_{R-1}$, when naturally defined on the full Hilbert space $ \ell^2(\mathcal{Q}_N) $, describes the hopping between the $R$th and $R-1$th Hamming sphere. Thus, $T_{R} - T_{R-1}$ and $T_{R-2} - T_{R-4}$ act on non overlapping parts of the configuration space.
		This allows us to write
		\begin{equation}\label{eq:dirsum}
			T_K = \left(\bigoplus_{R \leq K, \text{ R even}} T_R - T_{R-1} \right)  + \left(\bigoplus_{R \leq K, \text{ R odd}} T_R - T_{R-1} \right).
		\end{equation} 
		Consequently, it is enough to consider the operators $T_{R} - T_{R-1}$ on $ \ell^2(\mathcal{Q}_N) $. As all matrix elements of $T_{R} - T_{R-1}$ are nonnegative, the Perron-Frobenius Theorem implies that its eigenvector  $\psi_R$ corresponding to the maximal eigenvalue, which coincides with $\|T_{R} - T_{R-1} \|$, is positive. Moreover, $\psi_R$  is radial by symmetry and supported on $S_{R}(\pmb{\sigma}) \cup S_{R+1}(\pmb{\sigma}) $, i.e.,
		$ \psi_R = s_{R+1} \sum_{\pmb{\sigma}^\prime \in S_{R+1}(\pmb{\sigma})} \, \delta_{\pmb{\sigma}^\prime} + s_{R} \sum_{\pmb{\sigma}^\prime \in S_{R}(\pmb{\sigma}) } \, \delta_{\pmb{\sigma}^\prime}  $. 
		By an explicit calculation one thus has
		$ (T_{R} - T_{R-1})^2 \psi_R = R(N-R+1) \psi_R= \|T_{R} - T_{R-1} \|^2 \psi_R $ 
		and, hence using~\eqref{eq:dirsum}:
		\[ \begin{split} \|T_K \| &\leq \max_{R \leq K , \text{ R even}} \|T_R - T_{R-1} \|  +  \max_{R \leq K , \text{ R odd}} \|T_R - T_{R-1} \|  \\ &\leq 2 \max_{R \leq K} \|T_R - T_{R-1} \| = 2  \sqrt{K(N-K+1)}. \end{split}  \]
		A complementing variational bound for a proof of~\eqref{eq:tknorm2} is in~\cite[Appendix~C]{FriedTill05}. 
		
		\noindent 
		2.~The second assertion follows from a direct computation.
	\end{proof}

If $K$ is of order one as a function of $ N $, we have $\|T_K \| = \mathcal{O}_K(\sqrt{N})$. 
This drastic shift of the operator norm due to confinement should be compared to the finite-dimensional situation where this shift for a ball of radius $ K $ is propartional to $ K^{-2} $.

In the remaining part of this section, we will analyze $T_K$ and its Green function in the two extreme cases in relation to $ N $: 
1)~fixed-size balls in Subsection~\ref{subsec:K}, and 2)~growing balls with radius~$K = \varrho N $ with some $ 0 < \varrho < 1/2 $ in Subsection~\ref{subsec:rhoN}.  

\subsection{Green function for balls of  fixed size}\label{subsec:K}

The Green's function of the operator $T_K$ on $ \ell^2(B_K(\pmb{\sigma}_0)) $ is defined by
\be
G_{K}(\pmb{\sigma}, \pmb{\sigma}_0;E) \coloneqq \left\langle \delta_{\pmb{\sigma}} \vert \, (-T_K - E)^{-1}  \delta_{ \pmb{\sigma}_0} \right\rangle .
\ee
Before we derive decay estimates in case $ E \not\in [ -\|T_K \|,\|T_K \|]  $, we recall some general facts:
\begin{enumerate}
	\item By radial symmetry,
	$G_{K}(\pmb{\sigma}, \pmb{\sigma}_0;E)$ only depends on the distance $d(\pmb{\sigma}, \pmb{\sigma}_0)$.
	\item All $ \ell^2$-normalized eigenvectors $(\varphi_j) $ of $T_K$ with eigenvalues $(E_j)$ can chosen to be real, and we have 
	\[ \begin{split} G_{K}(\pmb{\sigma}, \pmb{\sigma}_0;E) &= \sum_{j}\frac{ \varphi_j(\pmb{\sigma}) \varphi_j(\pmb{\sigma}_0) }{E_j - E} = \sum_{j} (-1)^{d(\pmb{\sigma}, \pmb{\sigma}_0)} \frac{ \varphi_j(\pmb{\sigma}) \varphi_j(\pmb{\sigma}_0) }{-E_j - E} \\ & = (-1)^{d(\pmb{\sigma}, \pmb{\sigma}_0)+1}  G_{K}(\pmb{\sigma}, \pmb{\sigma}_0;-E), \end{split} \]
	where the second equality follows from the symmetry of the spectrum stated in Lemma~\ref{lem:tknorm}. Thus, it is sufficient to derive decay estimates for $ E < - \|T_K \|$.  
	\item
	The Green function at $ E < - \|T_K \|$ is related to the ground-state $\varphi$ of the rank-one perturbation 
	\begin{equation}\label{eq:hrk1}
		H^{(E)}	\coloneqq T_K - \alpha^{(E)} \vert \delta_{\pmb{\sigma}_0} \, \rangle \langle \delta_{ \pmb{\sigma} _0} \vert 
	\end{equation}
	on $ \ell^2(B_K(\pmb{\sigma}_0)) $. More precisely, by rank-one perturbation theory $  \alpha^{(E)} := G_{K}(\pmb{\sigma}_0, \pmb{\sigma}_0;E)^{-1} $  is the unique value at which $ 	H^{(E)} $ has a ground-state at  $ E < - \|T_K \|$, and 
	\begin{equation}\label{eq:greenrk1}
		G_{K}(\pmb{\sigma}, \pmb{\sigma}_0;E)	 = \frac{1}{ \alpha^{(E)} } \frac{\varphi(\pmb{\sigma})}{\varphi(\pmb{\sigma}_0)} ,
	\end{equation}
	cf.~\cite[Theorem~5.3]{AW15}. By the Perron-Frobenius theorem, $ \varphi $ and hence the Green function is strictly positive on $ B_K(\pmb{\sigma}_0) $.  A decay estimate  for $ G_{K}(\cdot, \pmb{\sigma}_0;E)$ translates to a bound on the ground state $\varphi $ of $ 	H^{(E)} $ and vice versa. Our proof of the localization results in Section~\ref{sec:loc} will make use of this relation. 
\end{enumerate}

In order to establish decay estimates, we employ the radial symmetry and write the Green function as a telescopic product 
\be \label{eq:prodGamma}
G_{K}(\pmb{\sigma}, \pmb{\sigma}_0;E) \  = \ 
\prod_{d=0}^{\text{dist}(\pmb{\sigma}, \pmb{\sigma}_0)} \Gamma_{K}(d;E)  
\ee  
with factors $  \Gamma_{K}(0;E)  \coloneqq G_{K}(\pmb{\sigma}_0, \pmb{\sigma}_0;E) $ and 
\[
\Gamma_{K}(d;E)  \coloneqq \frac{G_{K}(\pmb{\sigma}, \pmb{\sigma}_0;E)}{G_{K}(\pmb{\sigma}^\prime, \pmb{\sigma}_0;E)}  , \quad \text{if} \,   1\leq d = {\text{dist}(\pmb{\sigma}, \pmb{\sigma}_0)} = {\text{dist}(\pmb{\sigma}^\prime, \pmb{\sigma}_0)} -1 .
\]
The choice of $ \pmb{\sigma} \in S_d(\pmb{\sigma}_0) $ and  $ \pmb{\sigma}^\prime \in S_{d-1}(\pmb{\sigma}_0) $ in the last definition is irrelevant due to the radial symmetry.

The fundamental equation $(-T_{K} -E) G_{K}(\cdot, \pmb{\sigma}_0;E) = \delta_{\cdot,\pmb{\sigma}_0}$ yields for a configuration $\pmb{\sigma}$ with $1 \leq d = \text{dist}(\pmb{\sigma}, \pmb{\sigma}_0) \leq K$
\begin{align*}
	0 &= [(T_K -E) G_{K}(\cdot, \pmb{\sigma}_0;E)]( \pmb{\sigma}) \\ &= - d \prod_{j=0}^{d-1} \Gamma_{K}(j;E) - E  \prod_{j=0}^{d} \Gamma_{K}(j;E) - (N-d)  \prod_{j=0}^{d+1} \Gamma_{K}(j;E) \\
	&  = \left(  \frac{d}{ \Gamma_{K}(d;E) } -E + (N-d)  \Gamma_{K}(d+1;E) \right) \prod_{j=0}^{d} \Gamma_{K}(j;E),
\end{align*}
where we use the convention $  \Gamma_{K}(K+1;E) \coloneqq 0 $. In the case $d= 0$, we have 
$ 1 = (-N \Gamma_{K}(1;E) - E) \Gamma_{K}(0;E) $. 
That  translates 
to the following recursive relation of Riccati type: 
\be\label{eq:mgrinder}
\Gamma_{K}(d;E) \ = \  \mathcal{M}_{d,E}(\, \Gamma_{K}(d+1;E) \, )  , \qquad    0 \le d \le K .
\ee
with the fractional linear transformation   acting on $\cc$:
\be \label{eq:M}
\mathcal{M}_{d,E}(\Gamma) \ := \  \frac{\max\{d,1\}} {-E  - (N-d)\  \Gamma }  .
\ee 
We now analyze the behavior of solutions of the recursive relation in the various regimes of interest.

\begin{proposition}\label{prop:greenfix}
	For any $K \in \nn$ there is some $ C_K < \infty $ such that for any $N > 2K$ and $E < -\|T_K \|$ we have 
	\begin{equation}\label{eq:glocbound}
		G_{K}(\pmb{\sigma}, \pmb{\sigma}_0;E) \leq \frac{C_K}{\left\vert E + \|T_K \| \right\vert} \binom{N}{d(\pmb{\sigma}, \pmb{\sigma}_0)}^{-1/2} \left(\frac{\sqrt{N}}{\vert E\vert}\right)^{d(\pmb{\sigma}, \pmb{\sigma}_0)} .
	\end{equation}
\end{proposition}
%
	\begin{proof}
		In case $E \leq - 2 \sqrt{K N} \leq - \| T_K \| $ (cf.~Lemma~\ref{lem:tknorm}), we use the recursive relations~\eqref{eq:mgrinder} with initial condition $ \Gamma_{K}(K;E) = -\frac{K}{E}  $  to prove that for $1 \leq d \leq K-1$:
		\begin{align*} 	\left( 1+ \frac{(N-d)  \Gamma_{K}(d+1;E) }{E} \right) \geq \frac12 . \end{align*}
		This can be established directly in case $d = K-1$.  For $1 \leq d \leq K-2$ we proceed by induction. Indeed, we have 
		\begin{align*} \left( 1+ \frac{(N-d)  \Gamma_{K}(d+1;E) }{E} \right) & = 1 - \frac{\max\{d+1,1\}} {-E  - (N-d-1) \Gamma_{K}(d+1;E)  } \\ &\geq 1- \frac{2 \max\{d+1,1\}}{E^2} \geq \frac12, 	\end{align*}
		where we used the recursive relation~\eqref{eq:mgrinder}, the induction hypothesis and the upper bound on $E$.  This inductive argument also yields 
		\[ 	\Gamma_{K}(d;E)\leq  \frac{2 d}{ \vert E\vert}  \quad \mbox{for $ E \leq - 2\sqrt{KN}   $ and any $ 1\leq d < K $.} \] 
		
		Utilising the abbreviation $ d \coloneqq \text{dist}(\pmb{\sigma}, \pmb{\sigma}_0) $ and~\eqref{eq:prodGamma} together with the trivial bound $ \Gamma_{K}(0;E) \leq \left\vert E + \|T_K \| \right\vert^{-1} $,  this in turn implies:
		\[ G_{K}(\pmb{\sigma}, \pmb{\sigma}_0;E)  \leq \frac{2^{d} d!}{\vert E\vert^{d}}  \ \Gamma_{K}(0;E) \leq \frac{2^{K} \sqrt{K!}}{\vert E + \|T_K\| \vert}  \binom{N}{d}^{-1/2} \left(\frac{\sqrt{N}}{\vert E\vert}\right)^{d},   \]
		which agrees with the claim in case $E \leq - 2 \sqrt{K N}$. 
		In case $ E \in ( -  2 \sqrt{K N} , - \|T_K \| ) $,  we recall that
		\[ \sum_{\pmb{\sigma}} G_{K}(\pmb{\sigma}, \pmb{\sigma}_0;E)^2 \leq   \left\|(T_K - E)^{-2} \right\| \leq \frac1{\left\vert E + \|T_K \| \right\vert^2} ,   \]
		and the fact that $G_{K}(\cdot, \pmb{\sigma}_0;E)$ is a radially symmetric function. Consequently,
		\[ 	G_{K}(\pmb{\sigma}, \pmb{\sigma}_0;E) \leq \frac{1}{\left\vert E + \|T_K \| \right\vert} \binom{N}{d(\pmb{\sigma}, \pmb{\sigma}_0)}^{-1/2}, \]
		and the claim follows with an appropriate choice for the constant $C_K$. 
	\end{proof}

\subsection{Green function for growing balls}\label{subsec:rhoN}
We now turn to the behavior of the Green's function on balls, which grow with $ N $. This will require a more detailed analysis of the recursion relation~\eqref{eq:mgrinder}. 
To see what to expect, we first derive an estimate on the Green's function of the full Hamming cube.

\begin{lemma}   For any $N\in \mathbb{N} $, $E < - N = - \| T \| $ and $ \pmb{\sigma}, \pmb{\sigma}_0 \in \mathcal{Q}_N $:
	\be \label{eq:G_bound} 
	G_{N}(\pmb{\sigma}, \pmb{\sigma}_0;E)  \ 
	\le  \  \frac{1}{\vert E +  N\vert}\  \  \left(\frac{N}{\vert E\vert}\right)^d  \ {N \choose d(\pmb{\sigma}, \pmb{\sigma}_0)}^{-1}   \, . 
	\ee 
\end{lemma} 
\begin{proof}  
	The Neumann series formula readily implies the operator identity
	\begin{equation}\label{eq:opid} \frac{1}{1-X} = \sum_{k=0}^{d-1} X^k  +  X^d\ \frac{1}{1-X}  \end{equation}
	for any operator with $\| X \| < 1$. Setting $d = d(\pmb{\sigma}, \pmb{\sigma}_0)$, we thus obtain
	\[ 	\left\langle \delta_{\pmb{\sigma}} \big\vert (T  - E)^{-1} \delta_{  \pmb{\sigma}_0} \right\rangle = \frac{-1}{E}		\left\langle \delta_{\pmb{\sigma}} \big\vert (1 - T /E)^{-1} \delta_{  \pmb{\sigma}_0}\right\rangle = \frac{1}{E^d}	\left\langle  \delta_{\pmb{\sigma}} \Big\vert \frac{T^d}{ T-E} \  \delta_{  \pmb{\sigma}_0} \right\rangle , \]
	since terms in \eqref{eq:opid} corresponding to $k < d$ vanish.
	Radial symmetry of the Green function yields 
	\begin{align*}    \label{eq:66}
		\left\langle \delta_{\pmb{\sigma}}  \big\vert (T  - E)^{-1} \delta_{  \pmb{\sigma}_0} \right\rangle    &= {N \choose d}^{-1} \sum_{\pmb{\sigma} \in S_d( \pmb{\sigma}_0)}	\left\langle  \delta_{\pmb{\sigma}} \big\vert (T  - E)^{-1} \delta_{  \pmb{\sigma}_0} \right\rangle  \\  \leq {N \choose d}^{-1}  \frac{\sqrt{2^N}}{E^d} \  \ 	&\Big\langle	\Phi_\emptyset \big\vert T^d\, \frac{1}{T-E} \delta_{  \pmb{\sigma}_0} \Big\rangle \ \notag 
		= \   {N \choose d}^{-1}  \left(\frac{N}{\vert E\vert}\right)^d \   
		\  \frac{1}{\vert E\vert-N} \,,  
	\end{align*} 
	where $	\Phi_\emptyset (\pmb{\sigma}) = 2^{-N/2} $ denotes the lowest energy eigenfunction of $T$, and we applied the eigenfunction equation,   $T 	\Phi_\emptyset  \ = \ - N 	\Phi_\emptyset $, in the last step. 
\end{proof} 

A main difference between the small versus large ball behavior of the Green's function is in the  factor $(\sqrt{N}/\vert E\vert)^d$ in~\eqref{eq:glocbound} versus $(N/\vert E\vert)^d$ in~\eqref{eq:G_bound}. In the case of interest where $\vert E\vert $ is of order $ N$, we arrive at a decay of the order $N^{-d/2}$ versus  $e^{-C d}$. \\

There are at least two strategies to derive upper bounds on the Green function
$ G_{\varrho N}(\pmb{\sigma}, \pmb{\sigma}_0;E)  $
for $E < E_N(\varrho) =  - 2\sqrt{\varrho(1-\varrho)}N + o(N) $ and $ 0<\varrho < 1/2 $, cf.~\eqref{eq:tknorm2}. 
The first strategy is to apply the  arguments, which led to \eqref{eq:G_bound} and which yield
\be \label{eq:Bsigma0} 
\ 
G_{\varrho N}(\pmb{\sigma}, \pmb{\sigma}_0;E) \leq \frac{1}{E_N(\varrho) - E}\  \ \left(\frac{E_N(\varrho)}{E }\right)^d  \       \frac{\Psi_\varrho(\pmb{\sigma}_0)}{\Psi_\varrho(\pmb{\sigma})}   \ {N \choose d}^{-1}   \, , 
\ee 
with $ \Psi_\varrho \in \ell^2(B_{\varrho N}( \pmb{\sigma}_0))$  the $ \ell^2$-normalized, positive eigenfunction corresponding to $ E_N(\varrho) $. It then remains to establish a bound on the ratio $ \Psi_\varrho(\pmb{\sigma}_0)/\Psi_\varrho(\pmb{\sigma})$.  
We, however, will instead proceed by an analysis of the factors
$\Gamma_{\rho N}$ defined in \eqref{eq:prodGamma}.

\begin{proposition} \label{lem:freeB}
	Let $ 0 < \varrho  <1/2$, and $ \varepsilon > 0 $. Then for $E \leq E_N(\varrho) - \epsilon N$, all $\pmb{\sigma}\in B_{\rho N}(\pmb{\sigma}_0)$ and all $N$ large enough:
	\begin{equation}  \label{G_bound}
		G_{\varrho N}(\pmb{\sigma}, \pmb{\sigma}_0;E) \  
		\leq \ \frac{1}{\epsilon N} {N \choose d(\pmb{\sigma}_0,\pmb{\sigma})}^{-1/2} \ \  2^{-\min\{d(\pmb{\sigma}_0,\pmb{\sigma}), \, \rho_0(\varrho) N\}}
	\end{equation}
	where $ 0 < \varrho_0(\varrho) < \varrho $ is the unique solution of the equation $2 \sqrt{\varrho(1-\varrho)} = 3 \sqrt{\varrho_0(1-\varrho_0)}$. Moreover, for any fixed $K \in \nn$ there is some $ C_K < \infty $ such that for all $N$ large enough:
	\begin{enumerate}
		\item  for all $\pmb{\sigma} \in S_K(\pmb{\sigma}_0)$:\quad 
		$ \displaystyle 
		G_{\varrho N}(\pmb{\sigma}, \pmb{\sigma}_0;E) \  
		\leq \frac{1}{\varepsilon N} \frac{C_K}{\sqrt{N^K}}  {N \choose d(\pmb{\sigma}_0,\pmb{\sigma})}^{-1/2}.
		$
		\item
		$ \displaystyle \sum_{\pmb{\sigma} \not\in B_K(\pmb{\sigma}_0)} G_{\varrho N}(\pmb{\sigma}, \pmb{\sigma}_0;E)^2 \leq \frac{C_K}{\varepsilon^2 N^{K+2}} $. 
	\end{enumerate}

\end{proposition} 
\begin{proof}
	It is convenient to separate the combinatorial factor ${N \choose d(\pmb{\sigma}_0,\pmb{\sigma})}^{-1/2}$ and study
	\begin{equation}\label{eq:renorm}
		\hat{G}_{\varrho N}(\pmb{\sigma}, \pmb{\sigma}_0;E) \coloneqq {N \choose d(\pmb{\sigma}_0,\pmb{\sigma})}^{1/2} G_{\varrho N}(\pmb{\sigma}, \pmb{\sigma}_0;E), \quad \hat{G}_{\varrho N}(\pmb{\sigma}, \pmb{\sigma}_0;E) \  = \ 
		\prod_{d=0}^{d(\pmb{\sigma}, \pmb{\sigma}_0)} \hat{\Gamma}_{\varrho N}(d;E) .
	\end{equation}
	By direct inspection of \eqref{eq:renorm} one obtains the relation $\hat{\Gamma}_{\varrho N}(d;E) = \sqrt{\frac{N-d}{d}} \ \Gamma_{\varrho N}(d;E)$ for $ d \geq 1 $, which in turn implies the recursive relation
	\begin{align}
		\hat{\Gamma}_{\varrho N}(d;E) &= \frac{1}{\frac{\vert E\vert}{V(d)} -m(d) \hat{\Gamma}_{\varrho N}(d;E)  } \quad \text{ for } 1 \leq d \leq \varrho N \label{eq:recur} \\
		\text{with }\quad  V(d) &\coloneqq \sqrt{d(N-d)}, \quad m(d) \coloneqq \sqrt{\frac{(d+1)(N-d)}{d(N-d+1)}}  \notag \\
		\quad  \text{and} &\quad \hat{\Gamma}_{\varrho N}(\varrho N + 1;E) = 0, \quad  \hat{\Gamma}_{\varrho N}(0;E) = \Gamma_{\varrho N}(0;E) = 	G_{\varrho N}(\pmb{\sigma}_0, \pmb{\sigma}_0;E) . 
		\notag 
	\end{align}
	We will now analyze the solution of these recursive relations.
	
	We first claim that for all $N$ large enough:
	\begin{equation}\label{eq:leq1}
		\hat{\Gamma}_{\varrho N}(d;E) \leq 1 \text{ for all  } d \in [\varrho_0 N, \varrho N].	
	\end{equation}
	This is proven by induction on $ d $ starting from $ d= \varrho N+1 $, where it trivially holds. 
	For the induction step from $d+1$ to $d$, we recall that $E_N(\varrho) = -2 \sqrt{\varrho(1-\varrho)} +o_\varrho(N)$  
	from~\eqref{eq:tknorm2}. The monotonicity of $ V(d) $ and $ m(d) $ then implies that for all $\varrho_0 N \leq d \leq \varrho N$ and all $N$ large enough:
	\[ \frac{\vert E\vert}{V(d)} \geq 2 + \frac{\epsilon}{2\sqrt{\varrho(1-\varrho)}}, \quad m(d) \leq  m(\varrho_0 N) = \sqrt{\frac{1+ 1/(\varrho_0 N)}{1-1/(\varrho N)} } 
	= 1 + \mathcal{O}_\varrho(N^{-1}) . \]
	Inserting these estimates into the recursion relation~\eqref{eq:recur}, the claimed inequality \eqref{eq:leq1} follows. 
	
	We now control the recursion relation in the regime $ 1 \leq d \leq \varrho_0 N $. To this end, note that the definition of $\varrho_0$ implies that for any $ d \leq \rho N$ and $N$ large enough:
	$ \vert E\vert/ V(d)  \geq 3 + \epsilon/ (2\sqrt{\rho(1-\rho)})$.
	Using $\hat{\Gamma}_{\varrho N}(\varrho_0 N +1;E) \leq 1$ one readily establishes $\hat{\Gamma}_{\varrho N}(d ;E) \leq \frac12$ inductively as long as $m(d) \leq 2$. The monotonicity $ m(d) \leq m(1) = \sqrt{2} \ (1+ \mathcal{O}(N^{-1})) $  implies that this is true for any $d \geq 1$ at sufficiently large $ N $.  
	The proof of the claimed exponential decay \eqref{G_bound} is then completed using the trivial norm bound	\[ \hat{\Gamma}_{\varrho N}(0;E) =  	G_{\varrho N}(\pmb{\sigma}_0, \pmb{\sigma}_0;E) \leq \left\| (T_{\varrho N} - E)^{-1} \right\| = \dist(E,\spec(T_{\varrho N}))^{-1} \leq \frac{1}{\epsilon N}.  \]
	
	Let us finally consider the case of fixed integers $K$. Note that for any $K \geq 1$ we know by the above $\hat{\Gamma}_{\varrho N}(K+1 ;E) \leq 1/2$. The recursion relation \eqref{eq:recur} then yields for any $ 1 \leq d \leq K$
	\[ \hat{\Gamma}_{\varrho N}(d ;E) \leq \frac{d_K}{\sqrt{N}}  \]
	with some constants $d_K = d_K(\rho)$. This completes the proof of the first item. For the second item we organize the summation into sums over spheres of radius greater or equal to $ K+1 $:
	\[ \begin{split} &
		\sum_{\pmb{\sigma} \not\in B_K(\pmb{\sigma}_0)} G_{\varrho N}(\pmb{\sigma}, \pmb{\sigma}_0;E)^2  \\& \quad = \prod_{d=0}^K  \hat{\Gamma}_{\varrho N}(d ;E)^2 \left( \sum_{D=K+1}^{\varrho_oN}  \prod_{d=K+1}^D  \hat{\Gamma}_{\varrho N}(d ;E)^2 +  \sum_{D=\varrho_oN}^{\varrho N}   \prod_{d=K+1}^D  \hat{\Gamma}_{\varrho N}(d ;E)^2 \right).
	\end{split}	\]
	The product in the prefactor is estimated by $ C_K /(\varepsilon^2 N^{K+2}) $ using the first item.
	The second product is dominated by $ 4^{K - D} $ such that the summation over $ D \geq K+1 $ is bounded by a geometric series. The last product is bounded by $  4^{K - \varrho_0 N} $ such that the sum is bounded trivially by this exponential factor times $ \varrho N $. This completes the proof.%
\end{proof}

The decay established in Proposition~\ref{lem:freeB}  for fixed distance $ K $ to the center of the ball agrees in its dependence on $ N $ with the result of Proposition~\ref{prop:greenfix}.
Moreover, the rough decay estimate \eqref{G_bound} is 'qualitatively correct' in the sense that we expect an estimate of the form 
\[ 	G_{\varrho N}(\pmb{\sigma}, \pmb{\sigma}_0;E) \  
\leq \ \frac{1}{\epsilon N} {N \choose d(\pmb{\sigma}_0,\pmb{\sigma})}^{-1/2} \ \  e^{- L(E,\varrho,d(\pmb{\sigma}_0,\pmb{\sigma})) N} \]
with some positive function $L(E,\varrho,d(\pmb{\sigma}_0,\pmb{\sigma}))$. However, it is clear from the proof of Proposition~\ref{lem:freeB} that we did not attempt to derive a sharp bound for $L$ as it requires a more elaborate analysis of the factors $\hat{\Gamma}_{\varrho N}(d ;E)$.

\section{Delocalization regime}\label{sec:deloc}

%
\subsection{Spectral concentration}

The analysis of the low-energy spectrum in the paramagnetic phase is based on the Schur complement method~\cite[Theorem 5.10]{AW15} for which we define the spectral projections for $\epsilon \in (0,1)$
\begin{equation}
	Q_\varepsilon :=  \mathbbm{1}_{(- \varepsilon N, \varepsilon N)}(T) \, \qquad  P_\varepsilon :=  1- Q_\varepsilon  ,
\end{equation}
which separate eigenstates of $T$ with energies at the center of its spectrum from the edges. Here and in the following, $ \mathbbm{1}(\cdot) $ stands for the indicator function. 
A Chernoff bound shows that the dimension of the range of $ P_\varepsilon $ is only an exponential fraction of the total dimension of Hilbert space:
\begin{equation}\label{eq:dimcheb}
	\dim P_\varepsilon =   \sum_{\vert k -   \frac{N}{2}\vert >\frac{\varepsilon N}{2}} \binom{N}{k} \ \leq \ 2^{N+1} \, e^{-\varepsilon^2 N/2 } \, . 
\end{equation}
The exact asymptotics  of 	$\dim P_\varepsilon$ is in fact well-known, $  \ln \dim P_\varepsilon = (\gamma(\frac{1-\varepsilon}{2}) +o(1))N  $,  in terms of the binary entropy  $\gamma $ defined in~\eqref{eq:binaryent}. 

The following spectral concentration bound expresses the exponential smallness of the projection of symmetric random multiplication operators to the above subspace. It will be our main working horse  in the paramagnetic phase. 
\begin{proposition}\label{prop:PUP}
	Let $\epsilon > 0$ and $ W(\pmb{\sigma}) $, $\pmb{\sigma} \in Q_N $, be  independent and identically distributed random variables  such that
	\begin{enumerate}
		\itemsep.3ex
		\item[i.] the mean is zero, $ \mathbb{E}\left[W(\pmb{\sigma})\right] = 0 $,	
		\item[ii.] the variance of $W(\pmb{\sigma})$ is bounded by one, i.e. $ \mathbb{E}\left[W(\pmb{\sigma})^2\right] \leq1 $, and
		\item[iii.] $W$ is bounded , i.e. $ \| W \|_\infty \leq M_N $ with some $ M_N <\infty$, and $M_N^2 N\, \dim P_\varepsilon / 2^N \leq  1$.
	\end{enumerate}
	Then there are (universal) constants $ c, C \in (0,\infty) $ such for any  $ \lambda > 0 $:
	\begin{equation}\label{eq:PUP}
		\mathbb{P}\left( \| P_\varepsilon W P_\varepsilon \| - \mathbb{E}\left[ \| P_\varepsilon W P_\varepsilon \| \right] > \lambda \,  \sqrt{\frac{\dim P_\varepsilon}{2^N}} \right) \leq \ C e^{- c \lambda^2}.
	\end{equation}
	Moreover, we have the following bound: 
	\begin{equation}
		\mathbb{E}\left[ \| P_\varepsilon W P_\varepsilon \|  \right] \ \leq \ C\, \sqrt{N} \,  \sqrt{ \frac{\dim P_\varepsilon}{2^N}} \, . 
	\end{equation}
\end{proposition} 
\begin{proof}
	The first statement follows from Talagrand's concentration inequality \cite{Tal95} (see also~\cite[Thm. 2.1.13]{Tao}) by considering  $ F: \mathbb{R}^{Q_N} \to \mathbb{R} $ given by $ F(W) :=  \| P_\varepsilon W P_\varepsilon \|  $. We need to show that $F$ is Lipschitz continuous and convex. 
	Convexity, i.e., $ F(\alpha W + (1-\alpha) W') \leq \alpha F(W) + (1-\alpha) F(W') $ for all $ \alpha \in [0,1] $, is evident from the triangle inequality.
	To establish the Lipschitz continuity, let $W,W^\prime \in \mathbb{R}^{Q_N}$ and  $ \psi \in P_\varepsilon \ell^2(Q_N) $ with $ \| \psi \| = 1 $ be such that
	$ \| P_\varepsilon (W-W^\prime) P_\varepsilon \| = \langle \psi , (W-W^\prime) \psi \rangle $. Then one has
	\begin{align*}
		\left\vert F(W) - F(W') \right\vert \leq   & \  \langle \psi , (W-W^\prime)  \psi \rangle \notag = \sum_{\pmb{\sigma}} \vert\psi(\pmb{\sigma})\vert^2 (W(\pmb{\sigma})-W^\prime(\pmb{\sigma})) \\
		\leq  \ \| W - W' \|_2 &\| \psi \|_4^2  \ \leq  \| W - W' \|_2  \| \psi \|_\infty 
		\leq  \max_{\pmb{\sigma}} \sqrt{\langle \delta_{\pmb{\sigma}}\vert P_\varepsilon   \delta_{\pmb{\sigma}} \rangle} \; \| W - W' \|_2 \, . 
	\end{align*}
	The first estimate is the triangle inequality. The next two estimates are special cases of H\"older's inequality, in which we also use $ \| \psi \| = 1 $. The last estimate results from the Cauchy-Schwarz inequality applied to $   \| \psi \|_\infty =\max_{\pmb{\sigma}} \vert\langle P_\varepsilon\delta_{\pmb{\sigma}}\vert  \psi \rangle \vert $ and the fact that 
	$ \|P_\varepsilon\delta_{\pmb{\sigma}} \|  =   \sqrt{\vert\langle P_\varepsilon  \delta_{\pmb{\sigma}}\vert  \delta_{\pmb{\sigma}} \rangle \vert } $. 
	Since by symmetry for any $ \pmb{\sigma} \in \mathcal{Q}_N $:
	\begin{equation}\label{eq:sigmaP}
		\langle\delta_{\pmb{\sigma}} \vert P_\varepsilon  \delta_{\pmb{\sigma}} \rangle = \frac{\dim P_\varepsilon}{2^N} \, ,
	\end{equation}
	we conclude that $ F $ is Lipschitz 
	with constant $ 2^{-N/2 } \, \sqrt{\dim P_\varepsilon } $. This finishes the proof of~\eqref{eq:PUP}. 
	
	The second statement is derived from the matrix Bernstein inequality \cite{Oliv10,Tropp15}. For its application, we note that the matrix under consideration is a sum of independent random matrices,
	$$
	P_\varepsilon W P_\varepsilon = \sum_{\pmb{\sigma}} S(\pmb{\sigma}) \, , \qquad \mbox{with $ S(\pmb{\sigma}) := \frac{\dim P_\varepsilon}{2^N}  \ W(\pmb{\sigma}) \  \vert \psi(\pmb{\sigma}) \rangle \langle  \psi(\pmb{\sigma})  \vert$, }
	$$
	where $  \vert \psi(\pmb{\sigma}) \rangle \langle  \psi(\pmb{\sigma})  \vert $ denotes the rank-one projection onto the  vector $  \psi(\pmb{\sigma})  := \sqrt{ \frac{2^N}{\dim P_\varepsilon}} \ P_\varepsilon \delta_{ \pmb{\sigma} } $, which in view of~\eqref{eq:sigmaP} is normalised. 
	By assumption the matrices $ S(\pmb{\sigma}) $ are centred, $ \mathbb{E}\left[ S(\pmb{\sigma})\right] = 0 $, and bounded
	$$
	\left\| S(\pmb{\sigma}) \right\| \leq M_N  \frac{\dim P_\varepsilon}{2^N}  \leq  \sqrt{\frac{\dim P_\varepsilon}{N \, 2^N} } .
	$$
	The mean variance matrix of $ P_\varepsilon W P_\varepsilon $ is 
	$$
	\sum_{\pmb{\sigma}} \mathbb{E}\left[S(\pmb{\sigma})^2\right] = \left(\frac{\dim P_\varepsilon}{2^N}\right)^2 \sum_{\pmb{\sigma}} \mathbb{E}\left[W(\pmb{\sigma})^2\right]  \vert \psi(\pmb{\sigma}) \rangle \langle  \psi(\pmb{\sigma})  \vert \leq  \frac{\dim P_\varepsilon}{2^N} \ P_\varepsilon .
	$$
	The last inequality follows from the assumption, $  \mathbb{E}\left[W(\pmb{\sigma})^2\right] \leq 1 $, as well as the fact that $( \delta_{\pmb{\sigma}} )$ form an orthonormal basis.  
	Consequently, \cite[Thm.~6.6.1]{Tropp15} together with the trivial bound, $  \dim P_\varepsilon \leq 2^{N}$, on the dimension of the matrices implies
	$$
	\mathbb{E}\left[ \left\| P_\varepsilon W P_\varepsilon \right\| \right] \leq \left( \sqrt{2 \ln 2^{N+1}} + \frac{\ln 2^{N+1}}{3\sqrt{N}} \right) \sqrt{\frac{\dim P_\varepsilon}{2^N} },
	$$
	which completes the proof. 
\end{proof}
Alternatively to Talagrand's concentration inequality, the concentration of measure part of the matrix Bernstein inequality~\cite[Thm.~6.6.1]{Tropp15} would also have been sufficient for proving a slightly less sharp upper bound on the upper tail of the large-deviation probability~\eqref{eq:PUP}.

As an application, we state the following straightforward corollary. Its assumptions are tailored to fit in particular the case of the REM.

\begin{corollary}\label{cor:pup}
	Suppose that  $ W(\pmb{\sigma}) $, $ \pmb{\sigma} \in Q_N $ are i.i.d. random variables  which are
	\begin{itemize}
		\item[i.] mean zero with variance  $ w_N:= \mathbb{E}\left[W(\pmb{\sigma})^2\right] \leq N$ and obey a moment bound $  \mathbb{E}\left[W(\pmb{\sigma})^{8}\right] \leq c \, N^4 $ for some $ c < \infty $. 
		\item[ii.] linearly bounded in the sense that there is some $c < \infty $ such that $  \| W \|_\infty \ \leq \ c\,  N  $.
	\end{itemize}
	Then, there is some $ C \in (0,\infty) $ such that for any $ \tau \in (0,1) $ there are events $ \Omega_{N,\tau} $ with 
	\be\label{eq:defOmega} \pp( \Omega_{N,\tau}) \geq 1 - e^{-N/C} 
	\ee
	such that for all sufficiently large $ N $ and  at $ \varepsilon = {N}^{\frac{\tau-1}{2}} $:
	\begin{align}
		&  \left\| P_\varepsilon W P_\varepsilon \right\| \ \leq  \, C \, N \, e^{-N^{\tau} /4} \, , \label{eq:cor1} \\
		&  \left\| P_\varepsilon (W^2- w_N) P_\varepsilon \right\| \ \leq \  C \, N^{\frac{3}{2}}  \, e^{-N^{\tau} /4} \, , \label{eq:cor2} \\
		& \left\| P_\varepsilon W^p  P_\varepsilon \right\| \ \leq \  C N^{\frac{p}{2}} \quad \mbox{for all $ p \in [1,4] $.} \label{eq:cor3} 
	\end{align}
\end{corollary}
\begin{proof}
	The proof of these inequalities follows by three applications of Proposition~\ref{prop:PUP} with different $  W^\prime $ always at the same $ \lambda = \sqrt{N} $. We note that our choice   $  \varepsilon = {N}^{\frac{\tau-1}{2}}$ implies by \eqref{eq:dimcheb}  $ \dim P_\varepsilon \leq 2^{N+1} e^{-N^\tau /2} $. This in turn yields for any polynomial $M_N$ and $N$ large enough $M_N^2 N \dim P_\varepsilon / 2^N \leq 1$, which indeed checks one of the assumptions of Proposition~\ref{prop:PUP}.
	We then construct three events  $ \Omega_{N,\tau}^{(j)} $ with $ j \in \{ 1,2,3 \} $ each with probability $ \pp(  \Omega_{N,\tau}^{(j)}) \geq 1-  3^{-1} e^{-N/C} $ with some (universal) $ C < \infty $ and all $ N $ large enough. 
	Their intersection $  \Omega_{N,\tau} \coloneqq  \Omega_{N,\tau}^{(1)} \cap  \Omega_{N,\tau}^{(2)} \cap  \Omega_{N,\tau}^{(3)}$ then defines the required events. 
	
	More specifically, for a proof of~\eqref{eq:cor1}, we take $ W^\prime(\pmb{\sigma}) = W(\pmb{\sigma}) /\sqrt{N} $. The event $ \Omega_{N,\tau}^{(1)} $ on which~\eqref{eq:cor1} then satisfies the required probability estimate.  \\
	The proof of~\eqref{eq:cor2} follows again from Proposition~\ref{prop:PUP} with $ W^\prime(\pmb{\sigma}) = c^{-1/4} \, (W(\pmb{\sigma})^2-w_N) /N $ and the prefactor ensuring $ \mathbb{E}\left[W^\prime(\pmb{\sigma})^2\right] \leq 1 $. In this way, we construct $  \Omega_{N,\tau}^{(2)} $.
	
	By Jensen's inequality $ \langle \psi , W^p \psi \rangle^{4/p} \leq \langle \psi , W^4 \psi \rangle $  for any $ p \in [1,4] $, it suffices to establish~\eqref{eq:cor3} for $ p = 4 $. We choose  $ W'(\pmb{\sigma}) = c^{-1/2}  \, (W(\pmb{\sigma})^4-\mathbb{E}\left[W(\pmb{\sigma})^4\right]) /N^{2} $ to define $   \Omega_{N,\tau}^{(3)} $. 
\end{proof}

\subsection{Proof of  Theorem~\ref{thm:qpgs}}
We now use the estimates of the preceding subsection in our  Schur's complement  analysis for the proofs of Theorem~\ref{thm:qpgs} and \ref{thm:qpstate}. These results will actually follow from a slightly more general theorem 
on operators $ H = \Gamma T + W $ of QREM-type.
As a preparation and motivation of the following lemma, we collect some basic facts about these operators. 
The kinetic part of the block component $ Q_\varepsilon H Q_\varepsilon = \Gamma T Q_\varepsilon +  \,  Q_\varepsilon W Q_\varepsilon $ is   
estimated by
\begin{equation}\label{eq:fixKin}
	\| T Q_\varepsilon \| \ \leq \  \varepsilon N \, ,
\end{equation}
which  implies
\begin{equation}\label{eq:specab}
	- \|W\|_\infty - \Gamma \varepsilon \, N  \leq \inf \spec Q_\varepsilon HQ_\varepsilon  \, . 
\end{equation}
For any $  z \in \mathbb{C} $ with $ \Re z < \|W\|_\infty - \Gamma \varepsilon \, N $, the operator $  Q_\varepsilon H Q_\varepsilon - z $ is hence invertible on $ Q_\varepsilon \ell^2(Q_N) $ with inverse denoted by $ R_\varepsilon(z) := (Q_\varepsilon H Q_\varepsilon - z  Q_\varepsilon)^{-1}$. The latter features in Schur's complement formula for the resolvent of $ H $ projected onto the subspace $ P_\varepsilon \ell^2(Q_N) $:
\begin{equation}\label{eq:Schur} 
	P_\varepsilon (H-z)^{-1} P_\varepsilon = \left( P_\varepsilon (H-z) P_\varepsilon  - P_\varepsilon W Q_\varepsilon R_\varepsilon(z) Q_\varepsilon W P_\varepsilon  \right)^{-1} . 
\end{equation} 
Our main observation is that Schur's complement is approximated by an operator proportional to the identity. 
\begin{lemma}\label{lem:SchurRest}
	Consider the operator $ H := \Gamma T + W $ on $ \ell^2(\mathcal{Q}_N)$ with $ W $ satisfying the assumptions in Corollary~\ref{cor:pup} and let $ \Omega_{N,\tau} $ with $ \tau \in (0,1) $ be the events constructed there. Then on $ \Omega_{N,\tau} $ and at $ \varepsilon = {N}^{\frac{\tau-1}{2}} $ for all $ N $ large enough:
	\begin{equation}\label{eq:Schukontr}
		\left\| P_\varepsilon W R_\varepsilon(z)  W P_\varepsilon + P_\varepsilon \frac{w_N}{z}  \right\| \ \leq \ \max\{1,\Gamma\} \frac{C}{d^2} \ N^{\frac{\tau -1}{2}}  \, , \quad R_\varepsilon(z)  := (Q_\varepsilon H Q_\varepsilon - z  Q_\varepsilon)^{-1} ,
	\end{equation}
	for all $ z \in \mathbb{C} $ such that $ \min\{ \vert z\vert \, , \, \dist(\spec Q_\varepsilon H Q_\varepsilon , z ) \} \geq d \,  N $ with $ d \in (0,1] $.
\end{lemma}
\begin{proof}
	We use the resolvent equation to write
	\begin{align}\label{eq:Schurres}
		P_\varepsilon & \left( W R_\epsilon(z)  W +\frac{w_N}{z} \right)  P_\varepsilon = \frac1{z} P_\varepsilon \left(w_N -W Q_\varepsilon W + W  R_\epsilon(z) Q_\varepsilon H Q_\varepsilon W  \right)  P_\varepsilon\notag \\ & = \frac1z P_\varepsilon (w_N - W Q_\varepsilon W ) P_\varepsilon + \frac1z P_\varepsilon \left( W R_\varepsilon(z)  Q_\varepsilon  HQ_\varepsilon)  W \right) P_\varepsilon  ,
	\end{align}
	and estimate both terms in the second line separately. For the first expression we rewrite
	\be\label{eq:umschr}
	P_\varepsilon (w_N - W Q_\varepsilon W ) P_\varepsilon = \ P_\varepsilon (w_N - W^2 ) P_\varepsilon + P_\varepsilon W P_\varepsilon W P_\varepsilon \, . 
	\ee
	According to~\eqref{eq:cor1} and~\eqref{eq:cor2}, the norm of the two terms in the right side is negligible in comparison to $ N^{\frac{\tau - 1}{2}}  $ for all $ N $ large enough. 
	It hence remains to estimate the norm of the second term in the right side of~\eqref{eq:Schurres}. To do so, we split  the terms as follows
	\[ \frac1z P_\varepsilon  W R_\varepsilon(z)  Q_\varepsilon  HQ_\varepsilon  W  P_\varepsilon =  \frac1z P_\varepsilon  W R_\varepsilon(z)  Q_\varepsilon  \Gamma T Q_\varepsilon  W  P_\varepsilon + \frac1z P_\varepsilon  W R_\varepsilon(z)  Q_\varepsilon  W Q_\varepsilon W P_\varepsilon \] and use~\eqref{eq:fixKin}
	together with $ \| R_\varepsilon(z) \| \leq ( d N)^{-1} $ (since $\dist(\spec Q_\varepsilon H Q_\varepsilon , z ) \geq d N$  )  and $ \| P_\varepsilon W \|^2 = \| P_\varepsilon W^2 P_\varepsilon \| \leq C N $ by~\eqref{eq:cor2}.  On $ \Omega_{N,\tau} $  for all $ N $ large enough, we thus conclude:  	\be
	\vert z\vert^{-1} \left\| P_\varepsilon W R_\varepsilon(z)   \Gamma T Q_\varepsilon  W P_\varepsilon \right\| \ \leq \frac{C}{d^2  N}
	\| T Q_\varepsilon \| \leq \frac{C}{d^2 }  \, N^{\frac{\tau - 1}{2}} \, . 
	\ee
	Similarly, we estimate
	\begin{align}
		\vert z\vert^{-1} \left\| P_\varepsilon W R_\varepsilon(z)  Q_\varepsilon W   Q_\varepsilon  W P_\varepsilon \right\| \ & \leq \  \vert z\vert^{-1} \left\| P_\varepsilon W \right\| \, \| R_\varepsilon(z)  \| \, \left\| W  Q_\varepsilon  W P_\varepsilon \right\| \notag \\
		& \leq \ \frac{C}{d^2  N^{3/2}} \,  \sqrt{ \left\| P_\varepsilon W   Q_\varepsilon W^2  Q_\varepsilon   W P_\varepsilon \right\| }\, . 
	\end{align}
	In order to estimate the norm in the right side with the help of~\eqref{eq:cor3}, we rewrite 
	\be
	P_\varepsilon W   Q_\varepsilon W^2  Q_\varepsilon   W P_\varepsilon = P_\varepsilon W^4  P_\varepsilon - P_\varepsilon W^3 P_\varepsilon W P_\varepsilon - P_\varepsilon W P_\varepsilon W^3 P_\varepsilon + P_\varepsilon W  P_\varepsilon   W^2 P_\varepsilon W P_\varepsilon \, . 
	\ee
	On $ \Omega_{N,\tau} $ the norm of this operator is bounded by $ C \, N^2 $ for all $ N $ large enough by~\eqref{eq:cor3}. This concludes the proof.
\end{proof}

These preparations enable us to proof the following general result.%
\begin{theorem}\label{lem:qpgs}
	Consider the operator $ H = \Gamma T + W $ on $ \ell^2(\mathcal{Q}_N)$ with $ W $ satisfying the assumptions in Corollary~\ref{cor:pup} and let $ \Omega_{N,\tau} $ with $ \tau \in (0,1) $ arbitrary be the events constructed there. 
	Then on $ \Omega_{N,\tau}  $ and for all $ N $ large enough  the eigenvalues 
	of $ H$ below $ - \| W \|_\infty - \eta N $ with $ \eta > 0 $ are found in the union of intervals of radius $ \mathcal{O}_{\Gamma,\eta}(N^{\frac{\tau-1}{2}} )$  centered at 
	\begin{equation}\label{eq:centera}
		(2n-N)\Gamma +  \frac{w_N}{(2n-N)\Gamma}  
	\end{equation}
	with $ n \in \{ m \in \mathbb{N}_0  \, \vert  (2m-N) \Gamma <  - \| W \|_\infty -\eta N \} $.
	Moreover, the ball centered at~\eqref{eq:centera} contains exactly $ \binom{N}{n} $ eigenvalues of $ H$ if $ \Gamma > \eta +  \| W \|_\infty/N   $.
\end{theorem}
\begin{proof}
	We write $ H $ using the block decomposition of $ \ell^2(\mathcal{Q}_N) $ induced by $ P_{\varepsilon} $ and employ the Schur complement method. Since the $ Q_\varepsilon $ block is lower bounded according to \eqref{eq:specab},
	all eigenvalues $ E $ of $ H $ strictly below $ - \| W \|_\infty - \Gamma \varepsilon N $ can be read from the 
	equation
	\begin{align}\label{eq:schurneu}
		0 \in   \spec\left( T_\varepsilon(E) \right) 
		\quad \text{with}\quad&  T_\varepsilon(E) \coloneqq P_\varepsilon \Big(\Gamma T +  \frac{N}{E} \Big)  -E   + Y_\varepsilon(E) , \\
		& Y_\varepsilon(E) \coloneqq  \, P_\varepsilon W P_\varepsilon -  \left( P_\varepsilon \frac{N}{E}+  P_\varepsilon  W R_\varepsilon(E) W P_\varepsilon  \right) . \notag 
	\end{align}
	Lemma~\ref{lem:SchurRest} combined with~\eqref{eq:specab} and~\eqref{eq:cor1} implies that for any $ \eta > 0 $ at  $ \varepsilon = N^{(\tau -1)/2} $ and on the event  $ \Omega_{N,\tau} $ in Corollary~\ref{cor:pup}
	\be
	\sup_{E < -\| W \|_\infty - \eta N} \left\| Y_\varepsilon(E) \right\| \ \leq \ C \max\{1,\Gamma\} \ \eta^{-2} \ N^\frac{\tau-1}{2} \, ,
	\ee 
	for all $ N $ large enough. As a consequence  of standard perturbation theory \cite[Corollary 3.2.6]{Bha97} and using the explicit values~\eqref{eq:specT} of the spectrum of $ T $, within this energy region the solution of~\eqref{eq:schurneu} are found within the union of intervals of radius at most $ C \max\{1,\Gamma\}  \eta^{-2} N^{(\tau -1)/2}  $ from the solutions to the equation
	\[ (2n-N) \Gamma + \frac{w_N}{z} - z = 0   \]
	with integers $2n < N(\Gamma-\| W \|_\infty - \eta)/ \Gamma $.
	This leads to 
	\[
	z = \frac{2n-N}{2} \Gamma - \sqrt{\tfrac14 (2n-N)^2 \Gamma^2 + w_N} = (2n-N) \Gamma + \frac{w_N}{(2n-N) \Gamma} + \mathcal{O}_\Gamma\left(N^{-1}\right), 
	\]
	which completes the proof of~\eqref{eq:centera}. 
	The assertion concerning the range of the spectral projections on the small intervals around the above points follows from the monotonicity of $ T_\varepsilon(E)  $ and the fact that the eigenvalue $ 2n-N $ of $ T $ has multiplicity $ \binom{N}{n}  $. 
\end{proof}
Theorem~\ref{thm:qpgs} now immediately follows. 
\begin{proof}[Proof of Theorem~\ref{thm:qpgs}] 
	On $\Omega_{N,\eta/2}^\textup{REM}$ the REM's extremal values are bounded by $ \| U \|_\infty \leq N (\beta_c + \eta)  $. Moreover, $ \mathbb{E}\left[ U(\pmb{\sigma})^2\right] = N $ and $ \mathbb{E}\left[ U(\pmb{\sigma})^8\right] = 105 \ N^4 $, so that $ U $ satisfies all requirements on $ W $ in Corollary~\ref{cor:pup}. The claim is thus  a straightforward consequence of Theorem~\ref{lem:qpgs} with  $ W = U $. 
\end{proof}

\subsection{Proof of  Theorem~\ref{thm:qpstate}}\label{sec:deloc1}
The proof of our second main result, Theorem~\ref{thm:qpstate}, is based on  delocalization properties of the eigenprojection of $T$, which will be derived using the semigroup properties of $T$. More generally, let $B \subset \mathcal{Q}_N$ be any subset of the Hamming cube and $T(B)$ the corresponding restriction, i.e, the operator with matrix elements
$ \langle \delta_{\pmb{\sigma}} \, \vert \, T(B) \,  \delta_{\pmb{\sigma}^\prime} \rangle \coloneqq - \mathbbm{1}_{d(\pmb{\sigma},\pmb{\sigma}^\prime)=1} \mathbbm{1}_B(\pmb{\sigma}) \mathbbm{1}_B(\pmb{\sigma}^\prime) $. 
For Hamming balls $B_K(\pmb{\sigma}_0)$ the operator $T(B_K(\pmb{\sigma}_0))$ was studied in Section 2 and abbreviated there by $T_K$. 
	As all matrix elements are zero or negative, a stochastic representation is at hand: for any $\pmb{\sigma}, \pmb{\sigma}^\prime \in B$ and $\beta\geq 0$ there is a measure $\mu_{B}$ on the space of c\`{a}dl\`{a}g-paths 
	$\Omega(\pmb{\sigma},\pmb{\sigma}^\prime)$ on the hypercube with $\omega(0) = \pmb{\sigma}$ and 
	$\omega(1) = \pmb{\sigma}^\prime$ such that for any $V:B \to \mathbb{R} $ we have 
	$
	\langle \delta_{\pmb{\sigma}}  \, \vert \, e^{-\beta(T(B)+ V)}  \delta_{\pmb{\sigma}^\prime} \rangle = \int_{\Omega(\pmb{\sigma},\pmb{\sigma}^\prime)} e^{- \beta \int_{0}^{1} \, V(\omega(s)) \, ds } \mu_{B}[d \omega]
	$. 
	Such a representation can be derived via the Suzuki-Trotter-formula  (see e.g.~\cite[Appendix B]{LRRS19}). If $B = \mathcal{Q}_N$ the path measure is described in terms of independent Poisson jump processes, whereas for general $B$ one has to take into account that the process is not allowed to leave the set $B$. Here, we do not use the stochastic representation but the related positivity of  matrix elements for any $ \beta \geq 0 $:
	\begin{equation}\label{eq:semi}
		0\leq   \langle \delta_{\pmb{\sigma}}  \, \vert \, e^{-\beta(T(B)+ V)}  \delta_{\pmb{\sigma}^\prime} \rangle  \leq e^{- \min \, \beta V}  \langle \delta_{\pmb{\sigma}}  \, \vert \, e^{-\beta T(B)}  \delta_{\pmb{\sigma}^\prime} \rangle .
	\end{equation}
Since $-T(B)$ and $-T$ have nonnegative matrix elements and $ \langle \delta_{\pmb{\sigma}} \, \vert \, (-T(B)) \, \delta_{\pmb{\sigma}^\prime} \rangle \leq \langle \delta_{\pmb{\sigma}} \, \vert \, (-T) \, \delta_{\pmb{\sigma}^\prime} \rangle $ for any $\pmb{\sigma}, \pmb{\sigma}^\prime$, we also conclude
\begin{equation}\label{eq:semi2}
	\langle \delta_{\pmb{\sigma}}  \, \vert \, e^{-\beta T(B)}  \delta_{\pmb{\sigma}^\prime} \rangle \leq  \langle \delta_{\pmb{\sigma}}  \, \vert \, e^{-\beta T}  \delta_{\pmb{\sigma}^\prime} \rangle = (\cosh \beta)^N (\tanh \beta)^{d(\pmb{\sigma}, \pmb{\sigma}^\prime)} ,
\end{equation} 
where the last equality is by explicit calculation using the Hadamard transformation, i.e., the representation  of $ T $ in terms of Pauli matrices. 
\begin{proposition}\label{prop:ext}  
	Let $B \subset \mathcal{Q}_N$ and $V: B \to \mathbb{R} $ a potential with $V \geq - v N$ for some $0 \leq v < 1$.
	Then the eigenprojection $ P_E := \mathbbm{1}_{(-\infty, E)}(T(B)+V) $ onto eigenvalues 
	$E \in [ - N (1+v) , - v N ] $ satisfies:
	\be  \label{eq:tobefixed}
	\max_{\pmb{\sigma}}  
	\langle \delta_{\pmb{\sigma}}  \, \vert \, P_E \, \delta_{\pmb{\sigma}} \rangle \leq 2^{-N} \exp\left(N \gamma\left(\frac{1+\nu(E))}{2}\right)\right) 
	\ee 
	with the binary entropy $\gamma$ from~\eqref{eq:binaryent} and $\nu(E) := \frac{E}{N} +v $. Moreover, for all normalised states $\psi \in \ell^2(B) $:
	\be\label{eq:delocef}
	\left\| P_E \psi \right\|_\infty^2 \ \leq 2^{-N}  \exp\left(N \gamma\left(\frac{1+\nu(E))}{2}\right)\right) .
	\ee	
\end{proposition}

\begin{proof}
	The spectral theorem combined with an exponential Markov inequality implies for any $ \beta \geq 0 $:
	$$ \langle \delta_{\pmb{\sigma}}  \, \vert \, \mathbbm{1}_{(-\infty, E)}(T(B)+V) \, \delta_{\pmb{\sigma}} \rangle  \leq e^{\beta E}
	\langle \delta_{\pmb{\sigma}} \, \vert \, e^{-\beta(T(B)+V)} \, \delta_{\pmb{\sigma}} \rangle \leq e^{\beta \nu(E) N} (\cosh \beta)^N  .
	$$ 
	The last inequality is a combination of~\eqref{eq:semi} and~\eqref{eq:semi2}. 
		It remains to minimize the function $ f(\beta) \coloneqq \beta \nu(E) + \ln \cosh \beta $ on $ [0,\infty) $. 
		The minimum is attained at
		$	\beta^{\star} = \operatorname{artanh}(-\nu(E)) $. 
		To further simplify the result, we recall the elementary identities $
		\operatorname{artanh}(x) = \frac12 \ln \frac{1+x}{1-x} $ and 
		$ \cosh(\operatorname{artanh}(x)) = \frac{1}{\sqrt{1-x^2}} $ for $ x \in (-1,1) $, 
		which after some algebra lead to 
		$ f(\beta^{\star}) = - \ln 2 + \gamma((1+\nu(E))/2) $ and hence~\eqref{eq:tobefixed}. 
		The second assertion~\eqref{eq:delocef} is a direct consequence of \eqref{eq:tobefixed}.
\end{proof}

We are now ready to complete the proofs of the main results in the paramagnetic regime. 
\begin{proof}[Proof of Theorem~\ref{thm:qpstate}]
	We pick $ \tau \in (0,1) $ and $ 0 < \eta < (\Gamma - \beta_c)/4$ arbitrary and restrict our attention to the event $ \Omega_{N,\tau}^\textup{per} \cap \Omega_{N,\eta}^\textup{REM}$ on which the assertions of Corollary~\ref{cor:pup} for $ W = U $ and Theorem~\ref{thm:qpgs} are valid. 
	
	For a proof of the first assertion, we apply Schur's complement formula to the ground state $\psi = \psi_1 + \psi_2 $ of $ H = \Gamma T + U $. We split $ \psi $ into $ \psi_1 \in P_\varepsilon \ell^2(\mathcal{Q}_N) $ and $ \psi_2 \in Q_\varepsilon \ell^2(\mathcal{Q}_N) $ such that:
	\begin{align*}
		\left(P_\varepsilon  H  P_\varepsilon - E - 	P_\varepsilon  H  R_\varepsilon(E) H P_\varepsilon  \right) \psi_1 = 0& \notag \\
		\psi_2 = - R_\varepsilon(E) Q_\varepsilon H P_\varepsilon \psi_1, &
	\end{align*}
	where  $ E = \inf\spec H = - \Gamma N - \frac{1}{\Gamma} +  \mathcal{O}_\Gamma(N^{\frac{\tau-1}{2}} )$ is the ground-state energy according to Theorem~\ref{thm:qpgs} since $ \Gamma N - \|U \|_{\infty} > \frac12 (\Gamma - \beta_c) N > \eta N$ on $\Omega_{N,\eta}^\textup{REM}$ by the choice for $\eta$.
	Sticking to the notation~\eqref{eq:schurneu}, from the proof of Theorem~\ref{thm:qpgs} we conclude that the first equation can be rewritten in terms of
	\[ P_\varepsilon  H  P_\varepsilon - E - 	P_\varepsilon  H  R_\varepsilon(E) H P_\varepsilon= 	P_\varepsilon \Gamma T 	P_\varepsilon  +  (N {E}^{-1}  - E) P_\varepsilon  	+Y_\varepsilon(E) ,
	\]
	with
	$\| Y_\varepsilon(E)  \| \leq  \mathcal{O}_\Gamma(N^{\frac{\tau-1}{2}} )$. Since $T$ has an energy gap $ 2$ above its unique ground state $ \Phi_\emptyset $ (cf.~\eqref{eq:specT}), we thus conclude 
	\[ \| (\mathbbm{1} - \vert  \Phi_\emptyset  \rangle \langle  \Phi_\emptyset \vert ) \psi_1   \| \leq  \mathcal{O}_\Gamma\left(N^{\frac{\tau-1}{2}} \right).   \]
	To further estimate the norm of
	$ \psi_2 = - R_\varepsilon(E) Q_\varepsilon U  \psi_1 $, 
	we recall that $\|R_\varepsilon(E) \| \leq \frac{C_\Gamma}{N}$ 
	and $\| U \psi_1 \|^2 \leq \| P_\varepsilon U^2P_\varepsilon\| \leq \mathcal{O}(N) $ by Corollary~\ref{cor:pup}. Hence, $\| \psi_2  \|^2 \leq \mathcal{O}_\Gamma\left(\frac1N \right)$. We thus arrive at 
	\begin{equation}
		\| \psi -\Phi_\emptyset \|^2 = \mathcal{O}_\Gamma\left(N^{\tau-1}  \right).
	\end{equation}
	For the second part, we recall the bound \eqref{eq:Uinfty}, and write $H = \Gamma (T + U/\Gamma)$. The claim now follows directly from Proposition~\ref{prop:ext}.
\end{proof}

\subsection{Improved delocalization estimates}\label{sec:deloc2}
The delocalization estimates on low-energy eigenfunctions established in Proposition~\ref{prop:ext}  are not optimal.
In particular, they become trivial in case the minimum of the potential $ v $ is close to one and $ E $ is close to $ - N $. 
The latter corresponds to the critical case addressed in Proposition~\ref{prop:trans}.  In the following, we record an improved delocalization estimate, which involves the QREM with truncated potential on $ \ell^2(\mathcal{Q}_N) $:
	\[  H_{\delta} \coloneqq  \Gamma T  + U_\delta , \qquad U_\delta \coloneqq U  \mathbbm{1}_{\vert U\vert \leq  (\beta_c - \delta)N} , \]
	with truncation parameter $ \delta > 0 $, which will be allowed to be arbitrarily small. 
In that scenario we have the following non-optimal estimate:
\begin{proposition}\label{prop:deloc}
	There are $c,C > 0$ such that for any $ \delta \in (0,\Gamma/25) $ there is some $ \eta > 0 $ and a sequence of events $ \Omega_{N,\delta} $ with $ \pp(  \Omega_{N,\delta} ) \geq 1 - e^{- C N } $  on which for all $ N $ large enough at for all $ \Gamma \geq \beta_c - \delta/4 $:
	\begin{equation}\label{eq:deloccrit} \max_{\pmb{\sigma} \in \mathcal{Q}_N} \langle \delta_{\pmb{\sigma}} \,\vert \, \mathbbm{1}_{(-\infty, -(\Gamma - \eta) N)}(H_{\delta}) \, \delta_{\pmb{\sigma}} \rangle \leq e^{-c N}, \end{equation}
	Moreover, on $ \Omega_{N,\delta}  $ all eigenvalues of $ H_\delta $ below $  E < - (\Gamma - \delta/4) N \leq - (\beta_c - \delta/2) N$ are still described by \eqref{eq:center} with an error $\mathcal{O}_{\Gamma,\delta}(N^{-\frac{1}{4}} ) $.
\end{proposition}
The proof of Proposition~\ref{prop:deloc}  is spelled at the end of this subsection.

	As a preparation, we need the following probabilistic control on the frequency of large deviation sites in a ball of radius $ \alpha N $.
	\begin{lemma}\label{lem:probest}
		Let $\epsilon  , \alpha > 0 $ be such that $ \epsilon^2 > 2 \gamma(\alpha) $.  Then there is $K= K(\epsilon  , \alpha)  \in \nn$ and  $ c = c(\epsilon  , \alpha) > 0 $ such that the sequence of events
		\[ \Omega_{N,\epsilon,\alpha} \coloneqq\{\forall \pmb{\sigma}\in \mathcal{Q}_N : \; \vert B_{\alpha N}(\pmb{\sigma}_0) \cap \mathcal{L}_{\epsilon}\vert < K \} \]
		has a probability bounded by $ \pp( \Omega_{N,\epsilon,\alpha} ) \geq  1 - e^{-c N } $ for all sufficiently large $ N $.
	\end{lemma}
	\begin{proof}
		For any fixed $\pmb{\sigma}\in \mathcal{Q}_N$ we estimate using the independence of the basic random variables and the standard Gaussian tail bound
		\[ \pp(\vert B_{\alpha N}(\pmb{\sigma}) \cap \mathcal{L}_{\epsilon}\vert \geq K) \leq \binom{\vert B_{\alpha N}(\pmb{\sigma})\vert}{K} e^{-K \epsilon^2 N/2} \leq \exp\left(-\frac{K N}{2} \left(\epsilon^2 - 2 \gamma(\alpha) - o_\alpha(1)\right)   \right) .  \]
		The last estimate inserted the asymptotics $\ln  \vert B_{\alpha N}\vert  = N (\gamma(\alpha) +o_\alpha(1) ) $ of the size of a Hamming ball in terms of the binary entropy. 
		The union bound implies 
		$ \pp(\exists \, \pmb{\sigma} \, : \, \vert B_{\alpha_0 N}(\pmb{\sigma}) \cap \mathcal{L}_{\epsilon}\vert \geq K) \leq 2^{N} \pp(\vert B_{\alpha N}(\pmb{\sigma}) \cap \mathcal{L}_{\epsilon}\vert \geq K)    $,
		which is exponentially small for all $K$ large enough.
	\end{proof}
	
Proposition~\ref{prop:deloc} is based on a random-walk expansion for eigenvectors, which facilitates control of spherical means away from their maxima. 
Random walk techniques have been used in the theory of localization under the name 'locator' of 'Feenberg' expansions.
We do not aim to extract optimal information from this technique, which would require much more work. Rather we highlight the usefulness dealing with the critical case.  We will see another more refined use of spherical averaging techniques in Section~\ref{subsec:locL1}. 

As another preparation, we collect some basic properties of a simple random walk $ (Z_k)_{k \in \nn} $ on $\mathcal{Q}_N $, which starts at $Z_0 \coloneqq \pmb{\sigma}_0\in \mathcal{Q}_N  $ and chooses in each unit time step one of the $ N $ neighboring vertices with equal probability $ 1/N $. 
	Let \[ p_K(\pmb{\sigma}, \pmb{\sigma}_0  ) \coloneqq \Prob( Z_K = \pmb{\sigma} \vert Z_0  =  \pmb{\sigma}_0  ) \] stand for the probability to arrive at $ \pmb{\sigma} $ after $K $ steps. 
	Moreover, for a subset $ W  \subset \mathcal{Q}_N $ let  
	\[ M_{K} (W) \coloneqq \sum_{k =0}^{K} \mathbbm{1}[Z_k \in W] \] 
	be the number of visits of the random walk in $ W $ up to time $ K $.
	We will only need the following crude bounds, whose proofs we spell here for the reader's convenience. For further results on random walks on the Hamming cube, see e.g.~\cite{LP+17} and references therein.
	\begin{lemma}
		For all $ \pmb{\sigma}, \pmb{\sigma}_0 \in \mathcal{Q}_N $, all $ \alpha \in (0,3/16) $, and all $ N $ large enough:
		\begin{equation}\label{eq:pwalk}
			p_{\alpha N}(\pmb{\sigma}, \pmb{\sigma}_0 ) \leq \max\{e^{-\gamma(\alpha/8)N + o(N)}, e^{-\alpha N/8}  \}.
		\end{equation}
		For any finite subset $ W  \subset \mathcal{Q}_N $, all $ t \in (0,\alpha] $ and all $ N $ large enough:
		\be\label{eq:sejourntime}
		\Prob\left( M_{\alpha N} (W)  \geq t N \right) \leq  \exp\left( - tN \ln\left(\frac{t N}{\alpha \vert W\vert e} \right) \right) . 
		\ee
	\end{lemma}
	\begin{proof}
		If $d = d(\pmb{\sigma}, \pmb{\sigma}_0 ) \geq \alpha N/8$, we have by spherical symmetry and the asymptotics of the binomial coefficient
		\[ p_{\alpha N}(\pmb{\sigma}, \pmb{\sigma}_0) = \binom{N}{d}^{-1} \sum_{\pmb{\sigma}^\prime \in S_d(\pmb{\sigma}_0)}p_{\alpha N}(\pmb{\sigma}^\prime, \pmb{\sigma}_0  )  \leq  \binom{N}{d}^{-1}  \leq e^{-\gamma(\alpha/8)N +o(N)}. \]
		To complete the proof~\eqref{eq:pwalk}, we discuss the case $d(\pmb{\sigma}, \pmb{\sigma}_0 ) < \alpha N/8$. A simple random walk, which at step $ k $ is at distance $ d =   d(Z_{k},\pmb{\sigma}_0) ) \in [1,\alpha N] $ to the starting point, has $ N-d $ possibilities to move further away and only $ d \leq \alpha N $  to decrease the distance to $ \pmb{\sigma}_0 $ by one. Hence,  $ \Prob( d(Z_{k+1},\pmb{\sigma}_0) < d(Z_{k},\pmb{\sigma}_0)  ) \leq \alpha $  for any $0\leq k \leq \alpha N $. 
		However, to end after $ \alpha N $ steps at some $\pmb{\sigma} \in B_{\alpha N/8}(\pmb{\sigma}_0) $, the walk $ (Z_k)_{k \in \nn}$ has at least $\frac38 \alpha N$ steps, where it gets back closer to the center. Since the random variables $Y_k = \mathbbm{1}[ d(Z_{k+1},\pmb{\sigma}_0) < d(Z_{k},\pmb{\sigma}_0)   ] $ are distributed as conditionally independent Bernoulli variables with success probability at most $\alpha$, we thus arrive at
		\[  p_{\alpha N}(\pmb{\sigma}, \pmb{\sigma}_0) \leq \Prob(\left(\sum_{k=1}^{\alpha N} Y_k \geq 3 \alpha N/8 \right) \leq e^{- \alpha N/8}, \] 
		by a standard Chernoff-bound for Bernoulli variables. This establishes~\eqref{eq:pwalk}.
		
		For a proof of~\eqref{eq:sejourntime} we note that as $\vert W \cap S_1(\pmb{\sigma}) \vert \leq \vert W\vert $ for any $\pmb{\sigma} \in B_{\alpha N}(\pmb{\sigma}_0)$
		\[ \Prob(Z_{k+1} \in W \, \vert \, Z_k = \pmb{\sigma} ) \leq \frac{\vert W\vert}{N}. \]
		The claim thus follows again by a standard Chernoff bound, since $  M_{\alpha N} (W) $ is a sum of conditionally independent Bernoulli-type random variables  with success probability smaller than $ \vert W\vert/N$. 
	\end{proof}
	
	\begin{proof}[Proof of Proposition~\ref{prop:deloc}]
		One easily sees that $U_{\delta}$ meets the requirements of Theorem~\ref{lem:qpgs} with variance 
		$w_N = \Ee[U_\delta(\pmb{\sigma})^2] = N(1- \mathcal{O}( e^{-\frac12(\beta_c - \delta)^2 N}))$. Consequently, for any $ \tau \in (0,1) $ there is some sequence $  \Omega_{N,\tau,\delta} $ with $  \pp(  \Omega_{N,\tau,\delta} ) \geq 1 - e^{- C N}  $, and we arrive at the description of eigenvalues in~\eqref{eq:center} for all $ E < -(\beta_c - \delta/2) N$. 
		This also implies that for $ \eta \in (0,\delta/2) $ and all $ N $ large enough:
		\be\label{eq:entropynumber}
		\tr  \mathbbm{1}_{(-\infty, -(\Gamma - \eta) N)}(H_{\delta}) \leq  \tr \mathbbm{1}_{(-\infty, -(\Gamma - \eta) N)}(\Gamma T) = \sum_{n=0}^{N\eta/(2\Gamma)} \binom{N}{n}  \leq e^{N \gamma(\eta/(2\Gamma)) +o_\Gamma(N)} , 
		\ee
		again by the known asymptotics of the binomial coefficients. 
		
		For a proof of the exponential delocalization estimate~\eqref{eq:deloccrit}, we assume
		throughout the validity of $  \Omega_{N,\tau,\delta} $ with some $ \tau \in (0,1) $ and the event $  \Omega_{N,\epsilon,\alpha_0} $ of Lemma~\ref{lem:probest} with $ \epsilon \coloneqq \Gamma/100 $ and some fixed $ \alpha_0 > 0 $ small enough such that $ 2 \gamma(\alpha_0) < \epsilon^2 $. The intersection of these two events still has a probability, which is exponentially bounded from below independent of $ \delta $ as required.

		Let $\psi$ be an $ \ell^2 $-normalized eigenfunction of $\Gamma T +  U_{\delta}$ with eigenvalue $E < -(\beta_c - \delta/2)N$ and suppose that $\pmb{\sigma}_\psi \in \mathcal{Q}_N $ is a configuration, where $\psi$ takes its maximum absolute value. 
		To make the main idea transparent, we proceed in two steps.\\
		
		\noindent
		\emph{Step 1:}~~We first assume that $B_{\alpha_0 N}(\pmb{\sigma}_\psi) \cap \mathcal{L}_{\epsilon} = \emptyset$. 
		The eigenvalue equation for $ \psi $ at any $\pmb{\sigma} \in B_{\alpha_0 N}(\pmb{\sigma}_\psi)$ together with the bound $U \geq - \epsilon N$ implies 
		\be\label{eq:sperical}  \begin{split} \vert\psi(\pmb{\sigma})\vert & \leq \frac{\Gamma} { \vert E - U(\pmb{\sigma}) \vert } \sum_{\pmb{\sigma^\prime} \in S_1(\pmb{\sigma} )}
			\vert\psi(\pmb{\sigma}^\prime)\vert  \leq   \frac{C_\Gamma(E,\epsilon)}{N} \sum_{\pmb{\sigma^\prime} \in S_1(\pmb{\sigma} )}
			\vert \psi(\pmb{\sigma}^\prime)\vert  , \\ \quad C_\Gamma(E,\epsilon) &\coloneqq \frac{\Gamma}{ \vert E\vert/N -\epsilon} \end{split}  \ee
		We start at $ \pmb{\sigma}_\psi$ and use this estimate iteratively $\alpha N$ times for some $\alpha < \alpha_0$ to arrive at 
		\begin{align}\label{eq:appmain} \vert\psi(\pmb{\sigma}_\psi)\vert & \leq C_\Gamma(E,\epsilon)^{\alpha N}  \sum_{\pmb{\sigma} \in B_{\alpha N}(\pmb{\sigma}_\psi )} p_{\alpha N}(\pmb{\sigma}, \pmb{\sigma}_\psi ) \vert\psi(\pmb{\sigma})\vert \\ \leq   C_\Gamma(&E,\epsilon)^{\alpha N}  \Big(  \sum_{\pmb{\sigma} \in B_{\alpha N}(\pmb{\sigma}_\psi )} p_{\alpha N}(\pmb{\sigma}, \pmb{\sigma}_\psi )^2  \Big)^{1/2} \notag 
			\leq   C_\Gamma(E,\epsilon)^{\alpha N}  \max_{\pmb{\sigma}\in B_{\alpha N}(\pmb{\sigma}_\psi )} p_{\alpha N}(\pmb{\sigma}, \pmb{\sigma}_\psi )^{1/2} 
		\end{align}
		where $p_{\alpha N}(\pmb{\sigma}, \pmb{\sigma}_\psi ) $ is the probability of a simple random walk on $\mathcal{Q}_N$ starting at $ \pmb{\sigma}_\psi  $ to arrive at $\pmb{\sigma}$ after $\alpha N$ steps. Here we have the Cauchy-Schwarz inequality and the normalization of $ \psi $ as well as $ p_{\alpha N} $. 
		Since  $\vert E\vert > (\Gamma -\delta/4) N \geq \Gamma N (1 - 1/100) $  we see that for $N$ large enough:
		$ C_\Gamma(E,\epsilon) \leq \left(\frac{1}{1-1/50}\right)^{\alpha N} \leq e^{\alpha N/49} $. 
		Together with the probability bound~\eqref{eq:pwalk} our main estimate~\eqref{eq:appmain} yields 
		\be\label{eq:end}   \vert\psi(\pmb{\sigma}_\psi)\vert^2  \leq \max\{ e^{(\frac{2\alpha}{49} -\gamma(\alpha/8) + o(1)) \, N}, e^{-\frac{1}{12} \alpha N} \} \eqqcolon e^{-c_\alpha N/2} . \ee
		Due to the bound~\eqref{eq:entropynumber} on the number of eigenvalues, which allows us to pick $ \eta> 0 $ arbitrarily small not to spoil any exponential decay of the eigenfunctions, the proof of the exponential bound~\eqref{eq:deloccrit} in case of Step 1 is completed if we choose $ \alpha \in (0, \alpha_0) $ such that $\gamma(\alpha/8) > \frac{2\alpha}{49}$. This is always possible since $\gamma$ has an infinite slope at zero.\\[.5ex]

		\noindent
		\emph{Step 2:}~~We now turn to general case, in which the intersection
		$ W \coloneqq B_{\alpha_0 N}(\pmb{\sigma}) \cap \mathcal{L}_{\epsilon} $
		is nonempty, but finite of size at most $ K $. For $\pmb{\sigma} \in W$ the spherical mean estimate~\eqref{eq:sperical} turns into
		\[ \vert\psi(\pmb{\sigma})\vert \leq \frac{C_\Gamma(E,\beta_c -\delta)}{N} \sum_{\pmb{\sigma^\prime} \in S_1(\pmb{\sigma} )} \vert\psi(\pmb{\sigma}^\prime)\vert  \]
		as we still have $U(\pmb{\sigma}) \geq -(\beta_c - \delta)N$.
		For a random walk $ (Z_k)_{k\in\mathbb{N}} $, which starts again at $ Z_0 \coloneqq \pmb{\sigma}_\psi  $, let $ M_{\alpha N}(W) $ stand for the number of visits of sites in $ W $. 
		We way modify our prior estimate by distinguishing between the random walks with a high visit number $M_{\alpha N}(W) > tN$  and low visit number  $M_{\alpha N}(W) \leq tN$ . Indeed, abbreviating by $ p_{\alpha N}(\pmb{\sigma}, \pmb{\sigma}_\psi \vert M_{\alpha N}(W) < t N )  $ the transition probability of the random walk to reach $ \pmb{\sigma} $ after $ \alpha N $ steps and spending only less than $ t N $ steps in $ W $, we have  for any $t > 0$
		\begin{align*} 
			\vert\psi(\pmb{\sigma}_\psi)\vert  \leq & \ C_\Gamma(E,\epsilon)^{(\alpha-t) N} C_\Gamma(E,\beta_c-\delta)^{tN}  \sum_{\pmb{\sigma} \in B_{\alpha N}(\pmb{\sigma}_\psi )} p_{\alpha N}(\pmb{\sigma}, \pmb{\sigma}_\psi \vert M_{\alpha N}(W) < t N ) \; \vert\psi(\pmb{\sigma})\vert \notag \\
			& \quad + C_\Gamma(E,\beta_c-\delta)^{\alpha N}  \sum_{\pmb{\sigma} \in B_{\alpha N}(\pmb{\sigma}_\psi )} p_{\alpha N}(\pmb{\sigma}, \pmb{\sigma}_\psi \vert M_{\alpha N}(W) \geq t N ) \; \vert\psi(\pmb{\sigma})\vert  \\
			\leq & \ C_\Gamma(E,\beta_c-\delta)^{tN}  e^{-c_\alpha N/2} + C_\Gamma(E,\beta_c-\delta)^{\alpha N}  \Prob( M_{\alpha N}(W) \geq t N ) ,
		\end{align*}
		where $c_\alpha >0$ is the constant in~\eqref{eq:end}.
		According to~\eqref{eq:sejourntime} the probability in the right side decays for any $t>0$ faster than any exponential function. Thus, choosing $t = c_\alpha /(4 \ln C_\Gamma(E,\beta_c-\delta))$ yields
		$ \vert\psi(\pmb{\sigma}_\psi)\vert^2 \leq e^{-c_\alpha N/4} $
		for $N$ large enough and $c_\alpha > 0$ is a constant independent from $\delta$ and $\Gamma$.
		This completes the proof.
	\end{proof}

\section{Extreme localization regime}\label{sec:loc}

\subsection{Deep-hole geometry}

The proof of our main results in the spin-glass regime are based on the deep-hole geometry of the REM. They rest on the fact that the large extremal sites $ \mathcal{L}_{\beta_c - \delta} $ of the REM, which were defined in~\eqref{eq:large}, are well separated on $\mathcal{Q}_N$ at least if $ \delta \in (0, \beta_c) $ is not too large.

\begin{definition}\label{def:deephole}
	Let $ \varepsilon,\delta > 0 $ and $ \alpha \in (0,\tfrac{1}{2}) $. Then 
	$ U :\mathcal{Q}_N \to \mathbb{R} $ is said to satisfy:
	\begin{enumerate}
		\item  a  \emph{local $(\varepsilon,\delta,\alpha) $-deep hole scenario on $ B_{\alpha N}(\pmb{\sigma}) $} with  $\pmb{\sigma} \in \mathcal{L}_{\beta_c-\delta}$  if: 
		\begin{enumerate}
			\item  $\vert U(\pmb{\sigma}^\prime)\vert \leq \epsilon N$ for all $\pmb{\sigma}^\prime \in B_{\alpha N}(\pmb{\sigma})$ with $ \pmb{\sigma}^\prime \neq  \pmb{\sigma} $,
			\item $u(\pmb{\sigma})\coloneqq \frac{1}{N^2} \sum_{\pmb{\sigma}^\prime \in S_1(\pmb{\sigma}) } \vert U(\pmb{\sigma}^\prime) \vert \leq N^{-1/4}$.
		\end{enumerate}
		\item  a  \emph{global $(\varepsilon,\delta,\alpha) $-deep hole scenario} if:
		\begin{enumerate}
			\item $ U $ satisfies a local $(\varepsilon,\delta,\alpha) $-deep hole scenario on $ B_{\alpha N}(\pmb{\sigma}) $ for all $\pmb{\sigma} \in \mathcal{L}_{\beta_c-\delta}$,
			\item $ B_{\alpha N}( \pmb{\sigma}) \cap B_{\alpha N}( \pmb{\sigma}^\prime) = \emptyset $ for all pairs $ \pmb{\sigma},  \pmb{\sigma}^\prime \in \mathcal{L}_{\beta_c-\delta} $ with $  \pmb{\sigma} \neq \pmb{\sigma}^\prime $.
		\end{enumerate}
	\end{enumerate}
\end{definition}


The probabilistic estimate for the occurrence  of a global deep-hole scenario in the REM is the subject of the following lemma.

\begin{lemma}\label{lem:deeph}
	Let $\epsilon, \delta>0$ and $ \alpha \in (0,1/2) $ be such that
	\begin{equation}\label{eq:cond}
		2	\gamma(3\alpha) + \delta(2 \beta_c - \delta) < \epsilon^2 .
	\end{equation}
	The event $ \Omega_N(\varepsilon,\delta,\alpha) \coloneqq \left\{  \text{$U $ satisfies a global $(\varepsilon,\delta,\alpha) $-deep hole scenario} \right\} $ occurs with probability exponentially close to one, i.e., there is some $ c(\varepsilon,\delta,\alpha) > 0 $ such that for all $ N $ sufficiently large:
	\be\label{eq:GDH}
	\pp\left(\Omega_N(\varepsilon,\delta,\alpha)\right) \geq 1 - e^{-  c(\varepsilon,\delta,\alpha) N } .
	\ee 
\end{lemma} 
\begin{proof}
	We first bound  the probability of the event 
	\[ 
	\widehat  \Omega_N(\varepsilon,\delta,\alpha) \coloneqq \left\{ \exists \, \pmb{\sigma} \in \mathcal{L}_{\beta_c-\delta}, \, \pmb{\sigma}^\prime \in B_{3\alpha N}(\pmb{\sigma}) \backslash \{\pmb{\sigma} \}  \text{ s.t. } \vert U(\pmb{\sigma}^\prime)\vert > \epsilon N \right\} .
	\]
	On its complement, all  $  \pmb{\sigma} \in \mathcal{L}_{\beta_c-\delta} $ satisfy the first requirement in the local deep-hole definition on $ B_{\alpha N}( \pmb{\sigma}) \subset B_{3\alpha N}( \pmb{\sigma}) $, and the balls of radius $ \alpha N $ around the large deviation sites are disjoint., i.e., the second requirement in the global deep-hole definition is also checked. By a union bound and independence, we conclude:
	\[ \begin{split} \pp\left(\widehat  \Omega_N(\varepsilon,\delta,\alpha) \right) &\leq \sum_{\pmb{\sigma} \in \mathcal{Q}_N } \pp(U(\pmb{\sigma}) \leq - (\beta_c-\delta) N) \sum_{\pmb{\sigma}^\prime \in B_{3\alpha N}(\pmb{\sigma} )\setminus\{\pmb{\sigma}  \}}  \pp(\vert U(\pmb{\sigma}')\vert \geq \epsilon N ) \\
		& \leq 2^{N+1} \left\vert B_{3\alpha N} \right\vert  e^{-(\beta_c-\delta)^2N/2} e^{-\epsilon^2 N/2} \leq e^{\left(\gamma(3\alpha) +\beta_c \delta - \frac{\delta^2+\epsilon^2}{2} +o(1)\right)N }.\end{split} \]
	The second line is a result of the usual Gaussian-tail estimates and the fact that the volume of a Hamming ball of radius $ \alpha N < N/2 $ is asymptotically given in terms of the binary entropy, $ \ln \vert B_{\alpha N} \vert = N ( \gamma(\alpha) + o(1)) $ as $ N \to \infty $. Using assumption \eqref{eq:cond}, we see that the above probability is exponentially small in $N$.
	
	The proof is concluded by showing that the event 
	\be\label{def:Omu}
	\Omega_N^{u} \coloneqq \left\{ \max_{\pmb{\sigma} \in \mathcal{Q}_N}   u(\pmb{\sigma})  \leq N^{-1/4} \right\} 
	\ee
	occurs with a probability, which is exponentially close to one, i.e.
	\be\label{eq:uest} \pp(\exists \, \pmb{\sigma} \in \mathcal{Q}_N  \text{ s.t. } u(\pmb{\sigma}) > N^{-1/4}) \leq 2^{2N}  e^{-N^{3/2}/2} .\ee
	For a proof of this bound, we rewrite the moment-generating function of $  u(\pmb{\sigma}) $ for any $t >0$ in terms of a standard normal variable $ g $:
	\[ \mathbbm{E}[e^{t u(\pmb{\sigma}) }]  = \mathbbm{E}[e^{t N^{-3/2} \vert g\vert }]^N \leq 2^N \mathbbm{E}[e^{t N^{-3/2} g }]^N =   2^N e^{t^2/(2N^2)}.\]
	By an exponential Chebychev-Markov estimate with $t = N^{7/4}$, this  then yields
	$ \pp(u(\pmb{\sigma}) > N^{-1/4}) \leq 2^N e^{-N^{3/2}/2} $, 
	and hence the claim by a union bound using $\vert\mathcal{Q}_N\vert = 2^N$. 
\end{proof}

\subsection{Rank-one analysis}\label{sec:rank1}
If $ U $ satisfies a local $(\varepsilon,\delta,\alpha) $-deep hole scenario on  
$ B_{\alpha N}(\pmb{\sigma}) $ at some fixed $  \pmb{\sigma} \in \mathcal{L}_{\beta_c-\delta}  $, it is natural  
to consider the Hamiltonian $ H_{\alpha N}(\pmb{\sigma}) = \Gamma T_{\alpha N} + U  $ restricted to $ \ell^2(B_{\alpha N}(\pmb{\sigma})) $, i.e.
\[  \langle \delta_{\pmb{\tau}} \vert H_{\alpha N}(\pmb{\sigma}) \delta_{\pmb{\tau}^\prime} \rangle  = \langle \delta_{\pmb{\tau}} \vert H \delta_{\pmb{\tau}^\prime} \rangle  \ \mathbbm{1}_{B_{\alpha N}(\pmb{\sigma})}(\pmb{\tau})\mathbbm{1}_{B_{\alpha N}(\pmb{\sigma})}(\pmb{\tau}^\prime),   \]
A spectral analysis of these self-adjoint matrices is facilitated by rank-one perturbation theory. 
Since $ \delta_{\pmb{\sigma}} $ is a cyclic vector for $ H_{\alpha N}(\pmb{\sigma})  $, the spectrum can read from zeros of the meromorphic function given by
\be\label{eq:rank1} \begin{split}
	\langle \delta_{\pmb{\sigma}}  \vert \left( H_{\alpha N}(\pmb{\sigma}) - z\right)^{-1}   \delta_{\pmb{\sigma}} \rangle^{-1} &= U(\pmb{\sigma}) - \Sigma(\pmb{\sigma},z) , \quad  \\ \Sigma(\pmb{\sigma},z) &\coloneqq - \langle \delta_{\pmb{\sigma}}  \vert \big( H_{\alpha N}^{\prime} (\pmb{\sigma}) - z\big)^{-1}   \delta_{\pmb{\sigma}} \rangle^{-1} , \end{split}
\ee
where $  H_{\alpha N}^{\prime} (\pmb{\sigma})  $ coincides with the matrix $ H_{\alpha N}^{\prime} (\pmb{\sigma})  $ when setting $ U(\pmb{\sigma}) = 0 $.
Moreover, an $ \ell^2 $-normalized eigenvector $ \varphi_E $ corresponding to $ E \in \spec H_{\alpha N}(\pmb{\sigma}) $ is given in terms of the free resolvent, i.e., 
\be\label{eq:rank1ef}
\varphi_E( \pmb{\tau})  =  - U(\pmb{\sigma} ) \,   \varphi_E( \pmb{\sigma})  \langle \delta_{\pmb{\tau}}  \vert \left( H_{\alpha N}^{\prime} (\pmb{\sigma}) - E_{\pmb{\sigma} } \right)^{-1}   \delta_{\pmb{\sigma}} \rangle  ,
\ee
for any $  \pmb{\tau} \in B_{\alpha N}(\pmb{\sigma})  $, cf.~\cite[Theorem~5.3]{AW15}. The deep-hole scenario then entails the following information about the low-energy part of the spectrum.

\begin{lemma}\label{lem:rank1}
	Suppose $ U $ satisfies a local $(\varepsilon,\delta,\alpha) $-deep hole scenario on  
	$ B_{\alpha N}(\pmb{\sigma}) $ at some $  \pmb{\sigma} \in \mathcal{L}_{\beta_c-\delta}  $ with \be\label{eq:cond2}
	2 \Gamma \sqrt{\alpha (1-\alpha) } + \varepsilon < \beta_c - 2 \delta .
	\ee
	Then for all sufficiently large $ N $, the spectrum $ \spec_{E_\delta}  H_{\alpha N}(\pmb{\sigma})  \coloneqq  \spec H_{\alpha N}(\pmb{\sigma}) \cap (-\infty, E_\delta) $ below $ E_\delta := -N (\beta_c - \delta) $ consists only of one simple eigenvalue $ E_{\pmb{\sigma} }  $ which satisfies
	\begin{align}\label{eq:evasym}
		E_{\pmb{\sigma} } 
		& = U(\pmb{\sigma} ) + \frac{\Gamma^2 N}{ E_{\pmb{\sigma} }}   + \frac{\Gamma^2 }{ E_{\pmb{\sigma} }^2}  \sum_{\pmb{\sigma}^\prime \in S_1(\pmb{\sigma}) } U(\pmb{\sigma}^\prime) +  \mathcal{O}_{\Gamma,\delta,\epsilon}\left(N^{-5/4} \right) \notag \\ 
		&  = U(\pmb{\sigma} ) + \frac{\Gamma^2 N}{U(\pmb{\sigma} )}  + \mathcal{O}_{\Gamma,\delta}\left(N^{-1/4} \right) .
	\end{align}
	The $ \ell^2 $-normalized eigenfunction $ \psi_{\pmb{\sigma} } $ corresponding to  $  E_{\pmb{\sigma} }  $ satisfies:
	\begin{enumerate}
		\item for any $K \in \nn$ and for all  $\pmb{\sigma}^\prime \in S_{K}(\pmb{\sigma})$
		\begin{equation}\label{eq:decnear1}
			\vert\psi(\pmb{\sigma}^\prime)\vert = \mathcal{O}_{\Gamma,\delta,K}(N^{-K}), \quad \text{and} \quad \sum_{\pmb{\sigma}^\prime \notin B_K(\pmb{\sigma})} \vert\psi(\pmb{\sigma}^\prime)\vert^2 = \mathcal{O}_{\Gamma,\delta,K}(N^{-(K+1)}) .
		\end{equation} 
		\item  for any $ \alpha' \in (0,\alpha] $ there are  $ C = C(\Gamma,\delta) , c = c(\alpha,\alpha^\prime) \in (0,\infty) $,  such that
		\begin{equation}\label{eq:decfarlem}
			\sum_{\pmb{\sigma}^\prime \notin B_{\alpha' N }(\pmb{\sigma})} \vert\psi_{\pmb{\sigma}} (\pmb{\sigma}^\prime )\vert^2 \leq C N  \exp\left(- N  c \right) .
		\end{equation}
	\end{enumerate}
\end{lemma}
\begin{proof}
	The deep-hole scenario together with~\eqref{eq:tknorm2} and~\eqref{eq:cond2} implies that for all sufficiently large $ N $:
	\be\label{eq:lower}
	H_{\alpha N}^{\prime} (\pmb{\sigma}) \geq \Gamma T_{\alpha N} - \varepsilon N \geq - ( \beta_c - 2 \delta) N > E_\delta .
	\ee
	By rank-one perturbation theory, there is exactly one zero of \eqref{eq:rank1} and hence one simple eigenvalue  $ E_{\pmb{\sigma} }$ of $ H_{\alpha N}(\pmb{\sigma}) $ below $ \inf \spec H_{\alpha N}^{\prime} (\pmb{\sigma}) $. A Rayleigh-Ritz bound
	\be\label{eq:RR} E_{\pmb{\sigma} } \leq    \langle \delta_{\pmb{\sigma}}  \vert H_{\alpha N}(\pmb{\sigma})  \delta_{\pmb{\sigma}} \rangle = U(\pmb{\sigma} ) \leq E_\delta \ee
	provides a first, crude  estimate on this eigenvalue.  According to~\eqref{eq:rank1ef} the corresponding 
	$ \ell^2 $-normalized eigenvector $ \psi_{\pmb{\sigma}} $ satisfies for all $  \pmb{\sigma}^\prime \in B_{\alpha N }(\pmb{\sigma})$:
	\begin{align}
		\psi_{\pmb{\sigma}}( \pmb{\sigma}^\prime) & =  - U(\pmb{\sigma} ) \,  \psi_{\pmb{\sigma}}( \pmb{\sigma})  \langle \delta_{\pmb{\sigma}^\prime}  \vert \left( H_{\alpha N}^{\prime} (\pmb{\sigma}) - E_{\pmb{\sigma} } \right)^{-1}   \delta_{\pmb{\sigma}} \rangle \notag \\ &  \leq  - U(\pmb{\sigma} ) \,  \langle \delta_{\pmb{\sigma}^\prime}  \vert \left( \Gamma T_{\alpha N} - (E_{\pmb{\sigma} } + \varepsilon N) \right)^{-1}   \delta_{\pmb{\sigma}} \rangle \notag \\
		& \leq - U(\pmb{\sigma} ) \ \Gamma^{-1}  \ \langle \delta_{\pmb{\sigma}^\prime}  \vert \left( T_{\alpha N} - (U(\pmb{\sigma} )  + \varepsilon N) \Gamma^{-1} \right)^{-1}   \delta_{\pmb{\sigma}} \rangle .
	\end{align}
	As in~\eqref{eq:lower}, these inequalities are consequence of the deep-hole scenario, the crude bound~\eqref{eq:RR} combined with the positivity of the semigroup, cf.~\eqref{eq:semi}. The assertions~\eqref{eq:decnear1} and~\eqref{eq:decfarlem}  concerning the decay rates of the eigenfunction are now a straightforward consequence of Proposition~\ref{lem:freeB}. For its application, we note that the assumption~\eqref{eq:cond2} ensure that $ \dist( \Gamma^{-1} \spec  T_{\alpha N} , U(\pmb{\sigma} ) + \varepsilon N ) \geq \Gamma^{-1}(E_\delta - U(\pmb{\sigma} ) + \delta) N  \geq \frac{\delta}{\Gamma} N $. The first inequality in Proposition~\ref{lem:freeB} then yields
	\be\label{eq:L1est}
	\vert  \psi_{\pmb{\sigma}}( \pmb{\sigma}^\prime) \vert \leq \frac{\beta_c-\delta}{\delta} {N \choose d(\pmb{\sigma}_0,\pmb{\sigma})}^{-1/2} \ \  2^{-\min\{d(\pmb{\sigma},\pmb{\sigma}^\prime), \, \rho_0(\alpha) N\}} ,
	\ee
	where we also used that the function $U \mapsto \frac{-U}{x-U}$ is monotone increasing in $U$ on $(- \infty, x)$. Hence, \eqref{eq:decfarlem} follows after a summation over the spheres $ S_d(\pmb{\sigma}) $ with $ d \in (\alpha^\prime N , \alpha N] $.  The above binomial decay factor is thereby exactly compensated by the volume $ \vert S_d(\pmb{\sigma}) \vert =  {N \choose d} $. The claimed bounds~\eqref{eq:decnear1} follow analogously from the respective bounds in Proposition~\ref{lem:freeB}.  \\

	For a proof of the asymptotics~\eqref{eq:evasym}, we first consider the eigenvalue equation at any $\pmb{\sigma}^\prime \in S_1({\pmb{\sigma}})$:
	\begin{align}\label{eq:eigen1a}
		E_{\pmb{\sigma}} \psi_{\pmb{\sigma}}({\pmb{\sigma}}^\prime) &= U(\pmb{\sigma}^\prime) \psi_{\pmb{\sigma}}({\pmb{\sigma}}^\prime) - \Gamma \psi_{\pmb{\sigma}}({\pmb{\sigma}}) - \Gamma \sum_{\pmb{\sigma}^{\prime\prime} \in S_1(\pmb{\sigma}^\prime)\setminus \{\pmb{\sigma}\}} \psi_{\pmb{\sigma}}({\pmb{\sigma}}^{\prime\prime}) \notag \\
		&= U(\pmb{\sigma}^\prime) \psi_{\pmb{\sigma}}({\pmb{\sigma}}^\prime) - \Gamma \psi_{\pmb{\sigma}}({\pmb{\sigma}}) + \mathcal{O}_{\Gamma,\delta}(N^{-1}).
	\end{align}
	The uniform $ \mathcal{O}_{\Gamma,\delta}(N^{-1})$ estimate is a direct consequence of~\eqref{eq:decnear1}. 
	This equation
	can be rewritten as 
	\be\label{eq:evre}
	\psi_{\pmb{\sigma}}({\pmb{\sigma}}^\prime) = - \frac{\Gamma}{E_{\pmb{\sigma}}-U(\pmb{\sigma}^\prime)} \left( \psi_{\pmb{\sigma}}({\pmb{\sigma}})  + \mathcal{O}_{\Gamma,\delta}(N^{-1}) \right) , 
	\ee
	which we insert 
	into the eigenvalue equation at $ \pmb{\sigma}$:
	\begin{align}\label{eq:eigen1b} E_{\pmb{\sigma}} \psi_{\pmb{\sigma}}({\pmb{\sigma}}) &= U(\pmb{\sigma}) \psi_{\pmb{\sigma}}({\pmb{\sigma}}) - \Gamma \sum_{\pmb{\sigma}^{\prime} \in S_1(\pmb{\sigma})} \psi_{\pmb{\sigma}}({\pmb{\sigma}}^{\prime})\notag \\
		&= U(\pmb{\sigma}) \psi_{\pmb{\sigma}}({\pmb{\sigma}}) + \frac{\Gamma^2}{E_{\pmb{\sigma}}} \left(\sum_{\pmb{\sigma}^{\prime} \in S_1(\pmb{\sigma})} \frac{\psi_{\pmb{\sigma}}(\pmb{\sigma})  + \mathcal{O}_{\Gamma,\delta}(N^{-1})} {1-U(\pmb{\sigma}^{\prime})/E_{\pmb{\sigma}}} \right) \notag  \\
		&= U(\pmb{\sigma}) \psi_{\pmb{\sigma}}({\pmb{\sigma}}) + \frac{\Gamma^2}{E_{\pmb{\sigma}}} \left(\sum_{\pmb{\sigma}^{\prime} \in S_1(\pmb{\sigma})} \frac{\psi_{\pmb{\sigma}}(\pmb{\sigma})} {1-U(\pmb{\sigma}^\prime)/E_{\pmb{\sigma}}} \right) + \mathcal{O}_{\Gamma,\delta,\epsilon}(N^{-5/4}) \notag\\
		&= \left[ U(\pmb{\sigma})  + \frac{\Gamma^2 N }{E_{\pmb{\sigma}}}+ \frac{\Gamma^2}{E_{\pmb{\sigma}}}  \left(\sum_{\pmb{\sigma}^{\prime} \in S_1(\pmb{\sigma})} 
		\frac{U(\pmb{\sigma}^{\prime})}{E_{\pmb{\sigma}}} \right)\right] \psi_{\pmb{\sigma}}({\pmb{\sigma}}) + \mathcal{O}_{\Gamma,\delta,\epsilon}(N^{-5/4}) .
	\end{align}
	The third equality follow from a second-order Taylor expansion with an error estimate using $\vert U(\pmb{\sigma}^{\prime})\vert^2 \leq \epsilon N \vert U(\pmb{\sigma}^{\prime})\vert$ as well as the bound on $ u(\pmb{\sigma}) $ in the deep-hole assumption in Definition~\ref{def:deephole}. 
	Since $  \psi_{\pmb{\sigma}}({\pmb{\sigma}}) = 1 + \mathcal{O}(N^{-1})$, the first identity in~\eqref{eq:evasym} follows. For a proof of the second identity, we again use the bound on $ u(\pmb{\sigma}) $ as well as our crude estimate~\eqref{eq:RR}  to estimate the last term in the above square brackets by $\mathcal{O}_{\Gamma,\delta}(N^{-1/4})$. This concludes the proof.
\end{proof}

\subsection{Spectral averaging}
In order to control the probability of resonances between distinct extremal sites, we will use the spectral averaging technique from the theory of random operators~\cite[Chapter 4.1]{AW15}. 
\begin{lemma} \label{lem:specav}
	Let $\epsilon, \delta>0$ and $ \alpha \in (0,1/2) $ be such that~\eqref{eq:cond} and \eqref{eq:cond2} holds. Then there is some $ c = c(\varepsilon, \delta,\alpha) >  0 $ such that for all $ N $ sufficiently large and 
	\begin{enumerate}
		\item
		for any real interval $ I $:
		\be\label{eq:Wegner}
		\pp\left(\exists \, \pmb{\sigma} \in \mathcal{L}_{\beta_c-\delta} \text{ s.t. }  \spec_{E_\delta} H_{\alpha N}(\pmb{\sigma}) \cap I \neq \emptyset \right) \leq 2 \vert I\vert \ e^{\beta_c \delta N - \delta^2 N/2} + e^{-cN} .
		\ee
		\item for any $ r > 0 $:
		\begin{multline}\label{eq:Wegner2}
			\pp\left(\exists \, \pmb{\sigma},  \pmb{\sigma}^\prime \in \mathcal{L}_{\beta_c-\delta},  \pmb{\sigma}\neq \pmb{\sigma}^\prime \text{ s.t. } \dist\left( \spec_{E_\delta} H_{\alpha N}(\pmb{\sigma}),  \spec_{E_\delta} H_{\alpha N}(\pmb{\sigma}^\prime)  \right) 
			\leq  r  \right) \\ 
			\leq 4 r  e^{(2\beta_c \delta  - \delta^2 )N} + e^{-cN} .
		\end{multline}
	\end{enumerate}
\end{lemma}

	\begin{proof} 
		For a proof of the above estimates, we may thus restrict attention to events in $\Omega_N(\varepsilon,\delta,\alpha) $, cf.~Lemma~\ref{lem:deeph}. 
		\begin{enumerate}[wide, labelwidth=!, labelindent=0pt]
			\item According to Lemma~\ref{lem:rank1}, under the deep-hole scenario $  \spec_{E_\delta} H_{\alpha N}(\pmb{\sigma}) \cap I \neq \emptyset  $ if and only if $ E_{\pmb{\sigma} } = \inf\spec H_{\alpha N}(\pmb{\sigma}) \in I $. Since $  \psi_{\pmb{\sigma}}({\pmb{\sigma}})^2 \geq 1/2 $ by Lemma~\ref{lem:rank1} for sufficiently large $ N $ and all $ \pmb{\sigma} \in \mathcal{L}_{\beta_c-\delta} $, the latter implies $ \langle \delta_{\pmb{\sigma}} \vert P_I  \delta_{\pmb{\sigma}} \rangle \geq 1/2 $, where $ P_I $ denotes the spectral projection of $ H_{\alpha N}(\pmb{\sigma}) $ onto $ I $.  A union bound hence enables to estimate the probability of the event in the left side of~\eqref{eq:Wegner} and its intersection with  $\Omega_N(\varepsilon,\delta,\alpha) $ by
			\begin{align*}
				\sum_{\pmb{\sigma} \in \mathcal{Q}_N} \pp\left( \pmb{\sigma} \in \mathcal{L}_{\beta_c-\delta} \, \text{and} \,   \langle \delta_{\pmb{\sigma}} \vert P_I  \delta_{\pmb{\sigma}} \rangle \geq 1/2  \right)  \leq 2 \ \mathbb{E}\left[ \mathbbm{1}[ \pmb{\sigma} \in \mathcal{L}_{\beta_c-\delta} ] \  \langle \delta_{\pmb{\sigma}} \vert P_I  \delta_{\pmb{\sigma}}  \rangle \right] .
			\end{align*}
			The inequality is a Chebychev-Markov estimate. Conditioning on all random variables aside from $ U(\pmb{\sigma}) $, the integration of $ p_I( U(\pmb{\sigma})) :=  \langle \delta_{\pmb{\sigma}} \vert P_I  \delta_{\pmb{\sigma}}  \rangle $ with respect to the random variable $ U(\pmb{\sigma})  $  is bounded with the help of the spectral averaging lemma (also referred to as Wegner estimate, cf.~\cite[Thm.~4.1]{AW15}).
			 It yields
			\[
			\int_{-\infty}^{-(\beta_c-\delta) N} \mkern-20mu p_I(u) \ e^{-\frac{u^2}{2N} } \frac{du}{\sqrt{2\pi N}} \leq e^{- (\beta_c-\delta)^2 N/2} \vert I\vert . 
			\]
			This completes the proof of the first assertion.
			\item On $ \Omega_N(\varepsilon,\delta,\alpha) $, we may assume that $ B_{\alpha N}( \pmb{\sigma}) \cap B_{\alpha N}( \pmb{\sigma}^\prime) = \emptyset $ for all pairs $ \pmb{\sigma},  \pmb{\sigma}^\prime \in \mathcal{L}_{\beta_c-\delta} $. This ensures that the random variables  $ E_{\pmb{\sigma}^\prime } = \inf\spec H_{\alpha N}(\pmb{\sigma}^\prime) $ and $ U(\pmb{\sigma}^\prime) $ are independent of all random variables in $ B_{\alpha N}(\pmb{\sigma}) $. Using the strategy as in 1., we thus bound the probability of the event in the left side of~\eqref{eq:Wegner2} and its intersection with  $  \Omega_N(\varepsilon,\delta,\alpha) $ by
			\begin{multline*}
				\sum_{ \pmb{\sigma}, \pmb{\sigma}^\prime \in \mathcal{Q}_N} \mkern-10mu \mathbb{E}\left[ \mathbbm{1}[\pmb{\sigma}^\prime \in  \mathcal{L}_{\beta_c-\delta}  \, \text{and} \, B_{\alpha N}( \pmb{\sigma}) \cap B_{\alpha N}( \pmb{\sigma}^\prime) = \emptyset  ] \right. \\[-1ex]
				\times \left.  \pp\left(\pmb{\sigma}  \in  \mathcal{L}_{\beta_c-\delta}  \, \text{and} \,   \langle \delta_{\pmb{\sigma}} \vert P_{(E_{\pmb{\sigma}^\prime }-r, E_{\pmb{\sigma}^\prime }+r)} \delta_{\pmb{\sigma}} \rangle \geq 1/2  \,  \vert \,  B_{\alpha N}(\pmb{\sigma})^c \right)\right] 
				\leq 2^{2N+2}  e^{- (\beta_c-\delta)^2 N} r .   
			\end{multline*}
			where   $ \pp(  \cdot   \vert  B_{\alpha N}(\pmb{\sigma})^c ) $ denotes the conditional expectation, conditioned on all random variables aside from those in $ B_{\alpha N}(\pmb{\sigma}) $ and $  P_I  $ is still the spectral projection of $ H_{\alpha N}(\pmb{\sigma}) $ onto $ I $. The last inequality resulted from an application of the bound from~1.\ to the conditional expectation. This completes the proof of the second assertion. 
		\end{enumerate}
	\end{proof}

\subsection{Proof of Theorem~\ref{thm:sggs}}

The proof of Theorem~\ref{thm:sggs} makes use of the deep-hole geometry of the REM. If $ U $ satisfies a global $ (\varepsilon,\delta,\alpha) $-deep hole scenario, we study the auxiliary Hamiltonian
\be\label{eq:defHprime}
H^\prime \coloneqq \left( \bigoplus_{\pmb{\sigma} \in \mathcal{L}_{\beta_c-\delta} } H_{\alpha N}(\pmb{\sigma}) \right) \bigoplus H_r, 
\ee
with operators $ H_{\alpha N}(\pmb{\sigma}) $, whose action is restricted to the non-intersecting balls $ B_{\alpha N}(\pmb{\sigma}) $ around extremal sites $ \pmb{\sigma} \in \mathcal{L}_{\beta_c-\delta} $. These operator have been introduced and studied in Subsection~\ref{sec:rank1}. The remainder
$H_r$ is that part of $H$ which purely belongs to the complement of the union of balls, 
\[   \langle \delta_{\pmb{\tau}} \vert H_r \delta_{\pmb{\tau}^\prime} \rangle  =\langle \delta_{\pmb{\tau}} \vert H \delta_{\pmb{\tau}^\prime} \rangle \left(1-\sum_{\pmb{\sigma}\in \mathcal{L}_{\beta_c-\delta}}\mathbbm{1}_{B_{\alpha N}(\pmb{\sigma})}(\pmb{\tau})\right)\left(1- \sum_{\pmb{\sigma}\in \mathcal{L}_{\beta_c-\delta}}
\mathbbm{1}_{B_{\alpha N}(\pmb{\sigma})}(\pmb{\tau}^\prime) \right) .  \]
The difference between the Hamiltonian of interest $H = \Gamma T + U  $ and the auxiliary  $H^\prime$ is
\[ H - H^\prime \eqqcolon - \Gamma A \eqqcolon - \Gamma \bigoplus_{\pmb{\sigma} \in \mathcal{L}_{\beta_c-\delta}} A_{\pmb{\sigma}}  . \]
It describes the hopping between the balls and the complementary configuration space,   i.e.,
\[ \langle \delta_{\pmb{\tau}} \vert A_{\pmb{\sigma}}   \delta_{\pmb{\tau^\prime}} \rangle = \mathbbm{1}_{d(\pmb{\tau},\pmb{\tau}^\prime) =1} (\mathbbm{1}_{d(\pmb{\tau},\pmb{\sigma}) = \alpha N} \mathbbm{1}_{d(\pmb{\tau}^\prime,\pmb{\sigma}) = \alpha N+1} + \mathbbm{1}_{d(\pmb{\tau},\pmb{\sigma}) = \alpha N+1} \mathbbm{1}_{d(\pmb{\tau}^\prime,\pmb{\sigma}) = \alpha N}).  \] 
The norm of $A$ can be bounded as follows
\be\label{eq:norm:A}
\|A \| = \max_{\pmb{\sigma} \in \mathcal{L}_{\beta_c-\delta}} \|A_{\pmb{\sigma}} \| \leq \|T_{\alpha N + 1}  \| = 2 N \sqrt{\alpha(\alpha -1)} + o_{\alpha}(N)  , 
\ee
where the last equality is~\eqref{eq:tknorm}. 
It is easy to see that $ \| A \| $ is indeed of order $ N$. 
However, for energies below $ E_\delta = -N (\beta_c - \delta) $, the perturbation is of a much smaller magnitude. This is the basic idea in the proofs of our main results for the localization regime. As a preparation, we also need the following result, which is implicitly contained in~\cite{MW20c}.
\begin{proposition}[cf.~\cite{MW20c}]\label{lem:hr} For all $ \Gamma, \delta > 0 $ the truncated Hamiltonian $ H \coloneqq \Gamma T + U \mathbbm{1}_{U \geq  -(\beta_c-\delta)  N } $ acting on $ \ell^2(\mathcal{Q}_N) $ is lower bounded by
	\[ \inf\spec H \geq -N  \max\{ \Gamma , \beta_c-\delta \}  +  o_{\Gamma,\delta}(N)  \]
	except for an event of exponentially small probability.
\end{proposition}
	\begin{proof}
		The values of the truncated REM potential $U \mathbbm{1}[U(\pmb{\sigma}) \geq  -(\beta_c-\delta)  N ]$ are still independently distributed and satisfy a large deviation principle. The peeling principle \cite[Thm~2.3]{MW20c} then implies that the negative free energy of $ H $ at $ \beta > 0 $ is a maximum of the truncated REM and the pure paramagnet. The result on the ground-state energy thus follows in the limit $\beta \to \infty$. 
	\end{proof}
	We remark that the methods in \cite{MW20c} can be used to replace the $o(N)$ estimate by $\mathcal{O}(\sqrt{N})$.
\begin{proof}[Proof of Theorem~\ref{thm:sggs}]
	We only study the joint event $  \Omega_N(\Gamma,\delta,\alpha) $ on which i)~the bound in Proposition~\ref{lem:hr} applies, and ii)~$ U $ satisfies a global $ (\varepsilon,\delta,\alpha) $-deep hole scenario with parameters 
	\[
	\varepsilon = \frac{\beta_c}{2} \quad \text{and} \quad \delta \in (0,\min\{\beta_c - \Gamma, \beta_c/8 \} ), 
	\]
	and  $ \alpha > 0$ small enough such that~\eqref{eq:cond} and 
	$ 2 \Gamma \sqrt{\alpha (1-\alpha) } <  \delta/8 $, 	and hence in particular~\eqref{eq:cond2} is satisfied.
	Together with Lemma~\ref{lem:deeph} 
	this ensures that $  \Omega_N(\Gamma,\delta,\alpha)$  occurs with a probability of at least $ 1 - e^{- c  N} $ with at some $ c \equiv c(\Gamma, \delta,\alpha) >  0 $. 
	Moreover:
	\begin{enumerate}
		\item From Lemma~\ref{lem:rank1} we learn that 
		for any $\pmb{\sigma} \in \mathcal{L}_{\beta_c-\delta} $ the spectrum  $ \spec  H_{\alpha N}(\pmb{\sigma})  $ below $  E_\delta = - N (\beta_c -\delta) $ 
		consists of just one eigenvalue $  E_{\pmb{\sigma}} = \inf \spec  H_{\alpha N }(\pmb{\sigma})   $, which is given by~\eqref{eq:evasym} 
		with an error term $\mathcal{O}_{\Gamma,\delta}\left(N^{-1/4} \right)$ uniformly for all $ \pmb{\sigma}  \in \mathcal{L}_{\beta_c-\delta} $. 
		\item By the variational principle and the natural embedding of Hilbert spaces, the ground state energy of $H_r$ is bounded from below by that of $  \Gamma T + U \mathbbm{1}_{U \geq  -(\beta_c-\delta)  N } $ 
		on $ \ell^2(\mathcal{Q}_N) $. The lower bound in Proposition~\ref{lem:hr} then shows that 
		\[ \inf \spec H_r \geq - N \left(\beta_c - \delta +o_{\Gamma,\delta}(1) \right) .
		\] 
	\end{enumerate}
	Hence, $ H_r $ does not contribute to the low-energy spectrum of $ H' $ below  $ E_{\delta/2} = -N (\beta_c - \delta/2)  $ for all $ N $ large enough. 
	Moreover, the spectral projection $ P_{\delta}  \coloneqq \mathbbm{1}_{(-\infty,E_{\delta/2} )}(H^\prime) $ 
	can be written as
	\be\label{eq:Ps} 
	P_{\delta} = \sum_{\pmb{\sigma} \in \mathcal{L}_{\beta_c-\delta}, \, E_{\pmb{\sigma}} < E_{\delta/2}   } \vert \psi_{\pmb{\sigma}} \rangle \langle  \psi_{\pmb{\sigma}} \vert  \ee
	in terms of rank-one projections of the $ \ell^2 $-normalized ground states $  \psi_{\pmb{\sigma}} $ of $   H_{\alpha N}(\pmb{\sigma})  $. We thus conclude for some $ C = C(\Gamma,\delta) < \infty $, and $ c = c(\alpha) > 0 $
	\be\label{eq:PA}
	\|A P_{\delta} \| = \max_{\pmb{\sigma} \in \mathcal{L}_{\beta_c-\delta}} \|A_{\pmb{\sigma}} \psi_{\pmb{\sigma}} \| \leq  \|A\| 
	\max_{\pmb{\sigma} \in \mathcal{L}_{\beta_c-\delta}} \Big( \sum_{\pmb{\sigma}^\prime \in S_{\alpha N}(\pmb{\sigma})} \left\vert \psi_{\pmb{\sigma}} \left(\pmb{\sigma}^\prime\right) \right\vert^2 \Big)^{1/2} \leq C \ N^2 \ e^{-  c N} ,
	\ee
	where the inequalities follow from~\eqref{eq:norm:A}  and~\eqref{eq:decfarlem} together with the fact that $A_{\pmb{\sigma}}$ only acts on the part of $\psi_{\pmb{\sigma}}$ on $S_{\alpha N}(\pmb{\sigma})$.

	We then rewrite $ H $ using the block decomposition of $ \ell^2(\mathcal{Q}_N) $ induced by $ P_{\delta} $ and $ Q_\delta \coloneqq 1-P_{\delta} $ and again employ the Schur complement method.  Since $ H^\prime $ is diagonal in this decomposition and its  $ Q_{\delta}  $ projection has a spectrum above the threshold energy $ E_{\delta/2}  $,  it remains to investigate the blocks of the perturbation $ \Gamma A $:
	\begin{enumerate}
		\item 
		Since $ P_\delta $ is supported entirely on the balls, the first diagonal term  vanishes, i.e. $ P_\delta A P_\delta = 0 $. The operator norms of the off-diagonals 
		$
		\| P_\delta A (1-P_\delta) \|  \leq  \|A P_{\delta} \| 
		$
		are exponentially small by~\eqref{eq:PA}.  
		\item
		The  operator $Q_{\delta} A Q_{\delta} $ is bounded from below by $ - \|A\| $ which is estimated in~\eqref{eq:norm:A}.
		We thus conclude that for all $ N $ large enough:
		\begin{align*}
			Q_\delta H^\prime + Q_{\delta} A Q_{\delta}  &\geq E_{\delta/2} - \|A\| \geq -N \left( \beta_c - \delta/2 + 2 \Gamma \sqrt{\alpha (1-\alpha)} +o_{\Gamma,\alpha}(1) \right) \\& \geq -N \left( \beta_c - \delta/4 \right) . 
		\end{align*}
		Consequently, the Schur complement matrix
		\[ 
		S_\delta(E) := \left( Q_\delta H^\prime + Q_{\delta} A Q_{\delta}  - E \right)^{-1} 
		\]
		is well defined on $  Q_\delta \ell^2(\mathcal{Q}_N) $ and bounded, $ \| S_\delta(E) \|  \leq (E_{\delta/4}-E)^{-1} $ for any $ E < E_{\delta/4} $. 
	\end{enumerate}
	The spectrum of $ H $ below $ E_{\delta/4} = -N \left( \beta_c - \delta/4 \right) $ is thus characterized using Schur's method, which yields:
	\begin{enumerate}
		\item $ E < E_{\delta/4}$ is an eigenvalue of $ H $ if and only if $ E \in \spec\left( P_\delta H' - P_\delta A  S_\delta(E) A P_{\delta} \right) $.
		\item The $ \ell^2 $-normalized eigenvector $ \psi $ corresponding to $ E $ and $ H $ satisfies:
		\begin{align}\label{eq:schur2}
			(P_\delta H' - E P_\delta) \psi & =  P_\delta A  S_\delta(E) A P_{\delta} \psi \notag \\
			Q_\delta \psi & = - S_\delta(E) A P_\delta \psi . 
		\end{align}
	\end{enumerate}
	We now proceed with the completion of the proof of the assertion on the spectrum and eigenvectors separately.
	
	\bigskip 
	\noindent
	\textit{Spectrum:}~~The spectrum of $ H $ below $  E_{\delta/8} $ is determined through the above Schur complement method. Since for all $ E \leq  E_{\delta/8} $ at at some $ C = C(\Gamma,\delta) < \infty $ and $ c = c(\alpha) > 0 $
	\be\label{eq:pert}
	\| P_\delta A  S_\delta(E) A P_{\delta} \| \leq \| S_\delta(E) \| \  \|A P_{\delta} \|^2 \leq  C\ N^3 \, e^{- 2 c N } ,
	\ee
	the eigenvalues below $  E_{\delta/8} $ thus coincide with the eigenvalues of $ P_\delta H'  $ below this energy up to an error, which is exponentially small in $ N $~\cite[Corollary 3.2.6]{Bha97}. Since the eigenvalues of $ P_\delta H'  $ are given by~\eqref{eq:evasym}, the assertion in Theorem~\ref{thm:sggs} follows.
	
	\bigskip
	\noindent 
	\textit{Eigenvectors:}~~We concentrate our attention on energies below $ E_s = - N (\beta_c- s ) $ with $ s \in (0, \delta/8] $ small enough such that $ 2 \beta_c s <   c $ with the decay rate $ c > 0 $ from~\eqref{eq:pert}.
	This ensures that $ e^{-\alpha c N}  < e^{-2 \beta_c s N} \eqqcolon r(s) $ for all sufficiently large $ N $. 
	According to the spectral averaging Lemma~\ref{lem:specav}, since $ s \leq \delta/8 $ and the condition~\eqref{eq:cond2} is monotone in $ \delta $, the event 
	\be\label{eq:cond3}
	\left\{ \forall \, \pmb{\sigma},  \pmb{\sigma}^\prime \in \mathcal{L}_{\beta_c-s},  \pmb{\sigma}\neq \pmb{\sigma}^\prime :\quad  \dist\left( \spec_{E_s} H_{\alpha N}(\pmb{\sigma}),  \spec_{E_s} H_{\alpha N}(\pmb{\sigma}^\prime)  \right) 
	>  r(s)  \right\} 
	\ee
	has probability of at least $ 1 - 4 e^{- s^2 N } - e^{-c N} $ for some $ c > 0 $. We may therefore assume its occurrence. 
	
	Perturbation theory based on the above Schur complement analysis and~\eqref{eq:pert} (combined with the characterization of eigenvalues established in Theorem~\ref{thm:sggs}) then guarantees that the eigenvector $ \psi $ of $ H $ corresponding to the eigenvalue $ E= U(\pmb{\sigma}) + \Gamma^2 N/U(\pmb{\sigma}) + \mathcal{O}(N^{-1/4}) $, which is uniquely characterized by $  \pmb{\sigma} \in \mathcal{L}_{\beta_c-s} $, is norm-close to the ground-state eigenvector $ \psi_{  \pmb{\sigma}} $ 
	of $ H_{\alpha N}(\pmb{\sigma}) $, i.e.,
	\begin{align}
		\| \psi - \psi_{  \pmb{\sigma}} \| & \leq \| P_\delta\psi - \psi_{  \pmb{\sigma}} \| +  \| Q_\delta \psi \| \notag \\
		& \leq \frac{  \| P_\delta A  S_\delta(E) A P_{\delta} \|}{r(s)} +  \| S_\delta(E) \| \  \|A P_{\delta} \| \leq  C\   e^{- c N } . 
	\end{align}
	Here the inequalities combine~\eqref{eq:schur2}--\eqref{eq:cond3}. 
	The rest of the claim on the $ \ell^2 $-estimates of the eigenvectors then follows from the respective properties of $ \psi_{  \pmb{\sigma}}  $ established in Lemma~\ref{lem:rank1}.
	The event $\Omega_{N,\Gamma,\delta}^\textup{loc}$ is then defined by specifying a value for $\alpha_0 = \alpha_0(\Gamma,\delta)$ and intersecting $ \Omega_N(\Gamma,\delta,\alpha_0)$ with~\eqref{eq:cond3}.

\end{proof}

\subsection{Proof of Theorem~\ref{thm:sgstate}}\label{subsec:locL1}

All assertions concerning the $ \ell^2 $-properties of the ground-state can easily be collected from the proof of Theorem~\ref{thm:sggs}. 

\begin{proof}[Proof of Theorem~\ref{thm:sgstate} -- $ \ell^2 $-properties] 
	According to Theorem~\ref{thm:sggs} for all events aside from one of exponentially small probability, there is some $ \pmb{\sigma}_0 \in \mathcal{Q}_N $ such that  the ground state eigenvector is approximated by $  \| \psi - \delta_{\pmb{\sigma}_0 }\|_{\ell^2}^2 = \mathcal{O}_{\Gamma}\left(\frac{1}{N}\right) $. The estimate $\mathcal{O}_{\Gamma}\left(\frac{1}{N}\right) $ does not depend on $\delta$ anymore as we may fix $\delta$, if we only consider the ground state. This will be always assumed in the following. Moreover, the ground-state energy is 
	$ E = U(\pmb{\sigma}_0 ) + \frac{\Gamma^2 N}{U(\pmb{\sigma}_0 )}  + \mathcal{O}_{\Gamma}\left(N^{-1/4} \right) $, 
	where $ U(\pmb{\sigma}_0 ) $ is one of the REM's extremal energies for which 
	we may assume that 
	\be\label{eq:aprioriU} 
	\left\vert U(\pmb{\sigma}_0 ) + \beta_c N \right\vert \leq \mathcal{O}(\sqrt{N}) , \quad \text{and hence }\quad  \left\vert E + \beta_c N \right\vert \leq \mathcal{O}(\sqrt{N}) 
	\ee
	at the expense of excluding another event of exponentially small probability stemming from deviations to the known extremal statistics of the REM, cf.~\eqref{eq:Uinfty}.
	
	It thus remains to establish the assertion on the first order perturbation $ \xi \in \ell^2(\mathcal{Q}_N) $. 
	That $\langle \xi \vert H  \xi \rangle$ agrees with the ground state energy up to order $o_\Gamma(1)$ is a result of a simple calculation and a comparison with the above formula for $ E $. It remains to prove 
	$\|\psi - \xi \|^2 = \mathcal{O}_{\Gamma}(N^{-2})$. 
	To this end, we revisit the proof of Theorem~\ref{thm:sggs}. From the validity of the global $(\beta_c/2,\delta,\alpha) $-deep hole scenario specified there and in view of \eqref{eq:decnear1}, it suffices to show 
	\be\label{eq:claimsgs} \left\vert\psi(\pmb{\sigma_0}) - \sqrt{1-\frac{\Gamma^2}{\beta_c^2 N}} \right\vert^2 = \mathcal{O}_{\Gamma}(N^{-2})  \quad \text{and} \quad  \sum_{\pmb{\sigma} \in S_1(\pmb{\sigma}_0 )} \left\vert\psi(\pmb{\sigma}) - \frac{\Gamma}{\beta_c N} \right\vert^2 = \mathcal{O}_{\Gamma}(N^{-2}).  \ee

	For a proof of these assertions, we use the eigenvalue equation~\eqref{eq:evre} on $ S_1(\pmb{\sigma}_0) $ together with $\psi({\pmb{\sigma}_0}) = 1+ \mathcal{O}_{\Gamma}(N^{-1})$ . If we pick  $ \pmb{\sigma} \in S_1(\pmb{\sigma}_0) $, this  yields
	\[ \begin{split} \psi({\pmb{\sigma}}) - \frac{\Gamma}{\beta_c N}  &= - \frac{\Gamma}{E - U(\pmb{\sigma})} \left(1 + \mathcal{O}_{\Gamma}(N^{-1}) \right) - \frac{\Gamma}{\beta_c N} \\
		&= \frac{\Gamma U(\pmb{\sigma})}{\beta_c N ( E - U(\pmb{\sigma}) )} + \mathcal{O}_{\Gamma}(N^{-3/2}) .\end{split}   \]
	Here the last step also relied on the estimate $ \vert U(\pmb{\sigma})\vert \leq \varepsilon N $ valid in the $(\varepsilon,\delta,\alpha) $-deep hole scenario, as well as~\eqref{eq:aprioriU}. 
	With a suitable constant $ C = C(\Gamma) < \infty $, we then have 
	\[  \sum_{\pmb{\sigma} \in S_1(\pmb{\sigma}_0 )} \left\vert\psi(\pmb{\sigma}) - \frac{\Gamma}{\beta_c N} \right\vert^2  \leq  \frac{C}{N^4 }\sum_{\pmb{\sigma} \in S_1(\pmb{\sigma}_0 )} U(\pmb{\sigma})^{2} + \mathcal{O}_{\Gamma}(N^{-2}) \]
	An exponential Chebychev-Markov estimate leads to 
	$\pp( N^{-2} \sum_{\pmb{\sigma} \in S_1(\pmb{\sigma}_0 )} (U(\pmb{\sigma})^{2} - N) ) \geq 1) \leq e^{-c N}  $ 
	for some $ c >0$. Thus, except for an event of exponentially small probability the second claim in~\eqref{eq:claimsgs} holds.
	Since $\psi$ is $ \ell^2$-normalized, this leads to
	\[ \psi(\pmb{\sigma_0})^2 = 1 -  \sum_{\pmb{\sigma} \in S_1(\pmb{\sigma}_0 )} \psi(\pmb{\sigma})^2  + \mathcal{O}_{\Gamma}(N^{-2})
	=  1- \frac{\Gamma^2}{\beta_c^2 N} + \mathcal{O}_{\Gamma}(N^{-2}), \]
	which readily implies the first claim in~\eqref{eq:claimsgs}.
\end{proof}

For a proof of the $ \ell^1 $-estimate on the ground state eigenfunction, we need to sharpen estimates on the large-deviation geometry of the REM. To this end we define for $\epsilon, \delta > 0$ the following tripartition of the Hamming cube: 
\[ \begin{split}
	A_1(\epsilon) & \coloneqq \{ \pmb{\sigma} \in \mathcal{Q}_N \, \vert \,     \vert U(\pmb{\sigma})\vert \leq \epsilon N  \} \\
	A_2(\epsilon, \delta)  & \coloneqq \{ \pmb{\sigma} \in \mathcal{Q}_N \, \vert \,  \epsilon N <   \vert U(\pmb{\sigma})\vert \leq (\beta_c - \delta) N  \} \\
	A_3(\delta)  & \coloneqq \{ \pmb{\sigma} \in \mathcal{Q}_N \, \vert \,     \vert U(\pmb{\sigma})\vert > (\beta_c - \delta) N \} .  \
\end{split}  \]
A modification of ideas used in the proof of Lemma~\ref{lem:deeph} and~\cite[Lemma 2]{MW19} yields:
\begin{lemma}\label{lem:sep2}
	For any $\epsilon > 0$ there exist $K = K(\epsilon) \in \nn$ and a family of events $\Omega_{\epsilon, N}$ such that for $N$ large enough 
	\begin{enumerate}
		\item[(i)] For any $\pmb{\sigma} \in A_2(\epsilon, \delta) \cup A_3(\delta) $: \qquad $ \displaystyle \vert B_4(\pmb{\sigma}) \cap (A_2(\epsilon, \delta) \cup A_3(\delta)) \vert  \leq K $  \quad on $\Omega_{\epsilon, N}$.
		\item[(ii)] $\pp(\Omega_{\epsilon, N}) \geq 1 - 2^{-N}$.
	\end{enumerate}

\end{lemma}
\begin{proof}
Let $\Omega_{\epsilon, N,K}$ be the event, where the assertion (i) holds true with constant $K$. It remains to show that  the complement satisfies 
$\pp(\Omega_{\epsilon, N,K}^{c}) \leq 2^{-N}$ for an appropriate choice for $K$ and $N$ large enough. To this end we estimate
\[ \begin{split}&\quad \, \pp(\Omega_{\epsilon, N,K}^{c}) =   \pp( \exists\  \pmb{\sigma} \in A_2(\epsilon, \delta) \cup A_3(\delta)\; \text{s.t.} \; \vert B_4(\pmb{\sigma}) \cap (A_2(\epsilon, \delta) \cup A_3(\delta))\vert  \geq K ) \\ &\leq \sum_{\pmb{\sigma}_0  \in\mathcal{Q}_N} \pp(\vert U(\pmb{\sigma}_0)\vert \geq \epsilon N) \, \pp(\exists \, K-1 \text{ different } \pmb{\sigma}_1, \ldots  \pmb{\sigma}_{K-1} \in B_4(\pmb{\sigma}_0)\backslash \{ \pmb{\sigma}_0 \} \,  \text{s.t.}  \\[-2ex]
	& \mkern400mu \,    \vert U(\pmb{\sigma}_j)\vert \geq \epsilon N \text{ for } j=1,\ldots, K-1) \\ 
	&\leq \binom{N^4}{K-1}\,  2^N \, \pp(\vert U(\pmb{\sigma}_0)\vert \geq \epsilon N)^K \leq N^{4K} 2^N e^{- K N\epsilon^2/2},
\end{split}.  \]
Here the second line is a consequence of the union bound and the third line follows from the independence and a simple counting argument. Choosing $K > 4 \ln 2/\epsilon^2$, we see that $\pp(\Omega_{\epsilon, N,K}^{c}) < 2^{-N}$ for $N$ large enough.

\end{proof}

As a final preparation, we also need the following elementary observation on the size of large deviation sites.
\begin{lemma}\label{lem:card}
For any $ \delta \in (0,\beta_c) $ and all $ N $:
\[
\pp\left( \vert A_3(\delta)\vert \geq 2 e^{\beta_c \delta N} \right) \leq e^{-N\delta^2/2} .
\]
\end{lemma}
\begin{proof}
The cardinality $\vert A_3(\delta)\vert$ is a sum of $2^N $ independent Bernoulli variables with success probability $ p = \pp(\vert U(\pmb{\sigma})\vert > (\beta_c - \delta) N) \leq 2 e^{-\frac12 (\beta_c - \delta)^2 N} $
such that $ \mathbb{E}[\vert A_3((\delta))\vert] = 2^N p $. The claim thus follows from a standard Markov estimate. 
\end{proof}

We are finally ready to finish the proof of our main result in the localization regime.

\begin{proof}[Proof of Theorem~\ref{thm:sgstate} -- $ \ell^p $-properties]~~We first observe that the claims on the $\ell^p$-norms immediately follow from the $\ell^1$-norm asymptotics~\eqref{eq:ell1}.
To see this, recall that $\psi(\pmb{\sigma}_0 ) = 1 + o_{\Gamma}(1)$ for some $ \pmb{\sigma}_0 \in \mathcal{Q}_N $, and that 
$\psi(\pmb{\sigma}) \leq c \, N^{-1}$  for all $\pmb{\sigma} \neq \pmb{\sigma}_0$. Hence  for any $1 < p < \infty$:
\[  1 + o_{\Gamma}(1) \leq  \| \psi \|_{\ell^p}^p \leq 1 + \frac{c^{p-1}}{N^{p-1}} \sum_{\pmb{\sigma} \neq \pmb{\sigma}_0}\psi(\pmb{\sigma} ) \leq 1+ \frac{c^{p-1}}{N^{p-1}}  \| \psi \|_{\ell^1}  = 1+ o_{\Gamma,p}(1). \]
It therefore remains to establish~\eqref{eq:ell1}.\\

Recalling that the ground state wavefunction $\psi$ is positive, we can write 
$ \| \psi \|_{\ell^1} = \sum_{\pmb{\sigma}} \psi(\pmb{\sigma}) $. 
The eigenvalue equation for $\psi$ leads to 
\begin{equation}\label{eq:l1main} 
	\begin{split}  E \sum_{\pmb{\sigma}} \psi(\pmb{\sigma}) &= +\Gamma \sum_{\pmb{\sigma}} (T \psi)(\pmb{\sigma}) + \sum_{\pmb{\sigma}} (U \psi)(\pmb{\sigma}) =  - \Gamma N \sum_{\pmb{\sigma}}  \psi(\pmb{\sigma})  - \sum_{\pmb{\sigma}} U(\pmb{\sigma}) \psi(\pmb{\sigma})\\
		& = - \Gamma N \sum_{\pmb{\sigma}}  \psi(\pmb{\sigma}) + U(\pmb{\sigma}_0) \psi(\pmb{\sigma}_0) + \sum_{\pmb{\sigma} \neq \pmb{\sigma}_0} U(\pmb{\sigma}) \psi(\pmb{\sigma}). 
\end{split}  \end{equation}
The second equality follows from the fact that each $\pmb{\sigma}$ has $N$ neighbors.
The main idea is now to show that the remainder term $ \sum_{\pmb{\sigma} \neq \pmb{\sigma}_0} U(\pmb{\sigma}) \psi(\pmb{\sigma})$ can be controlled by the other two terms on the right side. Here, we use the tripartition $A_1(\epsilon), A_2(\epsilon,\delta), A_3(\delta)$ of the configuration space and bound the contribution of each $A_i$ separately. 

In the following, we fix $ \delta , \alpha > 0 $ small enough, such that the REM satisfies a global $ (\beta_c/2,\delta,\alpha ) $-deep hole scenario with a probability which is exponentially close to one.
Moreover, we pick  $ \varepsilon > 0 $ arbitrary and fix $ K=K(\varepsilon) \in \mathbb{N} $   the   assertions of Lemma~\ref{lem:sep2} hold on a joint event on which the global $ (\beta_c/2,\delta_0,\alpha_0 ) $-deep hole scenario applies as well. This event still has a probability of at least $ 1 - e^{-c(\delta,\alpha) N } $ with some $ c(\delta,\alpha) > 0 $, which is independent of $ \varepsilon $.\\

\noindent
\textit{Contribution of $A_1(\epsilon)$:} In this case we use the trivial estimate, 
$ \left\vert\sum_{\pmb{\sigma} \in A_1} U(\pmb{\sigma}) \psi(\pmb{\sigma})\right\vert \leq \epsilon N \| \psi \|_{\ell^1} $.\\

\noindent	\textit{Contribution of $A_3(\delta)$:} We only consider $\delta \leq \delta_0$, such that configurations $\pmb{\sigma} \in A_3(\delta) \setminus\{\pmb{\sigma}_0\}$ lie outside the ball $B_{\alpha_0 N}(\pmb{\sigma}_0)$. In particular, there is some $c > 0$ such that for all $N$ large enough and all $\pmb{\sigma} \in A_3(\delta) \setminus\{\pmb{\sigma}_0\}$
the ground state is uniformly bounded, $\vert\psi(\pmb{\sigma})\vert \leq e^{-cN}$. We now pick $ \delta \coloneqq \min\{\delta_0, c/(4 \beta_c) \}$ and shrink the considered event such that $ \vert A_3(\delta)\vert \leq 2 \ e^{N c/4} $. According to Lemma~\ref{lem:card} this event still has a probability greater than $ 1 - e^{-c(\delta,\alpha) N } $ with some $ c(\delta,\alpha) > 0 $.  On this event, we conclude for all $ N $ large enough
\[ \sum_{\pmb{\sigma} \in A_3(\delta) \setminus\{\pmb{\sigma}_0\}} \vert U(\pmb{\sigma})\vert \psi(\pmb{\sigma}) \leq e^{-N c/2 } . \]

\noindent	\textit{Contribution of $A_2(\epsilon,\delta)$:} 	We first consider the configurations in $A_2(\epsilon,\delta)$
close to the center $\pmb{\sigma}_0$, which we estimate for $N$ large by 
\[ \begin{split} \sum_{\pmb{\sigma} \in A_2(\epsilon,\delta) \cap B_4(\pmb{\sigma}_0) } \vert U(\pmb{\sigma})\vert \psi(\pmb{\sigma}) & \leq   \vert  A_2(\epsilon,\delta) \cap B_4(\pmb{\sigma}_0) \vert   \max_{\pmb{\sigma} \in B_4(\pmb{\sigma}_0) \backslash \{\pmb{\sigma}_0 \} }  \vert U(\pmb{\sigma})\vert \psi(\pmb{\sigma})\leq C \, K  \end{split}  \]
with some $C = C(\Gamma)$. We use $  \vert U(\pmb{\sigma})\vert  \leq  \beta_c N/2 $ due to the validity of the global  $ (\beta_c/2,\delta,\alpha ) $-deep hole scenario as well as the pointwise bound $\psi(\pmb{\sigma}) \leq C N^{-1} $ for all $ \pmb{\sigma} \in B_4(\pmb{\sigma}_0) $ with $\pmb{\sigma} \neq \pmb{\sigma}_0 $.

It remains to consider $\pmb{\sigma} \in A_2(\epsilon,\delta) \setminus B_4(\pmb{\sigma}_0)$.  The eigenvalue equation reads
\[ \vert E - U(\pmb{\sigma})\vert \psi(\pmb{\sigma}) = \Gamma \sum_{\pmb{\sigma}^\prime \in S_1(\pmb{\sigma})} \psi(\pmb{\sigma}^\prime).   \]
Since $E \leq (\beta_c - \delta/2) N$ for $N$ large enough, we obtain for $\pmb{\sigma} \in A_2(\epsilon,\delta)$ the bound
\[ \psi(\pmb{\sigma}) \leq \frac{2\Gamma}{\delta N } \sum_{\pmb{\sigma}^\prime \in S_1(\pmb{\sigma})}  \psi(\pmb{\sigma}^\prime).   \]
The essence of the following argument is that the value of any $\psi(\pmb{\sigma})$ is comparable to the mean on the corresponding $S_1(\pmb{\sigma})$ sphere and, thus, $\psi$ cannot take especially large values on $A_2(\epsilon,\delta)$.  
To make this intuition precise, we separate the $A_3(\delta)$ configurations, which we possibly encounter in the spherical mean and repeat the procedure for the remaining $\pmb{\sigma}^\prime \in S_1(\pmb{\sigma})$. This leads to 
\[ \begin{split} \psi(\pmb{\sigma}) &\leq  \frac{2\Gamma}{\delta N } \sum_{\pmb{\sigma}^\prime \in S_1(\pmb{\sigma}) \cap A_3(\delta)}  \psi(\pmb{\sigma}^\prime) +
	\frac{4 \Gamma^2}{\delta^2 N^2} \sum_{\pmb{\sigma}^\prime \in S_1(\pmb{\sigma}) \setminus A_3(\delta)} \sum_{\pmb{\sigma}^{\prime \prime} \in S_1(\pmb{\sigma}^\prime)} \psi(\pmb{\sigma}^{\prime \prime}) \\
	&\leq \frac{2\Gamma}{\delta N } \sum_{\pmb{\sigma}^\prime \in S_1(\pmb{\sigma}) \cap A_3(\delta)}  \psi(\pmb{\sigma}^\prime) + \frac{4 \Gamma^2}{\delta^2 N} \psi(\pmb{\sigma}) + \frac{8 \Gamma^2}{\delta^2 N^2} \sum_{\pmb{\sigma}^\prime \in S_2(\pmb{\sigma})}  \psi(\pmb{\sigma}^{ \prime}), \end{split}  \]
which for $N$ large enough implies 
\[ \psi(\pmb{\sigma}) \leq \frac{4 \Gamma}{\delta N } \sum_{\pmb{\sigma}^\prime \in S_1(\pmb{\sigma}) \cap A_3(\delta)}  \psi(\pmb{\sigma}^\prime) +  \frac{16 \Gamma^2}{\delta^2 N^2} \sum_{\pmb{\sigma}^\prime \in S_2(\pmb{\sigma})}  \psi(\pmb{\sigma}^{ \prime}).  \]
We now further shrink the considered event to ensure that  $\|U\|_\infty\leq 2 \beta_c N$ holds true. This happens for all but an event of exponentially probability,  cf.~\eqref{eq:REMevent}. Thus, for $N$ large enough
\[  \begin{split} \Bigg\vert\sum_{\pmb{\sigma} \in A_2 \setminus B_4(\pmb{\sigma}_0)}& U(\pmb{\sigma}) \psi(\pmb{\sigma})\Bigg\vert  \leq 2 \beta_c N \sum_{\pmb{\sigma} \in A_2 \setminus B_4(\pmb{\sigma}_0)}  \psi(\pmb{\sigma}) \\ & \leq  \sum_{\pmb{\sigma} \in A_2 \setminus B_4(\pmb{\sigma}_0)}  \Big( \frac{8 \beta_c \Gamma}{\delta} \sum_{\pmb{\sigma}^\prime \in S_1(\pmb{\sigma}) \cap A_3(\delta)}  \psi(\pmb{\sigma}^\prime) + \frac{16 \beta_c \Gamma^2}{\delta^2 N} \sum_{\pmb{\sigma}^\prime \in S_2(\pmb{\sigma})}   \psi(\pmb{\sigma}^{ \prime}) \Big) \\ & \leq \frac{8 \beta_c K \Gamma}{\delta} \sum_{\pmb{\sigma} \in A_3(\delta) \setminus\{\pmb{\sigma}_0\}} \psi(\pmb{\sigma})+   \frac{16 \beta_c \Gamma^2 K }{\delta^2 N} \|\psi\|_{\ell^1}  \leq e^{-N c/4 } +  \frac{16 \beta_c \Gamma^2 K }{\delta^2 N} \|\psi\|_{\ell^1}   . \end{split} \]
In the third line we used the observation that each configuration $\pmb{\sigma} \in \mathcal{Q}_N$ appears in the summation  at most $K$ times  due to Lemma~\ref{lem:sep2}. The last step is a consequence of our exponential bound on $ \psi(\pmb{\sigma}) $ on the $A_3(\delta)$-configurations.

Combining the partial results on each $A_J$, we arrive with some $C = C(\Gamma)$ at the bound
\[ \vert \sum_{\pmb{\sigma} \neq \pmb{\sigma}_0} U(\pmb{\sigma}) \psi(\pmb{\sigma}) \vert \leq (2 \epsilon + \mathcal{O}_{K,\Gamma}(N^{-1}) ) N  \| \psi \|_{\ell^1} + 4 C K  \]
which is valid on with probability of at least $ 1- e^{-c N }  $ with some $ c > 0 $ which in independent of $\varepsilon$. Since $\varepsilon> 0 $ was arbitrary, the claimed convergence now follows from~\eqref{eq:l1main}. 
\end{proof}

\section{Free energy asympotics}\label{sec:free}

For our proof of Theorem~\ref{thm:free} we exploit that the partition function is determined by the eigenvalues close to the thermal averages 
\begin{equation}
\langle U \rangle^{\text{cl}}_\beta \coloneqq \frac{ \tr U e^{-\beta U}}{\tr e^{-\beta U}} \quad\text{or}\quad \langle T \rangle^{\text{pm}}_\beta \coloneqq \frac{ \tr T e^{-\beta T}}{\tr e^{-\beta T}} ,
\end{equation}
depending on the phase. To determine their behavior we consider the local region around $\pmb{\sigma} \in \mathcal{L}_{\epsilon}$, where $ \epsilon > 0 $ has to be allowed to be arbitrarily small. In this case, we cannot guarantee anymore that all balls $B_R(\pmb{\sigma})$ are disjoint.  However, we will show that this is still true for isolated 
extremal sites $\pmb{\sigma} \in \mathcal{L}_{\epsilon}$, which are in the majority. 
Then, we establish the order-one corrections of Theorem~\ref{thm:sggs} for those isolated large deviations.
Based on these results, we prove Theorem~\ref{thm:free} via a suitable approximation argument using auxiliary operators on cut domains of the configuration space.

\subsection{Basic large deviations}
We first record some standard facts in the statistical mechanics of the pure REM and pure paramagnet.
\begin{proposition}\label{lem:saddle}
\begin{enumerate}
	\item For any $\beta \geq 0$ we have 
	$	\langle T \rangle^{\text{pm}}_\beta = - N \tanh \beta $.
	Moreover, for any $\beta \geq 0, \delta > 0$ there exists some $c = c(\beta,\delta) > 0$ such that 
	\begin{equation}
		\frac{ \tr \mathbbm{1}_{[-N(\tanh \beta + \delta),-N(\tanh \beta - \delta) ]}(T) e^{-\beta T}}{\tr e^{-\beta T}} \geq 1 - e^{-cN}.
	\end{equation}
	\item For $\beta < \beta_c$ we have almost surely
	$			\langle U \rangle^{\text{cl}}_\beta = - (\beta + o(1)) N $.
	Moreover, for $\beta < \beta_c $ and $ \delta > 0$ there exists some $c = c(\beta,\delta) > 0$ such that 
	\begin{equation}
		\frac{ \tr \mathbbm{1}_{[-N(\beta + \delta),-N( \beta - \delta) ]}(U) e^{-\beta U}}{\tr e^{-\beta U}} \geq 1 - e^{-cN}
	\end{equation}
	except for an exponentially small event.
	\item For $\beta > \beta_c$ we have almost surely
	$ \langle U \rangle^{\text{cl}}_\beta = - (\beta_c + o(1)) N $.
	Moreover, for $\beta > \beta_c $ and $ \delta > 0$ there exists some $c = c(\delta) > 0$ such that 
	\begin{equation}
		\frac{ \tr \mathbbm{1}_{(-\infty,-N( \beta_c - \delta) ]}(U) e^{-\beta U}}{\tr e^{-\beta U}} \geq 1 - e^{-cN}.
	\end{equation}
\end{enumerate}	
\end{proposition}
The proof for the expressions of the thermal averages $\langle T \rangle^{\text{pm}}_\beta, 	\langle U \rangle^{\text{cl}}_\beta $  is  by differentiating the explicit formulas for the pressure with respect to $\beta$. The results on the concentration of the Gibbs measure are then of Cram\'er type and follow from the usual convexity estimates of the (explicit) free energy. 

\subsection{Spectral analysis on clusters}

In the proofs of Theorem~\ref{thm:sggs} and \ref{thm:sgstate} we derived the order-one correction of the energy levels $U(\pmb{\sigma})$  caused by extremal sites $\pmb{\sigma} \in \mathcal{L}_{\epsilon}$ with $ \epsilon \approx \beta_c $ from a local analysis on non-overlapping balls $B_R(\pmb{\sigma})$ of some radius $ R $.  For the proof of Theorem~\ref{thm:free} however, we need good control on all eigenvalues with energy below $ - \varepsilon N $ with $ \varepsilon > 0 $ arbitrary  and, thus, the large deviation set 
$\mathcal{L}_{\epsilon}$ has to be considered for any $\epsilon >0$. The balls $B_R(\pmb{\sigma})$ then have a nonempty intersection and the aim of this subsection is to deal with this modified situation. 
Let us introduce some definitions and notation. 
\begin{definition}
Let $\epsilon > 0$ and  $ k \in \mathbb{N}_0 $. We denote 
$ \pmb{\sigma} \overset{k}{\sim} \pmb{\sigma}^\prime \iff d(\pmb{\sigma},\pmb{\sigma}^\prime) \leq 2k +2 $. 
We call a set $ G \subset \mathcal{L}_{\epsilon}$ \emph{$(k,\epsilon)$-connected} (with respect to $ \overset{k}{\sim} $) if for any $\pmb{\sigma},\pmb{\sigma}^\prime \in G$ there exists a sequence $\pmb{\sigma} = \pmb{\sigma}^{(0)}, \pmb{\sigma}^{(1)}, \cdots, \pmb{\sigma}^{(m)} = \pmb{\sigma}^{\prime}$ such that $\pmb{\sigma}^{(i)} \in G$ and $\pmb{\sigma}^{(i)} \overset{k}{\sim} \pmb{\sigma}^{(i+1)}$ for all $ 0 \leq i \leq m-1$. If $G \subset \mathcal{L}_{\epsilon}$ is $(k,\epsilon)$-connected and for any $(k,\epsilon)$-connected $G^\prime$ with $G \subset G^\prime \subset \mathcal{L}_{\epsilon}$ it follows $G = G^\prime$, we call $G$ a \emph{$(k,\varepsilon)$-component}. We denote the family of $(k,\epsilon)$-components of $\mathcal{L}_{\epsilon}$ by $\mathcal{G}_{k,\epsilon}$. 

We call $\pmb{\sigma} \in \mathcal{Q}_N$ \emph{$(k,\epsilon)$-isolated} if $G = \{\pmb{\sigma}  \} \in \mathcal{G}_{k,\epsilon} $ and $I_{k,\epsilon}$ denotes the collection of $(k,\epsilon)$-isolated configurations.
\end{definition}
The case $k=0$ coincides with the notion of 'gap-connected' used in \cite{MW19,MW20,MW20c}.

The extremal set $\mathcal{L}_{\epsilon}$ naturally decomposes in its components, i.e.,
$\mathcal{L}_{\epsilon} = \cup_{G \in \mathcal{G}_{k,\epsilon} } G$.
We define for each $(k,\epsilon)$-component $G$ the corresponding cluster 
\[ C_k(G) \coloneqq \bigcup_{\pmb{\sigma} \in G} B_k(\pmb{\sigma}). \]
By construction $d(C_k(G),C_k(G^\prime)) \geq 2$ for different $k$-components $G \neq G^\prime$. \\

We start with a combinatorial lemma which shows that the size of  $(k,\epsilon)$-components remains bounded and that most  $(k,\epsilon)$-components are isolated.
\begin{lemma}\label{lem:clus}
Let $\epsilon > 0$ and $k \in \nn_0$ be fixed, but arbitrary.
\begin{enumerate}
	\item There exists an $M = M(k,\epsilon) \in \nn$ such that 
	\begin{equation}\label{def:clusterevent}
		\Omega_{N,M}(\epsilon,k) \coloneqq \left\{ \max_{G \in \mathcal{G}_{k,\epsilon}} \vert G\vert \leq M \right\} .
	\end{equation}
	occurs with probability $ \pp( \Omega_{N,M}(k,\epsilon) ) \geq 1 - e^{-c N} $ for some $ c > 0 $. 
	\item Let $ \epsilon < a < b$ and $b < \beta_c$. Then for all events, but one of exponentially small probability:
	\begin{equation}\label{eq:isocomp}
		\frac{ \vert  \mathcal{L}_{a,b} \cap I_{k,\epsilon}^{c}  \vert}{\vert \mathcal{L}_{a,b}\vert} \leq e^{-\epsilon^2 N/4} ,
	\end{equation}
	where $ \mathcal{L}_{a,b} \coloneqq  \mathcal{L}_{a} \cap \mathcal{L}_{b}^{c}  $ and $ (\cdot)^c $ indicates the complement of that set.
	\item Let $a = \beta_c - \delta$ and suppose that $ \delta (2 \beta_c-\delta) < \epsilon^2$. Then, besides of an exponentially small event 
	\begin{equation}\label{eq:isobc}
		\mathcal{L}_{a} \cap I_{k,\epsilon}^{c} = \emptyset.
	\end{equation}
\end{enumerate}
\end{lemma}

\begin{proof}
	For a proof of the first assertion, we estimate for any $ M \in \mathbb{N} $ using a union bound
	\[ \begin{split} \pp(\max_{G \in \mathcal{G}_{k,\epsilon}} \vert G\vert \geq M) &\leq \pp(\exists \, \pmb{\sigma} \in \mathcal{Q}_N \text{ s.t. } \vert B_{(2k+2)M}( \pmb{\sigma}) \cap \mathcal{L}_\epsilon\vert \geq M) \\
		&\leq 2^N \binom{ \vert B_{(2k+2)M}\vert}{M} e^{-\frac12 \epsilon^2 M N}.
	\end{split} \]
	The first inequality follows from the definition of $k$-connectedness: any $k$-connected set with size $M$ is contained in some ball $B_{(2k+2)M}( \pmb{\sigma})$ and if $\max_{G \in \mathcal{G}_{k,\epsilon}} \vert G\vert \geq M$ there exists a $k$-connected set with size $M$ (e.g. as a subset of the component with maximal size). The second line follows from the union bound, the independence of the random variables $U$ and the standard Gaussian tail estimate. 
	As the binomial coefficient is a polynomial in $N$ the claim follows for $M > 2 \ln 2/ \epsilon^2$.
	
	For a proof of the second assertion, we rewrite 
	$ \vert \mathcal{L}_{a,b}\vert = \sum_{\pmb{\sigma}} Z_{\pmb{\sigma}} $,
	where $Z_{\pmb{\sigma}}$ are iid Bernoulli variables with success probability 
	\be\label{eq:defpN} p_N \coloneqq \pp(\pmb{\sigma} \in  \mathcal{L}_{a,b} ) \geq \frac{\sqrt{N} (b-a)}{\sqrt{2\pi}} \ e^{-Nb^2/2} .
	\ee
	Since $ b < \beta_c$, the average size  $  \mathbb{E}[\vert \mathcal{L}_{a,b}\vert] = 2^N p_N $ is exponentially large and, by a Markov estimate, the same applies to all events aside from one of super-exponentially small probability, i.e., $ \pp(\vert \mathcal{L}_{a,b}\vert \leq 2^{N-1} p_N  ) \leq e^{-e^{C N}} $ for some $ C > 0 $. 
	Similarly, the conditional probability $ \mathbb{P}_{\pmb{\sigma}} := \mathbb{P}\left( \cdot \vert \{ \pmb{\sigma} \}^c \right) $ of the configuration to not be $(\epsilon,k) $-isolated equals the probability to find on $B^{o}_{2k+2} \coloneqq B_{2k+2}(\pmb{\sigma} )\backslash\{ \pmb{\sigma}\} $ another large deviation in $\mathcal{L}_{\epsilon}$ and hence 
	$ \mathbb{P}_{\pmb{\sigma}}(\exists \ \pmb{\sigma}^\prime \in  I_{k,\epsilon}^c\cap B^{o}_{2k+2} ) \leq \vert B_{2k+2}\vert e^{-N \epsilon^2/2} \leq N^{2k+2} e^{-N \epsilon^2/2 } $ 
	by the union bound and the Gaussian-tail estimate. This allows us to estimate
	\be\label{eq:meanb} \mathbb{E}\left[ \vert \mathcal{L}_{a,b} \cap I_{k,\epsilon}^{c}  \vert \right] = \sum_{\pmb{\sigma} \in \mathcal{Q}_N} \mathbb{E}\left[ \mathbbm{1}[\pmb{\sigma} \in  \mathcal{L}_{a,b}] \,  \mathbb{P}_{\pmb{\sigma}}(\pmb{\sigma} \in  I_{k,\epsilon}^c) \right] \leq 2^N p_N \, N^{2k+2} e^{-N \epsilon^2/2 } 
	\ee
	with $p_N $ from~\eqref{eq:defpN}. Excluding the event on which $ \vert \mathcal{L}_{a,b}\vert \leq 2^{N-1} p_N $, we thus arrive at
	\begin{align*}
		\pp\left( \vert   \mathcal{L}_{a,b} \cap I_{k,\epsilon}^{c} \vert \geq e^{-\epsilon^2N/4 } \vert \mathcal{L}_{a,b}\vert  \right)  & \leq  \pp\left( \vert   \mathcal{L}_{a,b} \cap I_{k,\epsilon}^{c} \vert \geq e^{-\epsilon^2N/4 } 2^{N-1} p_N \right) + e^{-e^{CN}} \\
		& \leq \frac{e^{\epsilon^2N/4 } }{2^{N-1} p_N}  \mathbb{E}\left[ \vert \mathcal{L}_{a,b} \cap I_{k,\epsilon}^{c}  \vert \right] + e^{-e^{CN}} 
	\end{align*} 
	by a Chebychev-Markov estimate. Inserting the bound~\eqref{eq:meanb} completes the proof.
	
	For the last assertion, we note that by Lemma~\ref{lem:deeph}  the condition on $\delta$ implies that for $\alpha > 0$ small  enough a global $(\epsilon,\delta,\alpha)$-deep hole scenario occurs
	with probability exponentially close to one.
\end{proof}

The next lemma establishes the spectral properties of the restriction $H_{C_k(G)}$ of the QREM Hamiltonian to  the Hilbert space $ \ell^2(C_k(G)) $ of a cluster corresponding to $G \in \mathcal{G}_{k,\epsilon}$. For its formulation, we define for $ \delta > 0 $ the spectral projections 
\[ P_{\delta}(G) \coloneqq \mathbbm{1}_{(-\infty,-\delta N)}(H_{C_k(G)}), \quad Q_{\delta}(G) \coloneqq \mathbbm{1}-P_{\delta}(G).\]

Recall the events $ \Omega_N^{u}  $ defined in~\eqref{def:Omu} and $  \Omega_{N,M}(k,\epsilon) $ defined in \eqref{def:clusterevent}.
\begin{lemma}\label{lem:gsclus}
Let $\epsilon > 0$ and $k\geq 2 $. On the event  $ \Omega_N^{u} \cap  \Omega_{N,M}(\epsilon,k) $
the following assertions are valid for all $N$ large enough:
\begin{enumerate}
	\item $ \displaystyle
	\max_{G \in \mathcal{G}_{k,\epsilon}} \| H_{C_k(G)} - U_{C_k(G)} \| = \mathcal{O}_{\Gamma,k,M}(\sqrt{N}) $. 
	\item If $\psi$ is an $ \ell^2 $-normalized eigenfunction of  $H_{C_k(G)}$ with $\langle \psi, \, H_{C_k(G)}\psi \rangle \leq - \frac32 \epsilon N$, we have 
	\begin{equation}\label{eq:esteigen}
		\vert\psi(\pmb{\sigma})\vert = \mathcal{O}_{\Gamma,k,M,\epsilon}( N^{-\dist(G,\pmb{\sigma})}),
	\end{equation}
	and 
	\begin{equation}\label{eq:esteigena}
		\| \vert\mathbbm{1}_{C_k(G)\backslash G }\psi \|^2 = \mathcal{O}_{\Gamma,k,M,\epsilon}(N^{-1}) , \qquad  \|\mathbbm{1}_{\partial C_k(G)}\psi \|^2 = \mathcal{O}_{\Gamma,k,M,\epsilon}(N^{-k}), 
	\end{equation}
	where $\mathbbm{1}_{\partial C_k(G)}$ is the natural projection onto the boundary of $C_k(G)$. In particular, all estimates are independent of $ \psi $ and $ G $.
	\item $ \displaystyle
	\sup_{ G\in \mathcal{G}_{k,\epsilon}}  \sup_{\pmb{\sigma} \in \mathcal{L}_{2 \epsilon} \cap G} \langle \delta_{\pmb{\sigma}} \vert  Q_{3 \epsilon/2}(G) \delta_{\pmb{\sigma}} \rangle = \mathcal{O}_{\Gamma,k,M,\epsilon}(N^{-1})
	$.
	\item If $G = \{\pmb{\sigma_0}  \}$ is $(k,\epsilon)$-isolated and $U(\pmb{\sigma}_0) \leq - 2\epsilon N$, then the ground state energy of $H_{C_k(G)}$ is given by 
	\begin{equation}\label{eq:isol}
		E_{\pmb{\sigma_0}} \coloneqq \inf\spec H_{C_k(G)} =	U(\pmb{\sigma}_0)	+ \frac{\Gamma^2 N}{	U(\pmb{\sigma}_0)} + \mathcal{O}_{\Gamma,k,\epsilon}(N^{-1/4}).
	\end{equation}
\end{enumerate}	
\end{lemma}

\begin{proof}
\begin{enumerate}[wide, labelwidth=!, labelindent=0pt]
	\item We write $H_{C_k(G)} = U_{C_k(G)} + \Gamma T_{C_k(G)}$ and recall that $C_k(G)$ is a union of at most $M$ Hamming balls $B_k(\pmb{\sigma})$ with $ \pmb{\sigma} \in G $. Thus, by the triangle inequality and Proposition~\ref{lem:tknorm} we obtain 
	$ \|T_{C_k(G)}\| \leq M c_k \sqrt{N} $,
	and hence the claim.
	
	\item We introduce the modified spheres $S_r(G)$ for $0 \leq r \leq k$,
	\[ S_r(G) \coloneqq C_r(G) \setminus C_{r-1}(G) = \{ \pmb{\sigma} \in C_k(G) \, \vert \, \text{dist}(\pmb{\sigma},G) = r \}  \]
	and for the eigenvector $ \psi $ the maximal values on the spheres,
	$ s_r \coloneqq \max_{\pmb{\sigma} \in S_r(G)} \vert\psi(\pmb{\sigma})\vert $. 
	We use the convention $S_0(G) = G$ and note that $S_k(G) = \partial C_k(G)$. Moreover, we observe that for any $\pmb{\sigma} \in S_r(G)$ and $1 \leq r \leq k$:
	\[ \begin{split} \vv S_1(\pmb{\sigma}) \cap S_r(G)\vv \leq M, \qquad   r &\leq \vv S_1(\pmb{\sigma}) \cap S_{r-1}(G)\vv \leq r M, \\ 
		\quad   N - (r+1)M &\leq \vv S_1(\pmb{\sigma}) \cap S_{r+1}(G)\vv \leq N - r.
	\end{split}	\]
	We now use the eigenvalue equation\[ - E \psi(\pmb{\sigma}) = \Gamma \sum_{\pmb{\sigma}^\prime\in S_1(\pmb{\sigma})} \psi( \pmb{\sigma}^\prime) - U(\pmb{\sigma}) \psi(\pmb{\sigma}) ,
	\] 
	to derive the claimed decay estimate. Inserting the above geometric bounds into the eigenvalue equation, we obtain for all  $  1 \leq r \leq k $ with the convention $ s_{k+1} = 0 $:
	\begin{equation}\label{eq:evmax}
		- E s_r \leq \Gamma rM s_{r-1} + \Gamma (N-r) s_{r+1} + (\epsilon N + \Gamma M) s_r  .
	\end{equation}
	We claim that for all $  1 \leq r \leq k $ and $ N $ large enough
	\[s_r \leq \frac{2 M \Gamma k }{\vert E\vert-\epsilon N - \Gamma M} s_{r-1}.\]
	This is immediate from~\eqref{eq:evmax} in case $ r = k $ (even without the factor $ 2 $). In case $ 1 \leq r < k $, the bound is proven recursively. 
	If the inequality holds for $ r+1 $, then~\eqref{eq:evmax} implies 
	\[ \vert E\vert s_r \leq \Gamma rM s_{r-1}  + \left( \frac{2 N \Gamma^2 k  }{\vert E\vert-\epsilon N - \Gamma M } + \epsilon N  + \Gamma M \right) s_r ,
	\]
	and hence the claimed inequality for all $ N $ large enough. 
	Since  $s_0  \leq 1$, this establishes~\eqref{eq:esteigen} by iteration. 
	
	The first claim in~\eqref{eq:esteigena} follows from~\eqref{eq:esteigen} using $ \vert G\vert\leq M $. Indeed, for some $C = C(\Gamma,k,M,\epsilon)$
	\[
	\sum_{\pmb{\sigma} \neq \pmb{\sigma}_0} \vert\psi(\pmb{\sigma})\vert^2 \leq C \sum_{r=1}^K N^{-2r} \left\vert \{ \pmb{\sigma} \, \vert \,\dist(\pmb{\sigma},G)  = r \} \right\vert \leq C M \,  \sum_{r=1}^K N^{-r}  =\frac{C M}{N-1} . 
	\]
	Since $\vert\partial C_k(G)\vert \leq M N^k$ the second inequality in~\eqref{eq:esteigena} follows similarly. 
	\item 
	We will repeatedly make use of a coupling principle which follows from~1., namely the fact that the eigenvalues of $H_{C_k(G)} $ and $ U_{C_k(G)} $ agree up to a uniform error of order $ \mathcal{O}_{\Gamma,k,M}(N^{1/2}) $.  
	Since $\vert \mathcal{L}_\varepsilon \cap G \vert \leq \vert G\vert \leq M$, this implies that $ \dim P_{3\epsilon/2}(G) \leq M$ for any component $G \in \mathcal{G}_{k,\epsilon}$ if $N$ is chosen large enough. By the pigeon-hole principle, for any component $G$ we find some $a = a(G) \in [3  \epsilon /2, 2 \epsilon] $ such that 
	\begin{equation}\label{eq:achoice} P_{a - \epsilon/(2M)}(G) - P_{a + \epsilon/(2M)}(G) = \mathbbm{1}_{(-(a+\epsilon/(2M))N,- (a-\epsilon/(2M))N}(H_{C_k(G)}) = 0. \end{equation} 
	Since $ Q_{3 \epsilon/2}(G) \leq Q_a(G)$, it is enough to prove the assertion with $  Q_{3 \epsilon/2}(G) $ replaced by $ Q_a(G) $. 
	
	To this end, we fix a component $G \in \mathcal{G}_{k,\epsilon}$ 
	and observe that the coupling principle and \eqref{eq:achoice} yield
	\be\label{eq:choicea} \vert \mathcal{L}_a \cap G\vert = \dim  P_a(G) \eqqcolon m_a  \ee
	with a natural number $m_a \leq M$. We denote by $\psi_1, \ldots \psi_{m_a}  $ the  normalized low energy eigenfunctions of $H_{C_k(G)}$ corresponding to $P_a(G)$, which form an orthonormal basis for this subspace.   The first inequality in~\eqref{eq:esteigena} bounds the contribution of each eigenfunction to $  C_k(G)\setminus G $.
	Moreover, the eigenvalue equation readily implies for
	$\pmb{\sigma} \in G \setminus \mathcal{L}_a$: 
	\[ \vert\psi_j(\pmb{\sigma}))\vert \leq\frac{\Gamma}{\vert E_j - U(\pmb{\sigma}) \vert} \sum_{\pmb{\sigma}^\prime \in S_1(\pmb{\sigma} )} \vert\psi_j(\pmb{\sigma}^\prime))\vert \leq \frac{2M \Gamma}{\epsilon N} \sum_{\pmb{\sigma}^\prime \in S_1(\pmb{\sigma} )} \vert\psi_j(\pmb{\sigma}^\prime))\vert \leq  \frac{2M \Gamma }{\epsilon \sqrt{N}},     \]
	with $E_j = \langle \psi_j, \, H_{C_k(G)}\psi_j \rangle \leq - a N $ the eigenvalue corresponding to $\psi_j$. The second inequality follows from \eqref{eq:achoice} and the last step is a consequence of the Cauchy-Schwarz inequality and $ \| \psi_j \| = 1 $. As $\vert G \setminus \mathcal{L}_a\vert \leq M$, we also conclude that 
	\[ \sum_{\pmb{\sigma} \in C_k(G)\setminus \mathcal{L}_a } \vert\psi_j(\pmb{\sigma})\vert^2 \leq \frac{C}{N} \]
	with a uniform $C = C(\Gamma,k,M,\epsilon) < \infty$. We thus learn that   $\sup_j \| \mathbbm{1}_{\mathcal{L}_a \cap G}\  \psi_j - \psi_j \|^2 \leq C/N$. Lemma~\ref{lem:proj} below shows that (with $P =  \mathbbm{1}_{\mathcal{L}_a \cap G}$ and $F = P_a(G))$
	\[ \begin{split} \sup_{\pmb{\sigma} \in \mathcal{L}_{a } \cap G} \langle \delta_{\pmb{\sigma}} \vv  Q_{a}(G) \delta_{\pmb{\sigma}} \rangle &\leq \| Q_{a}(G) \mathbbm{1}_{\mathcal{L}_a \cap G} \| 
	 \leq 4(\sqrt{m_a} + m_a)^2 \frac{C}{N} \leq 16 M^{2} \frac{C}{N}. \end{split}   \]
	Since $a \leq 2 \epsilon$, this proves the claim.
	
	\item 
	By the Rayleigh-Ritz variational principle, we have $E_{\pmb{\sigma}_0} \leq - 2 \epsilon N$, and hence the results of~2. apply to the  corresponding ground state wavefunction  $\psi \in \ell^2(C_k(G) ) $. By~\eqref{eq:esteigena} this ensures  $ \psi(\pmb{\sigma}_0) = 1 + O_{\Gamma,k,\epsilon}(N^{-1/2}) $.	Following the steps in analysis~\eqref{eq:eigen1a}--\eqref{eq:eigen1b} of the eigenfunction equation, in which we use \eqref{eq:esteigen} and the assumed bound on $ u $, we thus conclude that~\eqref{eq:evasym}  remains valid. This concludes the proof of~\eqref{eq:isol}.
\end{enumerate}
\end{proof}

In the proof of Lemma~\ref{lem:gsclus} we used the following result on finite-rank projections:
\begin{lemma}\label{lem:proj}
Suppose $\mathcal{H}$ is a finite-dimensional Hilbert space, $P$ an orthogonal projection of rank $m$ and $f_1,f_2, \ldots f_m$ a sequence of $m$ orthonormal vectors in  $\mathcal{H}$, which span the projection $ F $.  If for some $ c < \infty $
\begin{equation}\label{eq:Pfcond}
	\max_{j = 1, \ldots, m}	\|Pf_j - f_j \| \leq c , 
\end{equation} 	
then $
\| P - F \| \leq (m + 2\sqrt{m}) c $.
\end{lemma}

\begin{proof}
We employ the triangle inequality 
$ \| P - F \| \leq \| PF - F \| +  \| PF - PFP \| + \| P - PFP  \| $ and
bound the three terms on the right-hand side individually. For the first term we invoke that $ PF - F$ vanishes on the orthogonal complement $\im F^{\perp}$ and, thus, a Frobenius norm estimate yields
\[ \| PF - F \| \leq \sqrt{ \sum_{j=1}^{m} \| (P-F) f_j \|^2  } = \sqrt{ \sum_{j=1}^{m} \|  Pf_j - f_j \|^2  } \leq \sqrt{m} c. \]
Our bound on the second term, relies on the norm estimate for the first term,
$ \| PF - PFP \| = \| P(F - FP) \| \leq \| F - FP \| = \| PF - F \| \leq  \sqrt{m} c $, 
where we used that $\|P \| = 1$ for the first bound and applied the elementary identity $\| A \| = \| A^{\ast} \|$. For the last term, we employ the operator inequality $0 \leq PFP \leq P$ and the fact that the operator norm is bounded by the trace norm $\| \cdot \|_1$: 
\begin{align*} \| P - PFP  \| & \leq  \| P - PFP  \|_1 = \tr P - \tr PFP = \sum_{j=1}^m \langle \psi_j , (1-P) \psi_j \rangle \leq m c .
\end{align*}
This completes the proof.
\end{proof}

\subsection{Proof of Theorem~\ref{thm:free}}

Before we dive into the details of the proof, we fix some notation. 
For $k \in \nn$ and $\epsilon >0$, we will use the restricted Hamiltonian corresponding to the collection of all clusters $C_k(G)$,
\[ H^{(c)} \coloneqq \bigoplus_{G \in \mathcal{G}_{k,\epsilon}} H_{C_k(G)} \]
acting on the complete Hilbert space $\ell^2(\mathcal{Q}_N)$.
We further denote by 
\[ P_{\epsilon}^{(c)} \coloneqq \mathbbm{1}_{(-\infty,-3 N\epsilon/2 )}( H^{(c)})  = \bigoplus_{G \in \mathcal{G}_{k,\epsilon}} P_{3\epsilon/2}(G), \quad Q_{\epsilon}^{(c)} \coloneqq \mathbbm{1} - P_{\epsilon}^{(c)} \]
the spectral projections of  $H^{(c)}$.
The factor $3/2$ is motivated by the third assertion of Lemma~\ref{lem:gsclus}. 
The subspace corresponding to $ P_{\delta}^{(c)} $ represents the "localized" part of the QREM and $ Q_{\delta}^{(c)}  $ corresponds  to the "delocalized" part. Corresponding to this block decomposition, we set the diagonal parts of $ H $ as well as their partition functions:
\begin{align}
&H^{(1)} \coloneqq P_{ \epsilon}^{(c)} H P_{ \epsilon}^{(c)}, \quad H^{(2)} \coloneqq Q_{ \epsilon}^{(c)} H Q_{ \epsilon}^{(c)}, \\
&Z_N^{(j)}(\beta,\Gamma) \coloneqq 2^{-N} \tr_j e^{-\beta 	H^{(j)}} , \quad j =1,2  .
\end{align}
Here the traces $ \tr_j(\cdot) $ run over the natural subspaces $P_{ \epsilon}^{(c)} \ell^{2}(\mathcal{Q}_N)$ in case $ j = 1 $, or  $Q_{ \epsilon}^{(c)} \ell^{2}(\mathcal{Q}_N)$ in case $ j = 2 $, on which $ H^{(j)} $ acts non-trivially.

The key observation is now that $P_{ \epsilon}^{(c)}$ commutes with the restriction of $H$ to the clusters $H^{(c)}$. If we denote by $A$ the  adjacency matrix between the inner and outer boundaries of the clusters $C_k(G)$, we see
\begin{equation}\label{eq:opA} P_{\epsilon}^{(c)} H Q_{ \epsilon}^{(c)} = \Gamma  P_{ \epsilon}^{(c)} A Q_{ \epsilon}^{(c)}.   \end{equation}
We recall that $d(C_k(G),C_k(G^\prime)) \geq 2 $ for two different components $G \neq G^\prime \in \mathcal{G}_{k,\epsilon}$, which implies that the adjacency matrix $A$ is a direct sum of operators $A_{C_k(G)}$ corresponding to each cluster $C_k(G) $. This in turn yields  
\be\label{eq:normA}
\| A \| \leq M c_k \sqrt{N}
\ee
by Proposition~\ref{lem:tknorm}on the event on which the assertions in Lemma~\ref{lem:clus} apply. 
We further observe that $A$ only acts nontrivially on the boundaries $\partial C_k(G)$. Exploiting the decay estimate~\eqref{eq:esteigena} from Lemma~\ref{lem:gsclus}, we arrive at 
\[ \|P_{ \epsilon}^{(c)} H Q_{ \epsilon}^{(c)}\| = \mathcal{O}_{\Gamma,k,M,\epsilon}(N^{-(k-1)/2}). \]
We conclude that for any $k \geq 2$:
\begin{equation}\label{eq:partsplit} 
Z_N(\beta,\Gamma) = e^{o_{\Gamma,k,M,\epsilon}(1)} (Z_N^{(1)} + Z_N^{(2)})  . 
\end{equation}
The proof of Theorem~\ref{thm:free} now reduces to an analysis of $Z_N^{(1)}$ and $Z_N^{(2)}$.

\begin{proof}[Proof of Theorem~\ref{thm:free}]

Since our claims in case $\beta = 0$ and $\Gamma = 0$ are trivial, we fix $\beta,\Gamma >0$ away from the phase transition, and pick  
\[ 0< \epsilon < \frac18 \min\{\beta, \beta_c, \Gamma \tanh \beta \Gamma, \min\{1, \beta^{-1} \} \ln \cosh \beta \Gamma \} . \]

In the following, we will only work on the event $  \Omega_{N,\beta_c}^\textup{REM} \cap \Omega_N^{u} \cap  \Omega_{N,M}(\epsilon,k) $, where the conditions of Lemma~\ref{lem:clus} are valid at $ k \geq 2 $ and some $ M $.
According to Lemma	~\ref{lem:gsclus} and~\eqref{eq:uest} as well as~\eqref{eq:Uinfty}, this event can be chosen to have a probability of at least $ 1 - e^{-c N } $.

We now proceed in four steps. We first analyze the localized part $Z_N^{(1)}$. As a  second and third step we derive an upper and lower bound for  $Z_N^{(2)}$. The last part then collects these estimates.\\

\noindent	\textit{Step 1 --  Analysis of $Z_N^{(1)}$:}~~Let us first remark that $H^{(1)} =  P_{ \epsilon}^{(c)} H P_{ \epsilon}^{(c)} = P_{ \epsilon}^{(c)} H^{(c)} P_{ \epsilon}^{(c)}$ and hence $H^{(1)} P_{ \epsilon}^{(c)} =  H^{(c)} P_{ \epsilon}^{(c)}$. It thus remains to consider the low energy spectrum of~$H^{(c)}$.
We abbreviate 
\[
u \coloneqq u(\beta) \coloneqq \lim_{N\to \infty} \langle  U \rangle_\beta^{\text{cl}}/N = \begin{cases} - \beta , & \beta \leq \beta_c , \\
	- \beta_c , & \beta > \beta_c ,
\end{cases}
\]  
by Proposition~\ref{lem:saddle}. Since $ 8 \epsilon < - u$, the dominant energy levels of $U$ are not effected by the projection $P_{ \epsilon}^{(c)}$. We now  split  $Z_N^{(1)}$ into the contribution arising from energy levels within $ J_\delta := [ (u-\delta)N , (u+\delta) N ] $ with 
$  0 < \delta < \min\{-u/2 -3\epsilon/4, \epsilon^2/(16\beta) \} $ 
arbitrarily small, and a remainder:
\[ \begin{split}
	Z_N^{(1)}(\beta,\Gamma) = 2^{-N} \left( \tr e^{-\beta 	H^{(1)}} \mathbbm{1}_{J_\delta}(H^{(1)}) + \tr e^{-\beta 	H^{(1)}} \mathbbm{1}_{(-\infty, -3 \epsilon N /2) \setminus J_\delta}(H^{(1)})  \right).
\end{split} \]
The second term is estimated using Lemma~\ref{lem:gsclus}  and subsequently Proposition~\ref{lem:saddle}, which yields some  $c = c(\beta,\delta) > 0$ such that for all sufficiently large $ N $: 
\be \tr e^{-\beta 	H^{(1)}} \mathbbm{1}_{\rr \setminus  J_\delta }(H^{(1)}) \leq  \ e^{\beta \mathcal{O}_{\Gamma,k,M}(\sqrt{N}) } \ \tr e^{- \beta U}  \mathbbm{1}_{\rr \setminus  J_{\delta/2} }(U) \leq e^{-c N  } \, 2^N Z_N(\beta,0)  .  \ee
The remaining term is decomposed further into the contribution of isolated and non-isolated clusters:
\begin{align}\label{eq:splitZ1} \tr e^{-\beta 	H^{(1)}} \mathbbm{1}_{J_\delta}(H^{(1)}) = & \sum_{G \in \mathcal{I}_{k,\epsilon}} \tr e^{-\beta H_{C_k(G)} }\mathbbm{1}_{J_\delta}(H_{C_k(G)}) \notag \\ & + 
	\sum_{G \in \mathcal{G}_{k,\epsilon} \setminus I_{k,\epsilon}^{c}} \tr e^{-\beta H_{C_k(G)} }\mathbbm{1}_{J_\delta}(H_{C_k(G)}) .
\end{align}
Since $ \sup_{G \in \mathcal{G}_{k,\epsilon}} \| H_{C_k(G)} - U_{C_k(G)} \| \leq \mathcal{O}_{\Gamma,k,M}(\sqrt{N}) $ by Lemma~\ref{lem:gsclus}, we bound the second term  for all sufficiently large $ N $ as follows:
\begin{multline}\label{eq:thincompl}   \sum_{G \in \mathcal{G}_{k,\epsilon} \setminus I_{k,\epsilon}^{c}} \tr e^{-\beta H_{C_k(G)} }\mathbbm{1}_{J_\delta}(H_{C_k(G)})  \leq e^{-\beta N (u-\delta)  } \sum_{G \in \mathcal{G}_{k,\epsilon} \setminus I_{k,\epsilon}^{c}}  \tr  \mathbbm{1}_{J_{2\delta}}(U_{C_k(G)})  \\
	\leq e^{-N \varepsilon^2/4} e^{3\beta N \delta}  \sum_{G \in \mathcal{G}_{k,\epsilon} }  \tr  e^{-\beta U_{C_k(G)} }\mathbbm{1}_{J_{2\delta}}(U_{C_k(G)}) \leq e^{-c N}  \, 2^N Z_N(\beta,0)   .
\end{multline}
At the expense of throwing out another event of exponentially small probability, we consult Lemma~\ref{lem:clus} 
and assume in case $ u > - \beta_c$ the validity of~\eqref{eq:isocomp}  with $ a = -u - 2 \delta $ and $ b = - u + 2 \delta $ and in case $ u = - \beta_c$ the validity of~\eqref{eq:isobc} with $ a = -u- 2 \delta$.  This  guarantees  that 
non-isolated clusters are exponentially rare. The last inequality is a consequence of the choice of $ \delta $ and of the fact  our definition of the partition function $ Z_N $ includes a normilisation by $2^{-N}$.	

The first term on the right side of~\eqref{eq:splitZ1} can be expressed using the energy correction formula \eqref{eq:isol} for isolated extremal sites. At the expense of excluding or including small subintervals at the boundary of $ J_\delta $, which are negligible in comparison to the main term by Proposition~\ref{lem:saddle},
this first term is of the form
\[ \sum_{\pmb{\sigma} \in  I_{k,\epsilon} \cap \mathcal{L}_{u-\delta, u+\delta}} e^{-\beta(U(\pmb{\sigma})	+ \frac{\Gamma^2 N}{	U(\pmb{\sigma})} + o(1))} =  S	- R , 
\]
where, similarly to~\eqref{eq:thincompl},  the remainder is again bounded using~\eqref{eq:isocomp}:
\[
R \coloneqq  \sum_{\pmb{\sigma} \in  I_{k,\epsilon}^c \cap \mathcal{L}_{u-\delta,u+\delta}} e^{-\beta(U(\pmb{\sigma})	+ \frac{\Gamma^2 N}{	U(\pmb{\sigma})} + o(1))} \leq  e^{-c N}  \, 2^N Z_N(\beta,0)  . 
\]
The main term is
\[ S \coloneqq \sum_{\pmb{\sigma} \in \mathcal{Q}_N} e^{-\beta U(\pmb{\sigma}) + \frac{\Gamma^2 N}{	U(\pmb{\sigma})} + o(1))}  \mathbbm{1}[  U(\pmb{\sigma}) \in J_\delta] .
\]
By definition of $J_\delta $ and since the REM's partition function concentrates around $ u $ by Proposition~\ref{lem:saddle}, $ S $ equals $ 2^{N} Z(\beta,0 ) $ plus an error which is bounded by $ e^{-c N} 2^{N} Z(\beta,0 ) $.

In summary, in this first step we have shown that for any $\delta >0$ small and $N$ large enough: 
\be\label{eq:REMshift} e^{-\frac{\beta \Gamma^2}{u-\delta} + o(1)} Z_N(\beta,0) 
\leq  Z_N^{(1)}(\beta, \Gamma)  \leq e^{-\frac{\beta \Gamma^2}{u+\delta} + o(1)} Z_N(\beta,0) . 
\ee

\medskip
\noindent	
\textit{Step 2 -- Upper bound on $Z_N^{(2)}$:} 
We write $U_\epsilon^< \coloneqq U  \mathbbm{1}_{U \geq - 2 \epsilon N} $ and 
$ U_\epsilon^> \coloneqq U   \mathbbm{1}_{U > - 2 \epsilon N}$ as well as $U_\epsilon \coloneqq  U  \mathbbm{1}_{\vert U\vert \leq  2 \epsilon N} $, and estimate using the Jensen-Peierls inequality~\cite{Ber72}: 
\[ \begin{split} 2^N Z_N^{(2)}(\beta, \Gamma) &= \tr_2 e^{-\beta Q_{ \epsilon}^{(c)}[\Gamma T +  U_\epsilon^<   + Q_{ \epsilon}^{(c)} U_\epsilon^>  Q_{ \epsilon}^{(c)}  ] Q_{ \epsilon}^{(c)}} \\
	&\leq  \tr Q_{ \epsilon}^{(c)} \ e^{-\beta [\Gamma T +  U_\epsilon^< + Q_{ \epsilon}^{(c)}  U_\epsilon^> Q_{ \epsilon}^{(c)} ]}  \leq  \tr e^{-\beta [\Gamma T +  U_\epsilon + Q_{ \epsilon}^{(c)} U_\epsilon^> Q_{ \epsilon}^{(c)} ]} .
\end{split} \]	
The last inequality follows from a trivial extension of the trace and the monotonicity of eigenvalues in the potential, $ U_\epsilon^<  \geq U_\epsilon $. 
From Lemma~\ref{lem:gsclus} we learn that $\max_{ \pmb{\sigma} \in \mathcal{L}_{2\epsilon}} \| Q_{ \epsilon}^{(c)} \delta_{\pmb{\sigma} } \|^2 \leq C N^{-1} $. Moreover, if $\pmb{\sigma} \in C_k(G)$ for some component $G$, the projection  $Q_{ \epsilon}^{(c)} \delta_{\pmb{\sigma}}$ has only support on $C_k(G)$. As any cluster $C_k(G)$ has at most $M$ configurations $\pmb{\sigma} \in \mathcal{L}_{2\epsilon}$, these observations result in the norm estimate
\[\|Q_{ \epsilon}^{(c)} U_\epsilon^> Q_{ \epsilon}^{(c)} \| \leq C M  \|U \|_{\infty} N^{-1}.\]
Since on the event considered we also have $ \| U \|_\infty \leq 2 \beta_c $ and the operator $Q_{ \epsilon}^{(c)} U_\epsilon^> Q_{ \epsilon}^{(c)}$ only acts 
non trivially  on the clusters $C_k(G)$, we thus conclude that for some $D\in (0,\infty)$:
\[ Q_{ \epsilon}^{(c)} U_\epsilon^> Q_{ \epsilon}^{(c)} \geq V \coloneqq -D \mathbbm{1}_{C} =  - D \sum_{G \in \mathcal{G}_{k,\epsilon}} \mathbbm{1}_{C_k(G)}, \quad \mathcal{C} \coloneqq \bigcup_{G \in \mathcal{G}_{k,\epsilon}}  C_k(G) \]
To summarize, we have thus shown that 
$ Z_N^{(2)}(\beta, \Gamma) \leq 2^{-N}  \tr e^{-\beta [\Gamma T +U_\epsilon + V ]} $.

From here, there are at least two possible ways to continue the proof. One could show that the potential $U_\epsilon + V$ meets the requirements of Theorem~\ref{lem:qpgs}. Then, one needs to control V, which is a little bit technical. Instead,  we will employ a convexity argument. To this end, we introduce for $\lambda \in \rr$  the family of pressures and corresponding Hamiltonians on $ \ell^2(\mathcal{Q}_N) $: 
\be\label{eq:defPhilamb} \Phi_N(\beta,\Gamma,\lambda) \coloneqq \ln  2^{-N}\tr  e^{-\beta   H(\lambda) } , \qquad   H(\lambda) \coloneqq \Gamma T + U_\epsilon+ \lambda V \ee
The pressure $\Phi_N(\beta,\Gamma,\lambda)$ is convex~\cite{Sim79} in $\lambda$, and $\lambda = 1$ is the case of interest.

Let us first discuss the case $\lambda = 0$ in which case Theorem~\ref{lem:qpgs} is applicable with $ W =U_\epsilon $. Since $\| U_\epsilon \|_\infty \leq 2 \epsilon N $ and $ \mathbb{E}\left[  U_\epsilon(\pmb{\sigma})^2\right] \leq N (1- e^{-2\epsilon^2 N} ) \leq N    $,    Theorem~\ref{lem:qpgs} guarantees that all eigenvalues of $ \Gamma T +  U_\epsilon $ below $E < -4 \epsilon N$, counted with multiplicity, are shifted with respect to the eigenvalues $ E $ of $\Gamma T$ to $E + \frac{N}{E} +o(1) $. Since $\langle \Gamma T \rangle^{\text{pm}}_\beta = - N \Gamma \tanh \beta \Gamma  \leq -8 \epsilon N $,  Proposition~\ref{lem:saddle} allows to spectrally focus the partition function onto an  interval around $ \langle \Gamma T \rangle^{\text{pm}}_\beta $ of arbitrarily small size $ 0 < \delta < \Gamma \tanh \beta \Gamma - 4 \epsilon $. A similar argument as in Step 1, then yields for all sufficiently large $ N $:
\[ \Phi_N(\beta,\Gamma,0) \leq N \ln \cosh \beta \Gamma  + \frac{\beta}{\Gamma \tanh \beta \Gamma -\delta} + o(1).  \]
Next, we consider general parameters $\lambda$. Recall that $ \mathbbm{1}_{\mathcal{C}} $ stands for the orthogonal projection onto the subspace of the union of all clusters, and that  $ \mathbbm{1}_{\mathcal{C}^c} $ is the orthogonal complement. In terms of the operator $A$ introduced in \eqref{eq:opA}, the norm estimate~\eqref{eq:normA} yields:
$ \| H(\lambda) - \mathbbm{1}_{\mathcal{C}} H(\lambda)\mathbbm{1}_{\mathcal{C}} - \mathbbm{1}_{\mathcal{C}^c} H(\lambda) \mathbbm{1}_{\mathcal{C}^c} \| \leq \| A \|  \leq C\sqrt{N} $ 
and hence
\begin{align*} \tr e^{-\beta H(\lambda)} & \leq e^{C \sqrt{N}} \tr e^{-\beta(\mathbbm{1}_{\mathcal{C}}H(\lambda)\mathbbm{1}_{\mathcal{C}} + \mathbbm{1}_{\mathcal{C}^c} H(\lambda) \mathbbm{1}_{\mathcal{C}^c})} 
	\\ & = e^{C \sqrt{N}} \left[ \tr \mathbbm{1}_{\mathcal{C}} e^{-\beta(\mathbbm{1}_{\mathcal{C}}H(\lambda)\mathbbm{1}_{\mathcal{C}})} +  \tr \mathbbm{1}_{\mathcal{C}^c} e^{-\beta(\mathbbm{1}_{\mathcal{C}^c} H(\lambda)\mathbbm{1}_{\mathcal{C}^c})}  \right]  
\end{align*}
at some $C< \infty $, which is independent of $N$ and $\lambda$. 
Each of the traces in the right side is now estimated separately:
\[ 2^{-N}  \tr \mathbbm{1}_{\mathcal{C}} e^{-\beta(\mathbbm{1}_{\mathcal{C}}H(\lambda)\mathbbm{1}_{\mathcal{C}})}  \leq e^{\beta \| \mathbbm{1}_{\mathcal{C}}H(\lambda)\mathbbm{1}_{\mathcal{C}} \|  } \leq \exp\left( \beta \left( C \Gamma \sqrt{N} + 2 \epsilon N + \lambda D \right) \right)
\]
where we used the triangle inequality for the operator norm as well as~\eqref{eq:normA} again.
Since $ H(\lambda) $ and $ H(0) $ agree on $\mathbbm{1}_{\mathcal{C}^c} \, \ell^2(\mathcal{Q}_N)$, we also have
\begin{align*} 2^{-N} \tr\mathbbm{1}_{\mathcal{C}^c} e^{-\beta \mathbbm{1}_{\mathcal{C}^c} H(\lambda)\mathbbm{1}_{\mathcal{C}^c}} & =  2^{-N} \tr \mathbbm{1}_{\mathcal{C}^c} e^{-\beta \mathbbm{1}_{\mathcal{C}^c} H(0)\mathbbm{1}_{\mathcal{C}^c}}  \\ & \leq 2^{-N} \tr \mathbbm{1}_{\mathcal{C}^c} e^{-\beta H(0) }
	\leq  e^{\Phi_N(\beta, \Gamma, 0)} . 
\end{align*}
The first inequality relied on the Jensen-Peierls estimate, which allows to pull down the projections~\cite{Ber72}.\\
Since $ \Phi_N(\beta, \Gamma, 0) > 4 \beta \epsilon N$, the correction to the pressure at  $\lambda_0 \coloneqq \frac{\epsilon N}{D} $,   is still of order $ \mathcal{O}(\sqrt{N}) $: 
\[ \Phi_N(\beta,\Gamma,\lambda_0 ) \leq \Phi_N(\beta,\Gamma,0 ) + C \sqrt{N}.  \]
We are now in the situation to exploit convexity:
\begin{align*} \Phi_N(\beta,\Gamma,   1 ) & \leq (1-\lambda_0^{-1})  \ \Phi_N(\beta,\Gamma,0) + \lambda_0^{-1} \Phi_N(\beta,\Gamma,\lambda_0 ) \\ & \leq N \ln \cosh \beta \Gamma  + \frac{\beta}{\Gamma \tanh \beta \Gamma -\delta} + o(1) .  
\end{align*}

\medskip
\noindent	
\textit{Step 3 -- Paramagnetic lower bound:}~~To show that the upper bound of Step 2 is also an asymptotic lower bound, it is more convenient to work with the full partition function $Z_N$, which by~\eqref{eq:partsplit}  is a lower bound on $Z_N^{(2)}$ up to a multiplicative error of $ e^{o(1)} $. 

For an estimate on $ Z_N $, we split the potential $ U = U_\epsilon + V_\epsilon $, where $  V_\epsilon \coloneqq U \mathbbm{1}[\vert U\vert > 2 \epsilon ] $. The pressure $ \Phi_N(\beta,\Gamma,0) $ of $ H(0) = \Gamma T + U_\epsilon $, which is defined in~\eqref{eq:defPhilamb}, was already analyzed in Step 2. 
Here, we now consider the following family of Hamiltonians $ H(\lambda) = \Gamma T + U_\epsilon +\lambda V_\epsilon $, which differs from the one in Step~2. By a slight abuse of notation, we nevertheless denote the corresponding pressure again by $ \Phi_N(\beta,\Gamma,\lambda) \coloneqq \ln \tr 2^{-N} \tr e^{-\beta H(\lambda)} $. The convexity of the pressure in $\lambda$ is again the basis for our argument.

Since on the event considered, we may assume $ \| U \| \leq 2 \beta_c N $, the potential $ W =  U_\epsilon +\lambda V_\epsilon  $ meets
the requirements of Theorem~\ref{lem:qpgs} with $ \| W \|_\infty \leq N \max\{ 2\epsilon, 2\lambda \beta_c\} $ and $ \mathbb{E}\left[  W(\pmb{\sigma})^2\right] \leq N (1- e^{-2\epsilon^2 N} ) \leq N    $.  Thus if  $\lambda < \epsilon/\beta_c$ the eigenvalues of $ H(\lambda) $  below $E < -4 \epsilon N $, counted with multiplicity, are shifted with respect to the eigenvalues $ E $ of $\Gamma T$ to $E + \frac{N}{E} +o(1) $ and we have $  \Phi_N(\beta,\Gamma,\lambda)  =  \Phi_N(\beta,\Gamma,0) +o(1) $.  
Fixing $\lambda_0 \coloneqq \frac{\epsilon}{2 \beta_c} < 1 $, convexity implies:
\begin{align*}  \Phi_N(\beta,\Gamma,1)& \geq 
	\Phi_N(\beta,\Gamma,0) + \frac{1}{\lambda_0} (\Phi_N(\beta,\Gamma,\lambda_0) - \Phi_N(\beta,\Gamma,0 )) \\
	& =  N \ln \cosh(\beta \Gamma)  +  \frac{\beta}{\Gamma \tanh \beta \Gamma  +\delta } + o(1) .
\end{align*}
The last step, which holds for all $ \delta > 0 $ sufficiently small, is the desired lower bound, again relied on an explicit estimate based on the concentration of the partition function of $ \Gamma T $  around energies near $ -\Gamma \tanh \beta \Gamma $, cf.~Proposition~\ref{lem:saddle}.\\[.5ex]

\noindent
\textit{Step 4 -- Completing the proof.}~~Away from the first-order phase transition at $ \Gamma = \Gamma_c(\beta) $ described in Proposition~\ref{thm:old} and on the event on which Step~1-3 are valid, the partition function~\eqref{eq:partsplit}  is either dominated by the REM-term $ Z_N^{(1)} $ in case $ \Gamma < \Gamma_c(\beta)  $, or by the paramagnetic term $ Z_N^{(2)} $ in case $ \Gamma > \Gamma_c(\beta)  $. 

More precisely, in case $ \Gamma < \Gamma_c(\beta)  $ and since the probability of the  event, which is excluded in Step~1--3,  is exponentially small in $ N $ and hence summable, we conclude for any $ \varepsilon > 0 $:
\be\label{eq:fastconv}
\sum_{N\geq 1 } \pp\left( \left| \Phi_N(\beta, \Gamma)  -  \Phi_N(\beta,0) + \frac{\beta \Gamma^2}{u(\beta)}  \right\vert > \varepsilon \right) < \infty .
\ee
The claimed almost-sure convergence then follows by a Borel-Cantelli argument. 
The analogous argument establishes the claim in case $ \Gamma > \Gamma_c(\beta)  $.
\end{proof}

\appendix
	\section{Proof of Proposition~\ref{prop:trans} and \ref{prop:gap}}\label{appendix}
	
	In this section, we sketch the modifications necessary for covering the critical case.
	For any $\epsilon > 0$ we introduce the symmetrized large-deviation set
	\[ \mathcal{S}_\epsilon \coloneqq \{ \pmb{\sigma} \in \mathcal{Q}_N \, \vert \, \vert U(\pmb{\sigma})\vert \geq \epsilon N  \} . \]
	For parameters $ \epsilon ,\alpha, \delta > 0$, we use the notion of a \textit{symmetrized global $(  \epsilon ,\alpha, \delta ) $-deep-hole scenario} if  $\mathcal{L}_{\beta_c-\delta}$ is replaced by $\mathcal{S}_{\beta_c-\delta}$ in Definition~\ref{def:deephole}. 
	As is evident from the proof of Lemma~\ref{lem:deeph}, 
	due to the symmetry of the Gaussian distribution, if $ \epsilon ,\alpha, \delta > 0$ satisfy \eqref{eq:cond}  then the event
	\[ \Omega_N^s(  \epsilon ,\alpha, \delta ) := \left\{ \text{$ U $ satisfies a symmetrized global $(  \epsilon ,\alpha, \delta ) $-deep-hole scenario} \right\} 
	\]
	still has a probability bounded by $ 1 - e^{-c(  \epsilon ,\alpha, \delta )  N} $ for all sufficiently large $ N $. 
	We will again pick $ \epsilon = \beta_c/2 $ and the values of  $\alpha, \delta$ will be given (implicitly) as we go along the proof. 
	Similarly as in the proof of Theorem~\ref{thm:sggs}, under this scenario, we work with the auxiliary Hamiltonian 
	\begin{equation}\label{eq:haux} H^\prime = \left( \bigoplus_{\pmb{\sigma} \in \mathcal{S}_{\beta_c-\delta} } H_{\alpha N}(\pmb{\sigma}) \right) \bigoplus H_r ,  \end{equation}
	which is well-defined, since the balls $B_{\alpha N}(\pmb{\sigma})$ are disjoint for two different configurations in $\mathcal{S}_{\beta_c-\delta}$. 
	That we work with $\mathcal{S}_{\beta_c-\delta}$ instead of $\mathcal{L}_{\beta_c-\delta}$ is due to purely technical reasons: we want to ensure that the restriction of $U$ to the complement of the balls $H_{\alpha N}(\pmb{\sigma})$  is still a symmetric random variable so that we can apply Theorem~\ref{lem:qpgs}.

	The first main step in the proof of Proposition~\ref{prop:trans} and~\ref{prop:gap} is to show that for $\Gamma \simeq \beta_c$ there are $c, r > 0$ (which may still depend on $ \alpha ,\delta $) such that
	\[ \| (H- H^\prime) \mathbbm{1}_{(-\infty,(\Gamma-r)N}( H^\prime) \| \leq e^{-cN} . \]
	However, the paramagnetic part causes a challenge, since the less energy levels we cut out (i.e. choosing a smaller parameter $\delta > 0$), the low energy states of $H_r$ become less delocalized. Indeed our bound from Proposition~\ref{prop:ext} becomes trivial for $ \Gamma = \beta_c $ and $ r \leq \delta $.  To overcome this problem, we use the improved delocalization estimate from Proposition~\ref{prop:deloc}, which deals with  $  H_{\delta} =  \Gamma T + U \mathbbm{1}_{\vert U \vert \leq (\beta_c - \delta)N} $ on $ \ell^2(\mathcal{Q}_N) $. We first show that for $\delta > 0$ small enough the low energy states of $H_r$ and $H_{\delta}$ agree up to an experientially small error. Here and in the following we always assume that the parameter $\delta >0$ coincides in the definition of $H_r$ and $H_{\delta}$. Then, we show that the low eigenvalues of $H$ and $H^\prime$ again coincide besides an exponential error. From there, we establish Proposition~\ref{prop:trans} and~\ref{prop:gap}.\\
	
	\noindent	
	\textit{Step 1: Comparison of the low energy spaces of $H_r$ and $H_{\delta}$: } 
	We abbreviate 
	\[ 
	B_{\cup}\coloneqq \bigcup_{\pmb{\sigma} \in \mathcal{S}_{\beta_c-\delta} } B_{\alpha N}(\pmb{\sigma}) , \qquad B_{\cup}^{+} \coloneqq \bigcup_{\pmb{\sigma} \in \mathcal{S}_{\beta_c-\delta} } B_{\alpha N+1}(\pmb{\sigma})  . 
	\]
	and recall that $H_r$ is the restriction of $H$ onto the complement of $ B_{\cup} $. In the following, $\hat{H}_r$ stands for its canonical extension to $\ell^2(\mathcal{Q}_N)$. 
	Then, the spectra of $ H_r $ and $ \hat H_r $ agree outside zero. 
	Proposiotion~\ref{prop:deloc} guarantees the existence of  some  $\eta = \eta(\delta)  > 0$ and some  
	$ c > 0 $, which is independent of $ \delta $, such that for all $ N $ large enough:
	\begin{enumerate}[leftmargin=*]
		\item the spectral projection $ P_{\delta}(\eta) \coloneqq  \mathbbm{1}_{(-\infty, -(\Gamma - \eta) N)}(H_{\delta}) $ is finite-dimensional, and the eigenvalue below $ -(\Gamma - \eta) N $ are $ (2n-N) \Gamma + \frac{N}{(2n-N)\Gamma} + \mathcal{O}_\Gamma(N^{-1/4}) $ with $ 2n \leq \eta N/\Gamma $. 
		\item $  \max_{\pmb{\sigma} \in \mathcal{Q}_N} \langle \delta_{\pmb{\sigma}} \,\vert \, P_{\delta}(\eta)\, \delta_{\pmb{\sigma}} \rangle \leq e^{-2cN} $.
	\end{enumerate}
	Since $ \hat{H}_r$ and $H_{\delta}$ only differ on $B_{\cup}^{+} $, the triangle inequality thus yields
	\begin{align*} \|   (H_{\delta} - \hat{H}_r)  P_{\delta}(\eta) \|   & \leq \sum_{\pmb{\sigma} \in B_{\cup}^{+} } \|   (H_{\delta} - \hat{H}_r) \delta_{\pmb{\sigma}} \|   \ \|   P_{\delta}(\eta) \delta_{\pmb{\sigma}} \| \\ &    \leq  (\beta_c + \Gamma) N \ \vert  B_{\cup}^{+}  \vert   \max_{\pmb{\sigma} \in \mathcal{Q}_N}  \sqrt{\langle \delta_{\pmb{\sigma}} \,\vert  \, P_{\delta}(\eta)\, \delta_{\pmb{\sigma}} \rangle } \\
		& \leq \vert  \mathcal{S}_{\beta_c-\delta} \vert  \ e^{N \gamma(\alpha) + o(N) } e^{-cN}  . 
	\end{align*}
	
	Consulting Lemma~\ref{lem:card},  since the set $A_3(\delta)$ agrees with $\mathcal{S}_{\beta_c-\delta}$, the cardinality is bounded $\vert \mathcal{S}_{\beta_c-\delta} \vert \leq 2 e^{\beta_c \delta N}$ except for an exponentially small event. We now use that $c$ is independent of $\delta$. We thus find some (possibly) smaller parameters $\alpha,\delta> 0$ such that $\beta_c \delta + \gamma(\alpha) < c/2$. It then follows that for some $\eta = \eta(\delta)> 0$ and all $ N $ large enough:
	\[ \| (H_{\delta} - \hat{H}_r) P_{\delta}(\eta)  \| \leq e^{-cN}. \]
	In particular, the 'off-diagonal' blocks $P_{\delta}(\eta) \hat{H}_r Q_{\delta}(\eta)$ and $Q_{\delta}(\eta) \hat{H}_r P_{\delta}(\eta)$ with $ Q_{\delta}(\eta)  = 1 -  P_{\delta}(\eta) $ are bounded in norm by $e^{-cN}$. The spectrum of $\hat H_r$ hence agrees up to the error $e^{-cN}$ with those of the block matrices $P_{\delta}(\eta) {H}_\delta $ and $Q_{\delta}(\eta) \hat{H}_r Q_{\delta}(\eta)$. In order to see that the low-energy spectrum is entirely determined by the first block,  we need to derive a lower
	bound on ground state energy of the second block. We will prove below that for sufficiently large $ N $:
	\be\label{eq:lowerFS} \inf\spec Q_{\delta}(\eta) \hat{H}_r Q_{\delta}(\eta) \geq  -(\Gamma - \eta/2 ) N .
	\ee
	Perturbation theory based on the Feshbach-Schur method and the fact that the spectrum of $ H_\delta $ is discrete with eigenvalue clusters which are separated from each other by gaps of at least $ 1 $, then yields that $ \|\mathbbm{1}_{(-\infty, -(\Gamma - \eta_0) N)}(\hat H_{r}) - P_{\delta}(\eta_0) \| \leq C e^{-c N} $ for some $ \eta_0 = \eta_0(\delta) > 0 $ and all sufficiently large $ N $,
	and hence
	\[
	\max_{\pmb{\sigma} \in (B_\cup)^c} \langle \delta_{\pmb{\sigma}}  \,\vert \, \mathbbm{1}_{(-\infty, -(\Gamma - \eta_0) N)}(H_{r}), \delta_{\pmb{\sigma}} \rangle \leq   \max_{\pmb{\sigma} \in \mathcal{Q}_N} \langle \delta_{\pmb{\sigma}} \,\vert \, P_{\delta}(\eta_0)\, \delta_{\pmb{\sigma}} \rangle + C e^{-c N} \leq (1+C) \,  e^{-c N} . 
	\]
	It thus remains to proof~\eqref{eq:lowerFS}. 
	We will employ a counting argument based on the min-max principle. It rests on the following observations. For any $ E < 0 $ the eigenvalues of $ \hat{H}_r $ and $ H_r $ below $ E $ agree. The same applies to $ \hat{H}_r $ and the energy form $ q_{r}^\infty: \ell^2(\mathcal{Q}_N) \to \rr\cup\infty $:
	\[   q_{r}^\infty(\psi) \coloneqq \begin{cases}
		\langle \psi, \, H_{r} \psi \rangle, &\text{ if } \psi \in \ell^2((B_{\cup})^{c} ), \\ + \infty &\text{ otherwise }.
	\end{cases} \]
	Denoting by $  \hat{H}_r^\infty $ the corresponding operator, which results in driving the potential values 
	on $ B_{\cup}) $ to infinity   (see e.g.~\cite{Sim15}), we thus have $ H_\delta \leq   \hat{H}_r^\infty  $ and hence
	$
	\tr \mathbbm{1}_{(-\infty, E)}(\hat H_{r}) = \tr \mathbbm{1}_{(-\infty, E)}(\hat H_{r}^\infty ) \leq  \tr \mathbbm{1}_{(-\infty, E)}(H_\delta ) $. 
	However, since we have shown above that $ \hat H_{r} $ is $ e^{-cN } $-norm close to $ P_{\delta}(\eta) {H}_\delta + Q_{\delta}(\eta) \hat{H}_r Q_{\delta}(\eta) $ and all eigenvalues of $  P_{\delta}(\eta) {H}_\delta  $ below $ -(\Gamma -\eta) N $ are already accounted for at this energy, we conclude the validity of \eqref{eq:lowerFS} (even with $ \eta/2 $ replaced by $ \eta - e^{-cN }  $). \\

		\noindent
		\textit{Step 2: Comparison of the low energy spectra of $H$ and $H^\prime$: } We want to show that the low energies of $H$ and $H^\prime$ coincide up to an exponentially small error.
		We abbreviate $ P_{\delta}^\prime (\eta) \coloneqq  \mathbbm{1}_{(-\infty, -(\Gamma - \eta) N)}(H^\prime) $ and $Q_{\delta}^\prime (\eta) = \mathbbm{1} -P_{\delta}^\prime (r)$. 
		As in the first step,  by the Feshbach-Schur method it suffices to show
		\[ \| (H^\prime - H) P_{\delta}^\prime (\eta) \| \leq e^{-cN}, \quad Q_{\delta}^\prime (\eta) H Q_{\delta}^\prime (\eta) \geq -(\Gamma - \eta/2) N Q_{\delta}^\prime (\eta) \] 
		for some $\eta , c >0$.  
		To this end, let $\eta < \eta_0$ with $\eta_0$ from the first step. Due to the direct-sum structure of $H^\prime$ we have
		\[ \begin{split} \| (H^\prime - H)   P_{\delta}^\prime (\eta) \| &\leq  \| (H^\prime - H)\mathbbm{1}_{(B_\cup)^c}   \mathbbm{1}_{(-\infty, -(\Gamma - \eta) N)}(H_r)  \| \\&+  \Big\| (H^\prime - H) \mathbbm{1}_{B_\cup}   \mathbbm{1}_{(-\infty, -(\Gamma - \eta) N)}\Big(\Big( \bigoplus_{\pmb{\sigma} \in \mathcal{L}_{\beta_c-\delta} } H_{\alpha N}(\pmb{\sigma}) \Big)\Big)  \Big\| \end{split}  \]
		Note that we have replaced $\mathcal{S}_{\beta_c-\delta}$ by $\mathcal{L}_{\beta_c-\delta}$ in the above direct sum, as only those balls contain low energy states.
		That the first summand is exponentially small for $t$ small enough follows from the same argument we used to compare the low energy projections of $H_r$ and $H_{\delta}$ in the first step. The second summand is exponentially small, too. We simply recall the exponential decay estimate of the localized wavefunctions from Theorem~\ref{thm:sggs}.
		Moreover, we have 
		\begin{align*} Q_{\delta}^\prime (\eta) H Q_{\delta}^\prime (\eta) & \geq -(\Gamma - t) N Q_{\delta}^\prime (\eta) + Q_{\delta}^\prime (\eta) (H-H^\prime) Q_{\delta}^\prime(\eta) \\ & \geq - (\Gamma - \eta + 2\sqrt{\alpha(1-\alpha)} + o(N)) Q_{\delta}^\prime (\eta),  
		\end{align*}
		where we used that $H - H^\prime$ is the direct sum of the hopping operators on the boundaries $\partial B_{\alpha N}(\pmb{\sigma})$ and then used the norm estimate from Lemma~\ref{lem:tknorm}. If we choose $\alpha$ small enough, we arrive at $Q_{\delta}^\prime (\eta) H Q_{\delta}^\prime (\eta) \geq -(\Gamma - t/2) N Q_{\delta}^\prime (\eta)$. 
		We conclude that the energy levels $E < -(\Gamma - \eta/4) N$ of $H$ and $H^\prime$ agree up to an at most exponentially small shift $e^{-cN}$ for $t$ small enough.
		
		\medskip
		\noindent
		\textit{Step 3: Proof of Proposition~\ref{prop:trans}:}
		To conclude the proof of Proposition~\ref{prop:trans}, it remains to evaluate for $\Gamma = \beta_c$ the probability
		$ \pp( E_{\pmb{\sigma}_0} < \inf\spec H_r) $,
		where $E_{\pmb{\sigma}_0}$ denotes the minimal ball energy. As we are only interested in the order of this probability, we simplify the calculations and write $ x \simeq y$ if $x/y = \mathcal{O}(1)$. Then, 
		\[ \begin{split}
			\pp( E_{\pmb{\sigma}_0} <  \inf\spec H_r ) & \simeq \pp(\min U  < - \beta_c N ) \simeq \pp(\min U< s_N(\ln(N)/2) ) \\
			&= 1 - (1-2^{-N}e^{-\ln(N)/2})^{2^N} \simeq 1 - e^{-1/\sqrt{N}} \simeq \frac1{\sqrt{N}}.
		\end{split}  \]
		For the first line we recall that $E_{\pmb{\sigma}_0} = \min U  + \mathcal{O}(1)$ and $\inf\spec H_r = - \Gamma N + \mathcal{O}(1)$. Moreover, we recall the definition of  $s_N$ from \eqref{eq:rescaling}  and its connection to the extreme value process of $U$ via \eqref{eq:Uinfty}. The remaining expressions can derived by elementary calculus. This is enough to conclude the assertions on the ground state energy.
		
		It remains to show that aside from an exponentially small event, the ground state wave function is localized if $E_{\pmb{\sigma}_0} <  \inf\spec H_r$ and otherwise delocalized. To this end, we observe  case $ \inf\spec H_r$ and $E_{\pmb{\sigma}_0}$ are independent random variables. Thus, the spectral averaging method from Lemma~\ref{lem:specav} shows  that for any $c > 0$ there exists some $b(c) > 0$ such that
		\[ \vert \pp( \min_{\pmb{\sigma}} E_{\pmb{\sigma}} -  \inf\spec H_r)\vert \leq e^{-cN} ) \leq e^{-b(c)N} \]
		for $N$ large enough. The Feshbach-Schur formula then yields that except for an exponentially small event the ground state function $\psi$ of $H$ satisfies
		$ \|\psi - \psi_{\pmb{\sigma}_0} \|   \leq e^{-cN}$ if $E_{\pmb{\sigma}_0} <  \inf\spec H_r$, and otherwise 	$ \|\psi - \psi_{r} \|   \leq e^{-cN}$
		with the ground state wave functions $\psi_{\pmb{\sigma}_0}$ of $H_{\alpha N}(\pmb{\sigma}_0)$ and $\psi_r$ of $H_r$. From the previous steps, it follows that  $\psi_{\pmb{\sigma}_0}$ satisfies the assertions of Theorem~\ref{thm:sggs} and that $\psi_r$ is exponentially delocalized and satisfies the first assertion of Theorem~\ref{thm:sggs}.
		
		\medskip
		\noindent
		\textit{ Step 4: 	Proof of Proposition~\ref{prop:gap}}. Since 
		\begin{align*} \min_{\pmb{\sigma} \in \mathcal{L}_{\beta_c - \delta}} E_{\pmb{\sigma}}(\Gamma = 0) &<   \inf\spec H_r( \Gamma = 0), \\ \min_{\pmb{\sigma} \in \mathcal{L}_{\beta_c - \delta}} E_{\pmb{\sigma}}(\Gamma = 2 \beta_c ) &>   \inf\spec H_r(\Gamma = 2 \beta_c ),   
		\end{align*}
		by continuity of the eigenvalues of the involved matrices there exists some $\Gamma_N$ such that
		\[  \min_{\pmb{\sigma} \in \mathcal{L}_{\beta_c - \delta}} E_{\pmb{\sigma}}(\Gamma_N) =   \inf\spec H_r(\Gamma_N). \]
		For this $\Gamma_N$ the functions $\min_{\pmb{\sigma} \in \mathcal{L}_{\beta_c - \delta}} E_{\pmb{\sigma}}(\Gamma_N)$ and $ \inf\spec H_r( \Gamma_N)$ in particular agree up to order one, i.e.,~\eqref{eq:gammagap} holds.
		Since $H$ and $H^\prime$ agree at these low energies up to an exponentially small shift, we conclude that $ \Delta_N(\Gamma_N) $ is exponentially small. 
		The minimum is attained at a possibly different $\Gamma_N^{\star}$. However, since still 
		$ \min_{\pmb{\sigma} \in \mathcal{L}_{\beta_c - \delta}} E_{\pmb{\sigma}}(\Gamma_N^\star) - \inf\spec H_r( \Gamma_N^{\star}) = o(1) $, 
		we have $ \Gamma_N^{\star} - \Gamma_N = o(N^{-1}) $ as claimed. \qed

	\minisec{Acknowledgments}
	This work was supported by the DFG under EXC-2111 -- 390814868.

	\bigskip
	\bigskip
	\begin{minipage}{0.5\linewidth}
		\noindent Chokri Manai and Simone Warzel\\
		MCQST and Department of Mathematics\\
		Technische Universit\"{a}t M\"{u}nchen
	\end{minipage}

\end{document}